\documentclass{amsart}
\usepackage[foot]{amsaddr}

\RequirePackage{amsthm,amsmath,amsfonts,amssymb}
\RequirePackage[authoryear]{natbib}
\RequirePackage[colorlinks,citecolor=blue,urlcolor=blue]{hyperref}
 \usepackage[margin=3.5cm]{geometry} 
\usepackage{algorithm, algorithmic,enumerate,booktabs}
\usepackage{tikz}
\usepackage{caption} 
\usepackage{subcaption}
\usetikzlibrary{positioning,calc}
\usepackage{enumitem}
\usepackage{dsfont}

\usepackage{crossreftools}
\makeatletter
\newcommand{\optionaldesc}[2]{%
  \phantomsection
  #1\protected@edef\@currentlabel{#1}\label{#2}%
}
\makeatother

\theoremstyle{plain}
\newtheorem{thm}{Theorem}[section]
\newtheorem{lem}[thm]{Lemma}
\newtheorem{prop}[thm]{Proposition}
\newtheorem{cor}[thm]{Corollary}

\theoremstyle{remark}
\newtheorem{defn}[thm]{Definition}
\newtheorem{exmp}[thm]{Example}
\newtheorem{rem}[thm]{Remark}

\numberwithin{equation}{section}

\makeatletter
\@namedef{subjclassname@2020}{%
  \textup{2020} Mathematics Subject Classification}
\makeatother

\newcommand*{\defeq}{\mathrel{\vcenter{\baselineskip0.5ex \lineskiplimit0pt\hbox{\scriptsize.}\hbox{\scriptsize.}}}=}
\newcommand{\var}{\text{Var}}

\newcommand{\Cov}{\text{Cov}}

\newcommand{\ubar}[1]{\text{\b{$#1$}}}
\newcommand\smallO{
  \mathchoice
    {{\scriptstyle\mathcal{O}}}
    {{\scriptstyle\mathcal{O}}}
    {{\scriptscriptstyle\mathcal{O}}}
    {\scalebox{.7}{$\scriptscriptstyle\mathcal{O}$}}
}
\newcommand{\rank}{\text{rank}}

\allowdisplaybreaks

\title[Testing Many Constraints in Possibly Irregular Models]{Testing Many Constraints in Possibly Irregular Models Using Incomplete U-Statistics}

\author[N.~Sturma]{Nils Sturma}
\email{nils.sturma@tum.de}

\author[M.~Drton]{Mathias Drton} 
\address[N.~Sturma, M.~Drton]{
Munich Center for Machine Learning and Department of Mathematics,
School of Computation, Information and Technology, Technical University of Munich
}
\email{mathias.drton@tum.de}

\author[D.~Leung]{Dennis Leung}
\address[D.~Leung]{University of Melbourne, Australia; School of Mathematics and Statistics}
\email{dennis.leung@unimelb.edu.au}

\begin{document}
\sloppy

\begin{abstract}
We consider the problem of testing a null hypothesis defined by equality and inequality constraints on a statistical parameter. Testing such hypotheses can be challenging because the number of relevant constraints may be on the same order or even larger than the number of observed samples. Moreover, standard distributional approximations may be invalid due to irregularities in the null hypothesis. We propose a general testing methodology that aims to circumvent these difficulties. The constraints are estimated by incomplete U-statistics, and we derive critical values by Gaussian multiplier bootstrap. We show that the bootstrap approximation of incomplete U-statistics is valid for kernels that we call mixed degenerate when the number of combinations used to compute the incomplete U-statistic is of the same order as the sample size. It follows that our test controls type I error even in irregular settings. Furthermore, the bootstrap approximation covers high-dimensional settings making our testing strategy applicable for problems with many constraints. The methodology is applicable, in particular, when the constraints to be tested are polynomials in U-estimable parameters.  As an application, we consider goodness-of-fit tests of latent tree models for multivariate data.
\end{abstract}



\maketitle

\section{Introduction} \label{sec:introduction}

Let $\{P_{\theta}: \theta \in \Theta\}$ be a statistical model with parameter space $\Theta \subseteq \mathbb{R}^d$. Given i.i.d.~samples $X_1, \ldots, X_n$ from an unknown distribution $P_{\theta}$,  we are interested in testing
\begin{equation} \label{eq:testing-problem}
    H_0 : \theta \in \Theta_0 \quad \textrm{ vs. } \quad H_1 : \theta \in \Theta \setminus \Theta_0
\end{equation}
for a subset $\Theta_0 \subseteq \Theta$. In this paper we consider the situation where the null hypothesis $\Theta_0$ is defined by constraints, that is,
\begin{equation} \label{eq:nullhypothesis}
\Theta_0 = \{\theta \in \Theta: \ f_j(\theta) \leq 0 \textrm{ for all } j=1,\ldots,p\},
\end{equation}
where each constraint $f_j$ is a function $f_j: \mathbb{R}^d \rightarrow \mathbb{R}$. The description of the null hypothesis also allows for equality constraints since $f_j(\theta)=0$ can be equivalently described by $f_j(\theta) \leq 0$ and $-f_j(\theta)\leq 0$. These types of hypotheses are an important and general class and appear in a variety of statistical problems. Our work is in particular motivated by polynomial hypotheses where each $f_j$ belongs to the ring $\mathbb{R}[\theta_1, \ldots, \theta_d]$ of polynomials in the indeterminates $\theta_1, \ldots, \theta_d$ with real coefficients. Examples of polynomial hypotheses feature in graphical modeling \citep{sullivant2010trek, Chen2014testable, shiers2016thecorrelation, chen2017identification}, testing causal effects \citep{spirtes2000causation, steyer2005analyzing, pearl2009causality, strieder2021confidence},  testing  subdeterminants, tetrads, pentads and more in factor analysis models \citep{bollen2000atetrad, gaffke2002on, silva2006learning, drton2007algebraic,drtonmassamolkin2008, dufour2013wald, drton2016wald, leung2018algebraic}, and in constraint-based causal discovery algorithms \citep{pearl1995atheory, spirtes2000causation, claassen2012abayesian}.

The standard method for dealing with testing problems like~\eqref{eq:testing-problem} is the likelihood ratio test. However, a likelihood function may be multimodal and difficult to maximize. If the likelihood ratio test is not suitable, then one might instead make use of the implicit characterization of $\Theta_0$ given by~\eqref{eq:nullhypothesis}. In Wald-type tests, for example, the strategy is to form estimates of the involved functions $f_1, \ldots, f_p$ and aggregate them in a test statistic.  Standard Wald tests require the number of restrictions $p$ to be smaller than or equal to the dimension $d$. However, in many of the above examples this might not be the case.

Another challenge for the classical likelihood ratio and Wald test is that the null hypothesis $\Theta_0$ may contain irregular points. For example, when the hypothesis is polynomial, it may contain singularities; a rigorous definition of singularities can be found in \citet[Section 4.1]{drton2009likelihood} or \citet[Section 9]{cox2015ideals}. 
At singularities the rank of the Jacobian of the constraints being tested drops and the asymptotic behaviour of the likelihood ratio test and the Wald-type test can be different than at regular points, resulting in an (also asymptotically) invalid test \citep{gaffke1999on, gaffke2002on, drton2009likelihood, dufour2013wald, drton2016wald}. In practice, it is unknown whether the true parameter $\theta$ is an irregular point, and it is therefore desirable to construct a test statistic for which one can give asymptotic approximations that accommodate and remain valid in irregular settings. 

In this work we propose a testing strategy that aims to cover setups where the number of restrictions $p$ can be much larger than the sample size $n$ and where the true parameter may be an irregular point. A precise definition of regular and irregular points is given later. Our method incorporates estimating the constraints $f_1, \ldots, f_p$ by an \emph{incomplete} $U$-statistic. By first considering the commonly used \emph{complete} $U$-statistic, we now give intuition for how this allows for high dimensionality and irregular points. Complete $U$-statistics provide an efficient method for unbiased estimation of $f \defeq (f_1, \ldots, f_p)$. We assume that $f(\theta)$ is ($U$-)estimable, i.e., for some integer $m$ there exists a $\mathbb{R}^p$-valued measurable symmetric function $h(x_1, \ldots, x_m)$ such that 
\[
\mathbb{E}_{\theta}[h(X_1, \ldots, X_m)] = f(\theta)\quad \textrm{ for all } \theta \in \Theta,
\]
when $X_1, \ldots, X_m$ are i.i.d.~with distribution $P_{\theta}$. The $U$-statistic with kernel $h$ is the average of $h(X_{i_1}, \ldots, X_{i_m})$ over all distinct $m$-tuples $(i_1, \ldots, i_m)$ from $\{1, \ldots, n\}$, in formulas
\begin{equation} \label{eq:complete-Ustat}
    U_n = \frac{1}{|I_{n,m}|} \sum_{(i_1, \ldots, i_m) \in I_{n,m}}  h(X_{i_1}, \ldots, X_{i_m}),
\end{equation}
where $I_{n,m} = \{(i_1,\ldots, i_m): 1 \leq i_1 < \ldots < i_m \leq n\}$. Due to the form of the null hypothesis $\Theta_0$, it is natural to define the test statistic as the maximum of the studentized $U$-statistic, i.e.,
\begin{equation} \label{eq:teststat-complete}
    \max_{1 \leq j \leq p} \sqrt{n} \, U_{n,j} / \hat{\sigma}_{j}
\end{equation}
and reject $H_0$ for ``large'' values of it. Here, $U_{n,j}$ refers to the $j$-th coordinate of $U_n$ for all $j=1, \ldots, p$ and $\hat{\sigma}_{j}^{2}$ is a ``good'' estimator of the asymptotic variance of $U_{n,j}$.

\begin{exmp} \label{exmp:tetrad-kernel}
As a leading example we consider testing of so-called ``tetrad constraints'', a problem of particular relevance in factor analysis \citep{ bollen2000atetrad, spirtes2000causation, hipp2003model, drton2016wald, leung2016identifiability} that can be traced back to \citet{spearman1904general} and \citet{wishart1928sampling}. For a given symmetric matrix $\Sigma=(\sigma_{uv})$, a tetrad is an off-diagonal $2 \times 2$ sub-determinant. An example is $f_j(\Sigma) = \sigma_{uv} \sigma_{wz} - \sigma_{uz} \sigma_{vw}$ with four different indices $u,v,w,z$. If an $l$-dimensional normal distribution $N_l(0, \Sigma)$ follows a one-factor analysis model, where it is assumed that all variables are independent conditioned on one hidden factor, then all tetrads that can be formed from the covariance matrix $\Sigma$ vanish. Thus, for given i.i.d.~samples  $X_1, \ldots, X_n \sim N_l(0, \Sigma)$, one might be interested in testing whether all tetrads vanish simultaneously. It is easy to see that a tetrad $f_j(\Sigma) = \sigma_{uv} \sigma_{wz} - \sigma_{uz} \sigma_{vw}$ is estimable by the kernel  \looseness=-1
\begin{align*}
    h_j(X_1,X_2)= \frac{1}{2}\{&(X_{1u}X_{1v}X_{2w}X_{2z}-X_{1u}X_{1z}X_{2v}X_{2w}) + (X_{2u}X_{2v}X_{1w}X_{1z}-X_{2u}X_{2z}X_{1v}X_{1w})\}.
\end{align*}
Interestingly, the resulting $U$-statistic corresponds to the ``plug-in'' estimate, i.e., $U_{n,j} = \frac{n}{n-1}f_j(S)$  where $S=\frac{1}{n}\sum_{i=1}^n X_i X_i^{\top}$ is the sample covariance. The plug-in estimate is considered in previous work on testing tetrads; see for example \citet{shiers2016thecorrelation}. \looseness=-1
\end{exmp}

Critical values for the test statistic~\eqref{eq:teststat-complete} can be derived by bootstrap methods that approximate the sampling distribution of $U_n$. Crucial to the validity of the bootstrap is the approximation by a Gaussian distribution. Recent progress in high-dimensional central limit theory yields valid Gaussian approximation of $U$-statistics in settings where $p \gg n$ is allowed. In particular, \citet{chen2018gaussian} and \citet{chen2020jackknife} derive finite-sample Berry-Esseen type bounds on the Gaussian approximation in a two-step procedure: In the first step, the centered $U$-statistic is approximated by the linear component in the Hoeffding decomposition (a.k.a. the H\'ajek projection) 
\begin{equation} \label{eq:hajek-proj}
    \frac{m}{n}\sum_{i=1}^n (g(X_i) - f(\theta)),
\end{equation}
where $g(x) = \mathbb{E}_{\theta}[h(x,X_2,\ldots, X_m)]$ and the expectation is taken with respect to the unknown distribution $P_{\theta}$. In the second step, the linear term~\eqref{eq:hajek-proj} is further approximated by a Gaussian random vector using the ``classical'' central limit theorem for high-dimensional independent sums  \citep{chernozhukov2013gaussian, chernozhukov2017central}. Crucially, this procedure assumes that the $U$-statistic is \emph{non-degenerate}, that is, the individual variances $\sigma_{g,\theta, j}^2 = \var_{\theta}[g_j(X_1)]$ of the linear component do not vanish. 
Hence, the Berry-Esseen type bound on the approximation relies on the standard assumption that the minimum $\ubar{\sigma}_{g, \theta} = \min_{1 \leq j \leq p} \sigma_{g,\theta,j}^2$ is bounded away from zero. However, the minimum  $\ubar{\sigma}_{g, \theta}$  depends on the unknown parameter $\theta$, and there may be certain points in $\Theta_0$ where we have $\ubar{\sigma}_{g,\theta}=0$. 
\begin{defn} \label{def:regular-irregular}
We say that a point $\theta \in \Theta$ is \emph{regular} with respect to the kernel $h$ if $\ubar{\sigma}_{g,\theta} = \min_{1 \leq j \leq p} \sigma_{g,\theta,j}^2 > 0$. Otherwise, we say that $\theta$ is an \emph{irregular} point. 
\end{defn}
If $\theta$ is an irregular point, then the Gaussian approximation of the $U$-statistic is not valid any more. This is illustrated in Example~\ref{ex:tetrad-degeneracy}. Even if the parameter $\theta$ is only ``close" to an irregular point, the minimum $\ubar{\sigma}_{g,\theta}$ can be very small. In this case, a very large sample size may be required for the Gaussian limiting distribution to provide a good approximation of the $U$-statistic since convergence is not uniform and the rate depends on the minimum $\ubar{\sigma}_{g,\theta}$. \looseness=-1

\begin{exmp} \label{ex:tetrad-degeneracy}
Recall the kernel $h_j$ of the tetrad $f_j(\Sigma) = \sigma_{uv} \sigma_{wz} - \sigma_{uz} \sigma_{vw}$ from Example~\ref{exmp:tetrad-kernel}. The corresponding random variable $g_j(X_1)$ in the H\'ajek projection~\eqref{eq:hajek-proj} is given by \looseness=-1
\begin{equation} \label{eq:tetrad-hajek-proj}
    g_j(X_1) =  \frac{1}{2}\left\{(X_{1u}X_{1v}\sigma_{wz}-X_{1u}X_{1z}\sigma_{vw}) + (\sigma_{uv}X_{1w}X_{1z}-\sigma_{uz}X_{1v}X_{1w})\right\}.
\end{equation}
By inspecting~\eqref{eq:tetrad-hajek-proj} we see that $g_j(X_1)$ is degenerate if $\sigma_{wz}= \sigma_{vw} = \sigma_{uv} = \sigma_{uz} = 0$ and, in general, non-degenerate if at least one of the covariances is non-zero. Thus, the covariance matrices $\Sigma=(\sigma_{uv})$ in the one-factor analysis model that have $\sigma_{wz}= \sigma_{vw} = \sigma_{uv} = \sigma_{uz} = 0$ correspond to irregular points. Moreover, when the covariance matrix is only close to irregular, then the variance $\sigma_{g,\theta,j}^2$ might already be very small. \looseness=-1
\end{exmp}

To accommodate irregularity, we propose to use \emph{randomized incomplete $U$-statistics} instead of the usual, complete $U$-statistic from~\eqref{eq:complete-Ustat}.  That is, for a \emph{computational budget parameter} $N \leq |I_{n,m}|$ we randomly choose on average $N$ indices from  $I_{n,m}$. Then the incomplete $U$-statistic is defined as the sample average of $h(X_{i_1}, \ldots, X_{i_m})$ taken only over the subset of chosen indices. The test statistic is then formed  as in~\eqref{eq:teststat-complete} by replacing $U_n$ with the incomplete counterpart.  

The main theoretical contribution of this work is to show that the high-dimensional Gaussian approximation of incomplete $U$-statistics remains valid under \emph{mixed degeneracy},  that is, for each index $j$, the variance of the H\'ajek projection
$\sigma_{g,\theta,j}^2$ is allowed to take more or less arbitrary values, including zero. Therefore, our result covers  testing hypotheses as in~\eqref{eq:nullhypothesis} when the underlying true parameter may be irregular or close to irregular.
The setting of mixed degeneracy constitutes a further development of prior results in the literature which prove validity of the Gaussian approximation of incomplete $U$-statistics in either the fully non-degenerate case or the fully degenerate case \citep{chen2019randomized, song2019approximating}. That is, either \emph{all} variances $\sigma_{g,\theta, j}^2$ are bounded away from zero  or \emph{all} of them are equal to zero. Our result is intermediate since we allow a different status of degeneracy for each index $j$. The crucial fact we exploit is that the asymptotic variance in the Gaussian approximation of the incomplete $U$-statistic is a weighted sum of the variance of the H\'ajek projection and the variance of the kernel itself. Hence, vanishing of the variance of the H\'ajek projection need not cause degeneracy of the asymptotic distribution. Indeed, the approximation is valid when we choose the computational budget parameter appropriately, typically of the same order as the sample size, i.e., $N = \mathcal{O}(n)$. 

\begin{figure}[t]
\captionsetup[subfigure]{labelformat=empty}
\centering
\begin{subfigure}{.32\linewidth}
\centering
\includegraphics[width=\linewidth]{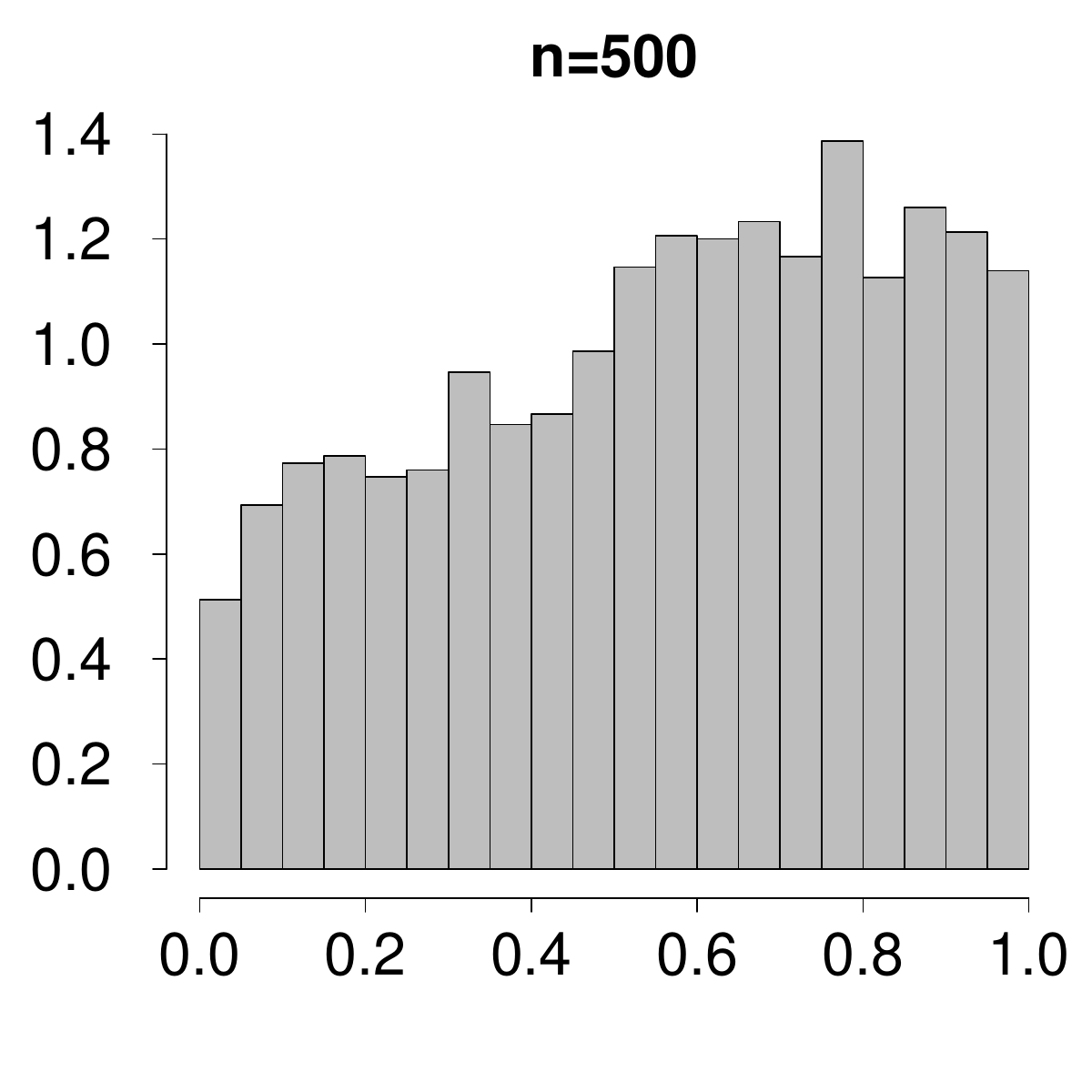}
\vspace{-0.9cm}
\caption{incomplete $U$-statistic}
\end{subfigure}
\hfill
\begin{subfigure}{.32\linewidth}
\centering
\includegraphics[width=\linewidth]{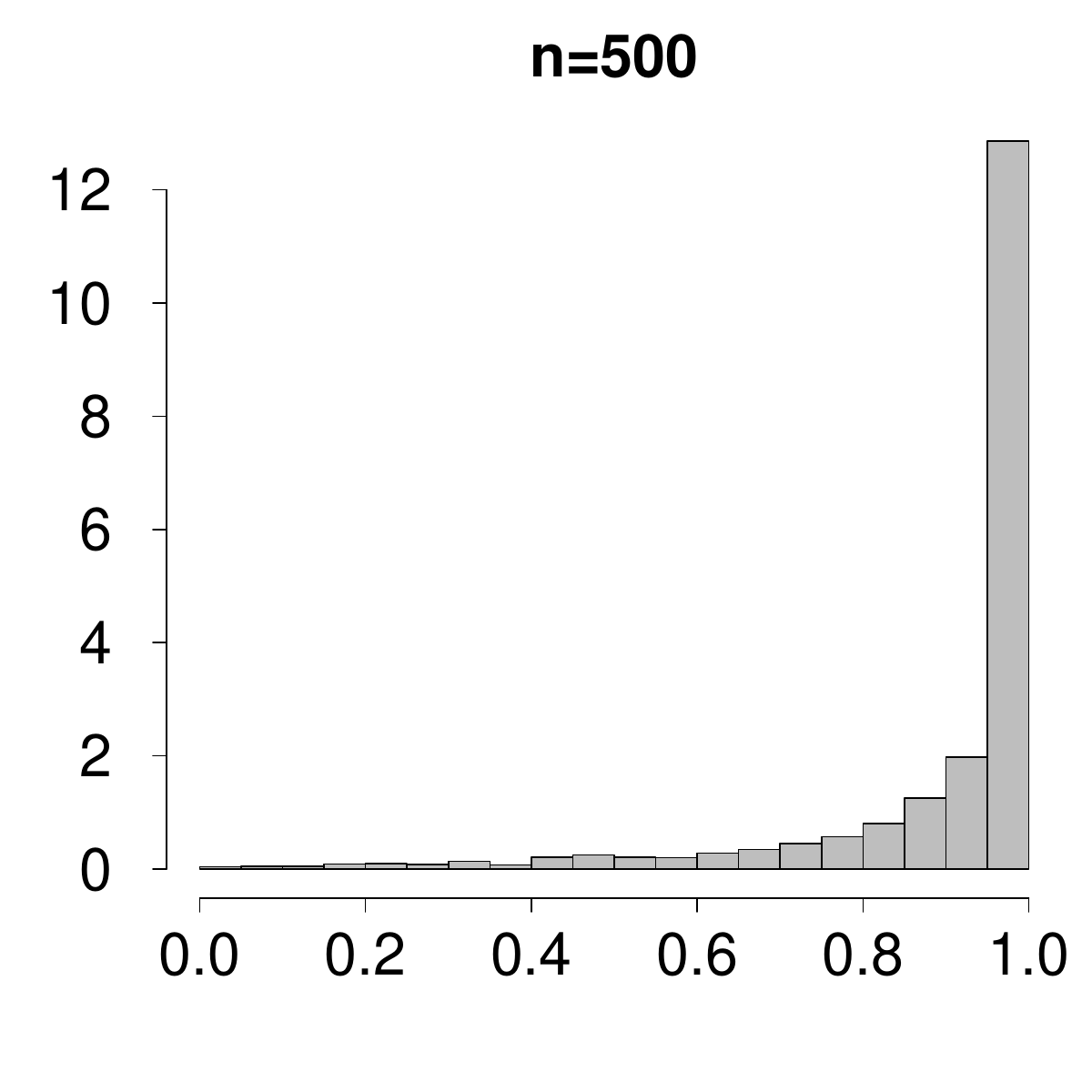}
\vspace{-0.9cm}
\caption{complete $U$-statistic}
\end{subfigure}
\hfill
\begin{subfigure}{.32\linewidth}
\centering
\includegraphics[width=\linewidth]{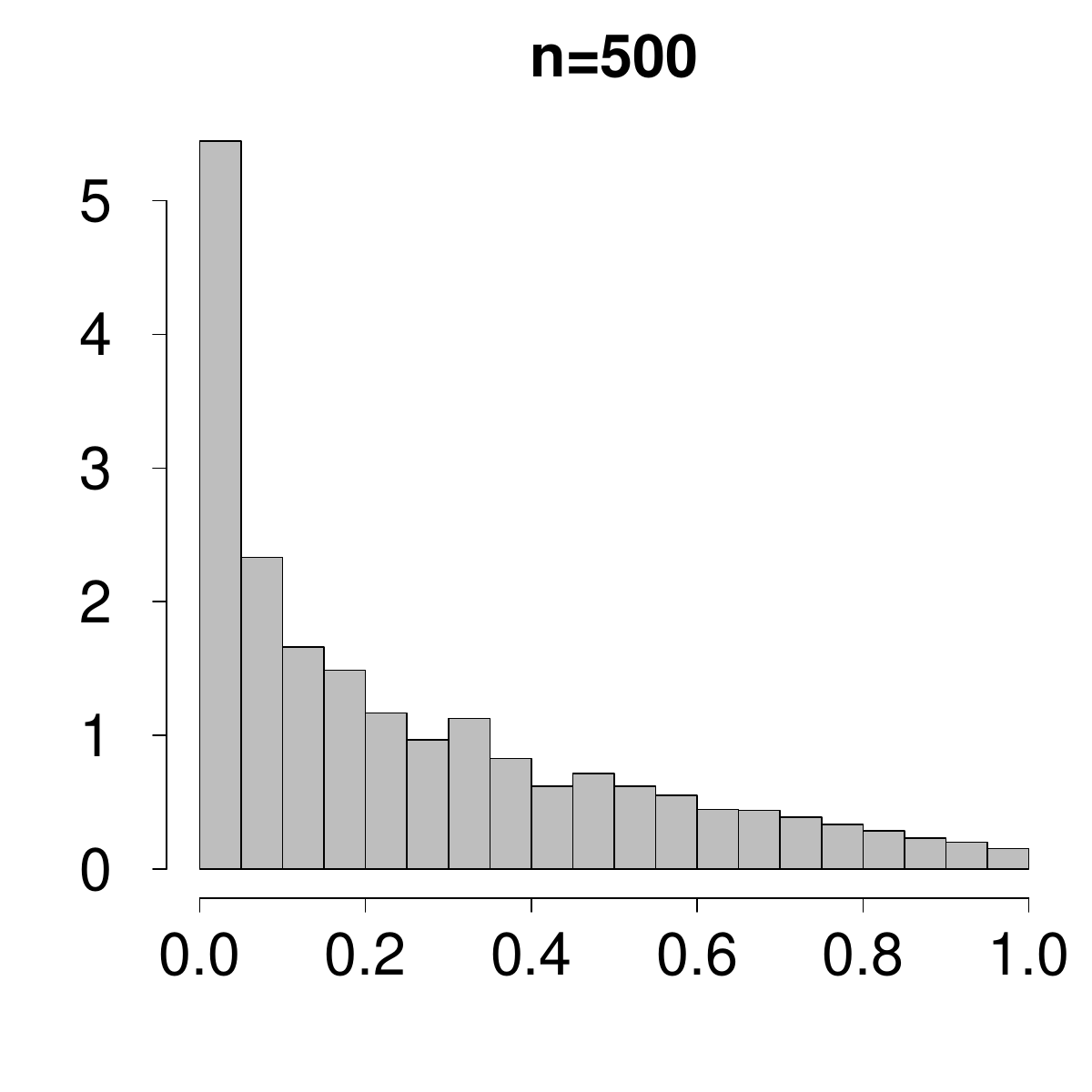}
\vspace{-0.9cm}
\caption{likelihood ratio test}
\end{subfigure}
\par\bigskip
\begin{subfigure}{.32\linewidth}
\centering
\includegraphics[width=\linewidth]{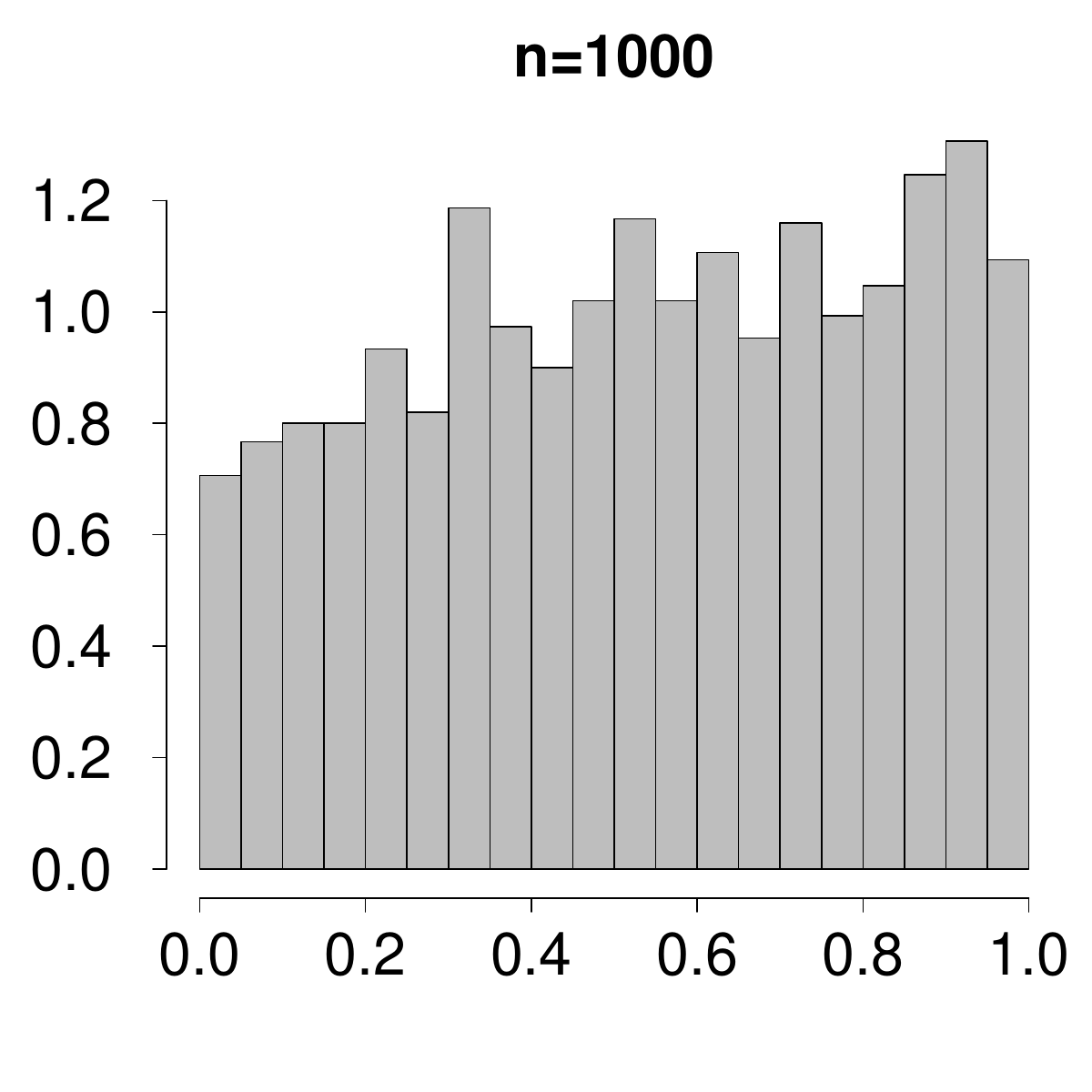}
\vspace{-0.9cm}
\caption{incomplete $U$-statistic}
\end{subfigure}
\hfill
\begin{subfigure}{.32\linewidth}
\centering
\includegraphics[width=\linewidth]{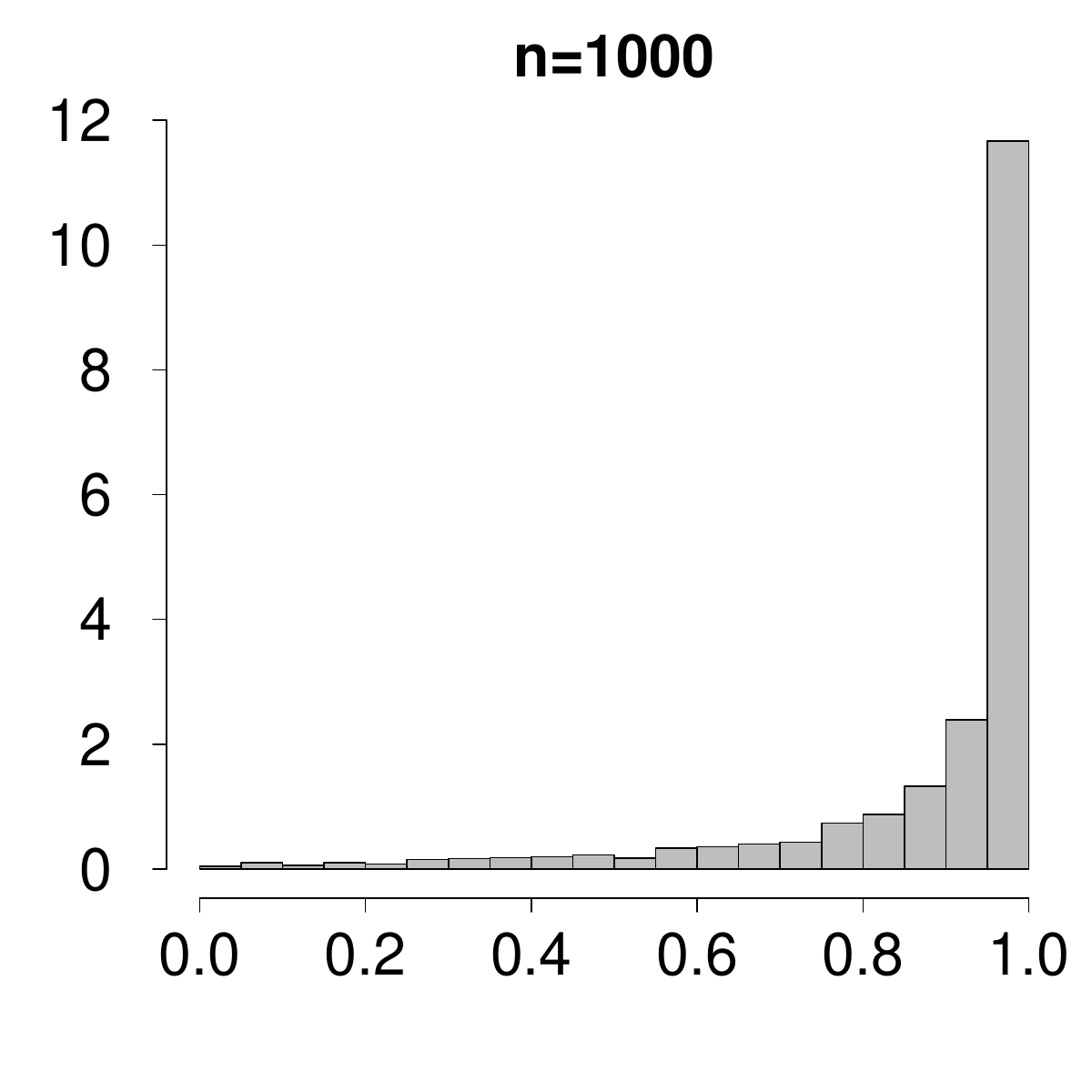}
\vspace{-0.9cm}
\caption{complete $U$-statistic}
\end{subfigure}
\hfill
\begin{subfigure}{.32\linewidth}
\centering
\includegraphics[width=\linewidth]{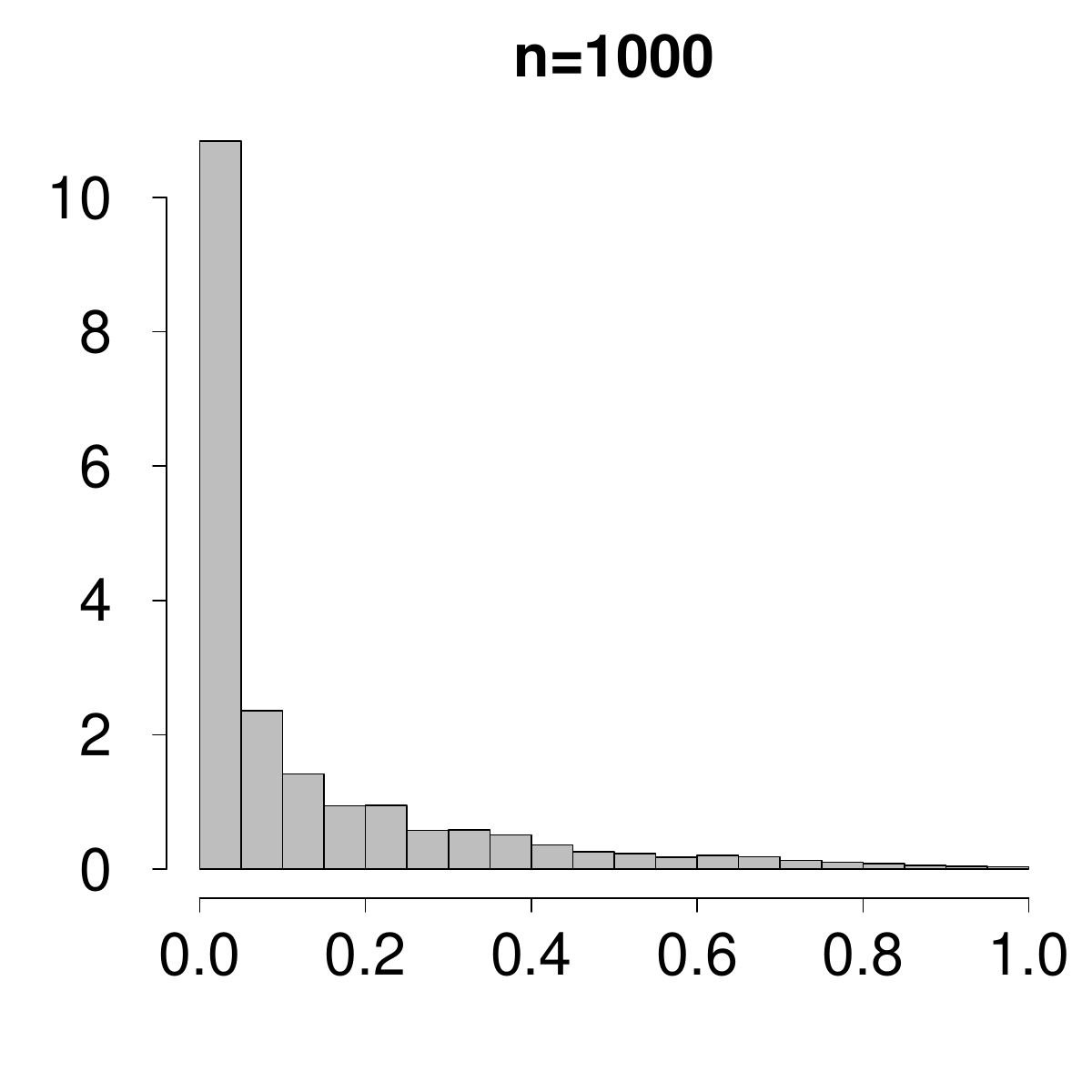}
\vspace{-0.9cm}
\caption{likelihood ratio test}
\end{subfigure}
\caption{Histograms of $5,000$ simulated $p$-values for simultaneously testing $2730$ tetrads constraints implied by the one-factor model with $l=15$ observed variables. The computational budget parameter for the incomplete $U$-statistic is $N=2n$ and the true covariance matrix is close to an irregular point, for exact parameter values see Section~\ref{sec:tree-models}, setup (b).}
\label{fig:lr-pvals}
\end{figure}

To derive critical values of our test statistic, we further approximate the limiting Gaussian distribution via a data-dependent Gaussian multiplier bootstrap as proposed in \citet{chen2019randomized}. We show that when choosing $N$ appropriately,  the bootstrap approximation remains trustworthy under mixed degeneracy and yields asymptotically valid critical values. The bootstrap is computationally feasible even for a large number of constraints $p$, as incomplete $U$-statistics also offer  computational advantages over complete $U$-statistics. The computation of the complete $U$-statistic~\eqref{eq:complete-Ustat} requires $\mathcal{O}(n^m p)$ operations, which can be challenging for $m\geq 3$, while the incomplete $U$-statistic only requires $\mathcal{O}(N p)$ operations. We would like to highlight that these computational advantages were the main motivaton for consideration of incomplete $U$-statistics in prior literature; in contrast, our work raises statistical advantages.
\looseness=-1

\begin{exmp}
Figure~\ref{fig:lr-pvals} shows histograms of simulated p-values for testing a large number of tetrad constraints in the one-factor analysis model when the true parameter matrix is close to an irregular point. 
In addition to our proposed strategy of using incomplete $U$-statistics, we  include two other strategies: one that is based on complete $U$-statistics and another based on the likelihood ratio test. The method using incomplete $U$-statistics yields $p$-values that are only slightly conservative (nearing a uniform distribution), while the other two methods fail drastically. \looseness=-1
\end{exmp}

Our strategy can be applied for general  hypotheses ~\eqref{eq:nullhypothesis} as long as the functions $f_j(\theta)$ are estimable using a kernel function $h$. This is an advancement of \citet{leung2018algebraic}, where only tetrads  
are considered. \citet{leung2018algebraic} also used a kernel function to estimate the tetrads, but the test statistic is an $m$-dependent average over the kernels instead of an incomplete $U$-statistic. The  randomized incomplete $U$-statistic proposed in this work is much more flexible, supports broad application, and yields better results in simulations. Moreover, we show useful theoretical guarantees of incomplete $U$-statistics that yield an asymptotically valid test even at irregularities. 

\begin{rem}[Nonparametric setups] \label{rem:nonparmetric-setup}
The setting we consider is formulated as pertaining to a parametric model with $d$-dimensional parameter space.  However, our testing methodology applies without change to  settings where we test constraints on a $d$-dimensional parameter $\theta(P)$ of the distributions in a nonparametric model.  For example, one could consider testing tetrad constraints in covariance matrices of non-Gaussian distributions. \looseness=-1
\end{rem}

\begin{rem}[Comparison to literature on shape restrictions] \label{rem:shape-restrictions}
It is natural to compare our work to recent progress in the literature on shape restrictions that also considers testing equality and inequality constraints; see \citet{chetverikov2018theeconometrics} for a review. In this line of work, the restricted parameter space $\Theta_0 \subseteq \Theta$ is usually considered to be infinite dimensional, so that more general parameters such as entire function classes are covered. On the other hand, to the best of our knowledge, the methods do not explicitly focus on setups with a large amount of constrains and possibly irregular points. For example, in  \citet{chernozhukov2023constrained}, the authors consider a test statistic that minimizes a generalized method of moments objective function over the restricted and the whole parameter space and compares the difference.
It is the main goal to study the behaviour of the test statistic in regions near the boundary of the restricted parameter space $\Theta_0$. 
The different focus is reflected in the conditions the authors of \citet{chernozhukov2023constrained} assume for their theoretical analysis. In particular, they assume that the Jacobian matrix of the equality constraints has full row rank in a neighborhood of the true parameter. This implies that the maximal number of equality constraints is smaller than the dimension of the parameter space and that algebraic singularities are excluded since the Jacobian is not allowed to drop rank. Thus, the conditions do not allow for \emph{many} equality constraints and \emph{irregular} points. 
Moreover, the test statistic requires an optimization over the restricted and the whole parameter space which requires extra assumptions such as convexity and compactness and can be difficult to implement in practice. In contrast, our method is optimization free and does not require any assumptions on the parameter space. 
\end{rem}

\subsection{Organization of the paper}

In Section~\ref{sec:incomplete-U-stat} we give non-asymptotic Berry-Esseen type bounds for high-dimensional Gaussian and bootstrap approximation of incomplete $U$-statistics. Importantly, the incomplete $U$-statistic is assumed to be of mixed degeneracy.
In Section~\ref{sec:methodology} we propose our testing methodology by formally defining the test statistic and showing how to derive critical values.  Our results on incomplete $U$-statistics yield that the test is asymptotically valid and consistent even in irregular settings.
In Section~\ref{sec:polynomial-hypotheses} we show that our test is applicable to  polynomial hypotheses by explaining a general method for constructing a kernel $h$. 
In Section~\ref{sec:tree-models} we then apply our strategy for testing the goodness-of-fit of  latent tree models of which the one-factor analysis model is a special case. In numerical experiments we compare our strategy with the likelihood ratio test. The supplement \citep{supplemental} contains additional material such as all technical proofs (Appendix~\ref{sec:proofs}), additional lemmas (Appendix~\ref{sec:useful-lemmas}), properties of sub-Weibull random variables (Appendix~\ref{sec:sub-Weibull}) and additional simulation results for Gaussian latent tree models (Appendix~\ref{sec:additional-simulations}). Moreover, we provide a second application of our methodology in Appendix~\ref{sec:2-factor}, where we test minors in two-factor analysis models.

\subsection{Notation}
Given $\beta \in (0, \infty)$, we define the function $\psi_{\beta}(x) = \exp(x^{\beta})-1$ for $x > 0$. A real random variable $Y$ is said to be \textit{sub-Weibull of order $\beta$} if  $\|Y\|_{\psi_{\beta}} \defeq  \inf\left\{t > 0 : \mathbb{E}\left[\psi_{\beta}\left(|Y| / t\right)\right] \leq 1 \right\}$ is finite. We use the usual convention that $\inf \{\emptyset\} = \infty$. For $\beta \geq 1$ we have that $\|Y\|_{\psi_{\beta}}$ is a norm, while for $\beta \in (0,1)$ it is only a quasinorm. If $\|Y\|_{\psi_{1}} < \infty$, then Y is called a \textit{sub-Exponential} random variable and if $\|Y\|_{\psi_{2}} < \infty$, then Y is called  \textit{sub-Gaussian}. For a random element $Y$, let $P|_Y(\cdot)$ and $\mathbb{E}|_Y[\cdot]$ denote the conditional probability and expectation given $Y$. We denote a sequence of random variables $Y_i, \ldots, Y_{i^{\prime}}$ by $Y_i^{i^{\prime}}$ for $i \leq i^{\prime}$ and for a tuple of indices $\iota = (i_1, \ldots, i_m)$ we write $Y_{\iota} = (Y_{i_1}, \ldots, Y_{i_m})$. 

For $a,b \in \mathbb{R}$, define $a \vee b = \max\{a,b\}$ and $a \wedge b = \min\{a,b\}$. For $a,b \in \mathbb{R}^p$, we write $a \leq b$ if $a_j \leq b_j$ for all $j=1, \ldots, p$, and we write $[a,b]$ for the hyperrectangle $\prod_{j=1}^p [a_j,b_j]$. If $a \leq b$, then the hyperrectangle $[a,b]$ is nonempty but if there is at least one index $j \in \{1, \ldots, p\}$ such that $a_j > b_j$, then the hyperrectangle is equal to the empty set.  The class of hyperrectangles in $\mathbb{R}^p$ is denoted by $\mathbb{R}^p_{\mathrm{re}} = \{\prod_{j=1}^p [a_j,b_j] : a_j, b_j \in \mathbb{R} \cup \{-\infty, \infty\}\}$. For a vector $a \in \mathbb{R}^p$ and $r,t \in \mathbb{R}$, we write $ra+t$ for the vector in $\mathbb{R}^p$ with $j$-th component $ra_j+t$. Finally, for a vector $a \in \mathbb{R}^p$ and two integers $p_1, p_2$ such that $p_1 + p_2 = p$ we write $a=(a^{(1)}, a^{(2)})$, where $a^{(1)}_j = a_j$ for all $j=1, \ldots, p_1$ and $a^{(2)}_j = a_{p_1+j}$ for all $j=1, \ldots, p_2$. Let $\|A\|_{\infty} = \max_{i,j}|a_{ij}|$ be the element-wise maximum norm of a matrix $A=(a_{ij})$.

\section{Incomplete U-Statistics under Mixed Degeneracy} \label{sec:incomplete-U-stat}
Suppose we are given i.i.d.~samples $X_1, \ldots, X_n$ from an unknown distribution $P$ in a statistical model. For some integer $m$, let $h(x_1, \ldots, x_m)$ be a fixed $\mathbb{R}^p$-valued measurable function that is symmetric in its arguments. In this section, we consider the general case of inference on the mean vector $\mathbb{E}[h(X_1, \ldots, X_m)] = (\mu_1, \ldots, \mu_p)^{\top} = \mu$, where $\mu$ is an arbitrary estimable parameter of the underlying distribution $P$. In our setup, where we want to test hypotheses characterized by constraints, the distribution $P$ depends on $\theta$ and $\mu$ is given by the  constraints $f(\theta)=(f_1(\theta), \ldots, f_p(\theta))$.

We begin with the formal definition of randomized incomplete $U$-statistics using similar notation as in \citet{chen2019randomized} and \citet{song2019approximating}. 
Let $N \leq \binom{n}{m}$ be a computational budget parameter and generate i.i.d.~Bernoulli random variables $\{Z_{\iota}: \iota \in I_{n,m}\}$ with success probability $\rho_n = N/\binom{n}{m}$. Then the incomplete $U$-statistics based on Bernoulli sampling is defined by \looseness=-1
\begin{equation} \label{eq:def-incomplete-Ustat}
    U_{n,N}^{\prime} = \frac{1}{\hat{N}} \sum_{\iota \in I_{n,m}} Z_{\iota}h(X_{\iota}),
\end{equation}
where $\hat{N} = \sum_{\iota \in I_{n,m}} Z_{\iota}$ is the number of successes. The variable $\hat{N}$ follows a Binomial distribution with parameters ($|I_{n,m}|, \rho_n$). Therefore, $\mathbb{E}(\hat{N})=|I_{n,m}| \rho_n = N$, and the incomplete $U$-statistic is on average a sum over $N$ objects. Thus, we may view the computational budget $N$ as a sparsity parameter for the incomplete $U$-statistic. We denote by $\sigma_{h,j}^2 = \mathbb{E}[(h_j(X_1^m) - \mu_j)^2]$ the variance of the $j$-th coordinate of the kernel and by $\sigma_{g,j}^2 =  \mathbb{E}[(g_j(X_1) - \mu_j)^2]$ the variance of $g_j(X_1)$; recall that the  H\'ajek projection is given by $g(x)=\mathbb{E}[h(x, X_2, \ldots, X_m)]$. Note that $\sigma_{h,j}^2$ and $\sigma_{g,j}^2$ depend on the underlying distribution $P$ as emphasized in the introduction, however, in the rest of the paper we omit the explicit dependence to simplify notation.

\subsection{Gaussian Approximation} \label{subsec:Gaussian-approx}
We will derive non-asymptotic Gaussian approximation error bounds for the incomplete $U$-statistic $U_{n,N}^{\prime}$ that allow for mixed degenerate kernels $h$ when choosing the computational budget parameter $N$ appropriately. To state the formal approximation results, we assume $2 \leq m \leq \sqrt{n}$, $n \geq 4$, $p \geq 3$ and $\rho_n = N/|I_{n,m}| < 1/2$. We start by making assumptions on the moment structure of the kernel $h$ before formally introducing mixed degeneracy. Let $\beta \in (0,1]$ and suppose there exists a constant $D_n \geq 1$ such that:
\vspace{0.25cm}
\begin{enumerate}[label=(C\arabic*)]\setlength{\itemsep}{0.25cm} 
    \item $\mathbb{E}[|h_j(X_1^m) - \mu_j|^{2+l}] \leq \sigma_{h,j}^2 D_n^{l}$ for all $j = 1, \ldots, p$ and $l=1,2$. \label{C1}
    \item $\|h_j(X_1^m) - \mu_j\|_{\psi_{\beta}} \leq D_n$ for all $j = 1, \ldots, p$. \label{C2}
    \item There exists $\ubar{\sigma}_h^2>0$ such that $\ubar{\sigma}_h^2 \leq \min_{1 \leq j\leq p} \sigma_{h,j}^2$. \label{C3}
\end{enumerate}
\vspace{0.25cm}

Assuming conditions similar to~\ref{C1}-\ref{C3} is standard in high-dimensional Gaussian approximation theory. Condition~\ref{C2} assumes that the kernel $h$ is sub-Weibull. In prior work on Gaussian approximation of high-dimensional $U$-statistics \citep{chen2018gaussian, chen2019randomized}, the authors usually consider sub-Exponential kernels with $\beta=1$. However, the kernel we propose in Section~\ref{sec:polynomial-hypotheses} for testing polynomial hypotheses will typically be sub-Weibull and not sub-Exponential, also see Example~\ref{ex:tetrad-sub-weibull}. 
We discuss important properties of sub-Weibull random variables in Appendix~\ref{sec:sub-Weibull}. Note that we allow the bound $D_n$ to depend on the sample size $n$ since, in the high-dimensional setting, the distribution $P$ may depend on $n$.
Condition~\ref{C1} is of a more technical nature and serves for clear presentation. In principle, it would be possible to omit this assumption, but the resulting error bound would be more complicated. 
Finally, Condition~\ref{C3} requires that the minimal variance of the individual kernels $h_j$ is bounded away from zero, even for large $p$. Put differently, we only considers kernels $h_j$ such that $h_j(X_1^m)$ is not almost surely constant.
It remains to make assumptions with respect to the degeneracy of the H\'ajek projection, i.e., we formally define mixed degeneracy.

\begin{defn} \label{def:mixed-degenerate}
Let $p_1, p_2 \in \mathbb{N}$ such that $p_1$ + $p_2 = p$. We say that the kernel $h$, or also simply the incomplete $U$-statistic, is \emph{mixed degenerate} for distribution $P$ if the following two conditions are satisfied:
\vspace{0.25cm}
\begin{enumerate}[resume, label=(C\arabic*)]\setlength{\itemsep}{0.25cm} 
    \item There exists $\ubar{\sigma}_{g^{(1)}}^2 > 0$ such that $\ubar{\sigma}_{g^{(1)}}^2 \leq \min_{1 \leq j \leq p_1} \sigma_{g,j}^2$. \label{C4}
    \item There exists $k>0$ such that $\|g_j(X_1) - \mu_j\|_{\psi_{\beta}} \leq n^{-k} D_n$ for all $j=p_1+1, \ldots, p$.\label{C5}
\end{enumerate}
\end{defn}

In other words, mixed degeneracy of a kernel requires that each index $j=1, \ldots, p$ either satisfies condition~\ref{C4} or~\ref{C5}. By rearranging the indices, we then find that the first $p_1$ indices satisfy Condition~\ref{C4} and the remaining indices satisfy~\ref{C5}.
Note that whether a kernel is mixed degenerate or not depends on the underlying distribution $P$ from which the samples $X_1, \ldots, X_n$ are drawn.

Assuming mixed degenerate kernels is the main difference in comparison to the existing literature on Gaussian approximation of high-dimensional incomplete $U$-statistics \citep{chen2019randomized, song2019approximating}. Usually, either the non-degenerate case where $p_2=0$ or the fully degenerate case where $p_1=0$ and $\sigma_{g,j}^2 = 0$ for all $j=1, \ldots, p$ are treated. In contrast, we allow for a different status of degeneracy of the H\'ajek projection for each index $j$. In particular, Condition~\ref{C5} covers degenerate cases where $\sigma_{g,j}^2=0$ since zero variance implies that $\|g_j(X_1) - \mu_j\|_{\psi_{\beta}} = 0$ almost surely.
But our assumptions even allow for more flexibility. For example, it may be the case that the variance $\sigma_{g,j}^2$ decreases with the sample size $n$. In this case, Condition~\ref{C5} is an assumption on the rate of convergence to degeneracy, i.e., the rate is polynomial in $n$.

\begin{rem}[Discussion on mixed degeneracy] \label{rem:any-distribution}
The notion of mixed degeneracy is particularly interesting for kernels $h$, where we do not know whether the individual components $h_j$  are degenerate or non-degenerate, for example when the underlying distribution is unknown. If the number of constraints $p$ does \emph{not} grow with the sample size, then mixed degeneracy holds for any distribution $P$. Indeed, by letting $\ubar{\sigma}_{g^{(1)}}^2$ be the minimum of the nonzero variances $\sigma^2_{g,j}$, we see that~\ref{C4} is satisfied. All other indices $j$ have $\sigma^2_{g,j}=0$ which implies that~\ref{C5} is also satisfied. However, we emphasize that mixed degeneracy is a more subtle condition if the number of constraints $p$ is growing. 
\end{rem}

\begin{rem}[Parametric families and irregular points] \label{rem:irregular-setups}
In our testing problem~\eqref{eq:testing-problem}, we are considering a parametric family of distributions $\{P_{\theta}, \theta \in \Theta\}$ but the true parameter $\theta$ is unknown.  In this case, the variances $\sigma_{g,j}^2$ of the individual H\'ajek projections depend on $\theta$ as outlined in the introduction. If the number of constraints $p$ does \emph{not} grow with the sample size, then mixed degeneracy holds uniformly over the whole parameter space $\Theta$, as we have seen in Remark~\ref{rem:any-distribution}. 
However, the rate of convergence of the incomplete $U$-statistic depends on the minimum $\ubar{\sigma}_{g^{(1)}}^2$. 
If $\ubar{\sigma}_{g^{(1)}}^2$ is really small, the index corresponding to the minimum may already satisfy~\ref{C5} for large sample sizes $n$, so $\ubar{\sigma}_{g^{(1)}}^2$ can in fact be chosen larger. Therefore, mixed-degeneracy is also suitable for points $\theta \in \Theta$ that are close to irregular, see Corollary \ref{cor:Gaussian-case} for a precise statement.
\end{rem}

In principle, it would be possible to extend our results on more general sequences $\gamma_n$ converging to zero instead of $n^{-k}$ in Condition~\ref{C5}, but this would result in more involved error bounds. For simplicity, we also assume $p_1, p_2 \geq 3$, even though one could specify the bounds for arbitrary $p_1$ and $p_2$. Our last  technical assumption is similar to~\ref{C1} and is also necessary for the sake of clear presentation: 
\vspace{0.25cm}
\begin{enumerate} 
    \item[{\crtcrossreflabel{(C6)}[C6]}] $\mathbb{E}[|g_j(X_1) - \mu_j|^{2+l}] \leq \sigma_{g,j}^2 D_n^{l}$ for all $j = 1, \ldots, p$ and $l=1,2$. 
\end{enumerate}
\vspace{0.25cm}

Now, we state our main result that specifies a nonasymptotic error bound on the Gaussian approximation of incomplete $U$-statistics. We define $\alpha_n = n/N$, $\Gamma_h = \textrm{Cov}[h(X_1^m)]$ and $\Gamma_g = \textrm{Cov}[g(X_1)]$. For notational convenience, we further define the quantities
\[
    \omega_{n,1} = \left(\frac{m^{2/\beta} D_n^2 \log(pn)^{1+6/\beta}}{(\ubar{\sigma}_{g^{(1)}}^2 \wedge \ubar{\sigma}_h^2 \wedge 1) \ (n \wedge N)}\right)^{1/6},  \quad
    \omega_{n,2} = \frac{N^{1/2} m^2 D_n \log(p n)^{1/2 + 2/\beta}}{\ubar{\sigma}_h n^{\min\{1/2+k,\, 5/6\}}},
\]
\[
    \omega_{n,3} = \left(\frac{N m^2 D_n^2 \log(p)^{2}}{(\ubar{\sigma}_h^2 \wedge 1) n^{\min\{1+k,m\}}}\right)^{1/3}.
\]

\begin{thm} \label{thm:incompleteUstat}
Assume~\ref{C1} -~\ref{C6} hold. Then there is a constant $C_{\beta} > 0$ only depending on $\beta$ such that
\begin{align*} 
    \sup_{R \in \mathbb{R}^p_{\mathrm{re}}} &|P(\sqrt{n}(U_{n,N}^{\prime} - \mu) \in R) - P(Y \in R)| \leq C_{\beta} \{\omega_{n,1} + \omega_{n,2} + \omega_{n,3}\},
\end{align*}
where $Y \sim N_p(0,m^2 \Gamma_g + \alpha_n \Gamma_h)$. 
\end{thm}

Theorem~\ref{thm:incompleteUstat} shows that the distribution of $\sqrt{n}(U_{n,N}^{\prime} - \mu)$ can be approximated by the Gaussian distribution $Y \sim N_p(0,m^2 \Gamma_g + \alpha_n \Gamma_h)$ under mixed degeneracy. Since the computational budget parameter $N$ occurs in the numerator of the bound, one has to choose it in proportion to the sample size  such that the bound vanishes. In particular, we note that under the regime $N=\mathcal{O}(n)$ the bound vanishes when treating other quantities as constants.

\begin{exmp}\label{ex:comments-on-bound}
Let $N$ be of the same order as the sample size. Then $\alpha_n$ can be viewed as a constant and each coordinate of $Y$ is asymptotically non-degenerate. In this case, when we assume $k \geq 1/3$ and treat $m, \ubar{\sigma}_{g^{(1)}}, \ubar{\sigma}_h$ and $D_n$ as fixed constants, the bound $C_{\beta} \{\omega_{n,1} + \omega_{n,2} + \omega_{n,3}\}$ vanishes asymptotically under Conditions~\ref{C1} -~\ref{C6} if the dimension $p$ satisfies $\log(pn)^{3/2 + 6/\beta} = \smallO(n)$. 
On the other hand, one can also choose $N=n^{1+\varepsilon}$ for small enough $\varepsilon > 0$. Then, $\alpha_n \to 0$ as $n \to \infty$ and the distribution of $Y$ may become asymptotically degenerate. Nevertheless, the bound still vanishes if we fix all other quantities as before and if the dimension $p$ satisfies $\log(pn)^{\frac{3/2 + 6/\beta}{1/3-\varepsilon/2}} = \smallO(n)$. 
\end{exmp}

\begin{rem}[Order of the kernel] \label{rem:order}
The Berry-Esseen type bound in Theorem 2.4 depends explicitly on the order of the kernel $h$. In particular, the result also allows for kernels of diverging order, i.e., increasing $m$. However, larger $m$ imply  worse performance of the Gaussian approximation in terms of the required sample size. Also, if $m$ is increasing with $n$ one has to be careful in choosing the computational budget $N$;  it has to be chosen smaller to achieve convergence.
\end{rem}

The proof of Theorem~\ref{thm:incompleteUstat} relies on the seminal papers of  \citet{chernozhukov2013gaussian} and \citet{chernozhukov2017central} on Gaussian approximation of high-dimensional independent sums, and it extends the results of \citet{chen2019randomized} and \citet{song2019approximating} on incomplete $U$-statistics to the mixed degenerate case. Obtaining sharper bounds for the high-dimensional approximation of independent sums than in the original papers is an ongoing area of research, see for example \citet{fang2021high-dimensional}, \citet{chernozhukov2022improved}, \citet{lopes2022central} and \citet{chernozhukov2023nearlyoptimal}.

The bound in Theorem~\ref{thm:incompleteUstat} is stated as general as possible. Importantly, it entails that the incomplete $U$-statistic can be approximated by $N_p(0,m^2 \Gamma_g + \alpha_n \Gamma_h)$ even in irregular setups of our testing problem~\eqref{eq:testing-problem} as long as $\underline{\sigma}_{g^{(1)}}^2 > 0$. However, it might be difficult to read off precise rates for the speed of convergence in the \emph{close to irregular} scenarios. For large $n$, it is possible that there are close to irregular points such that some indices $j \in \{1, \ldots, p\}$ only satisfy mixed degeneracy if one chooses $\underline{\sigma}_{g^{(1)}}^2$ relatively small or $k$ small. On the other hand, under further distributional assumptions, we manage to improve the bound in Theorem 2.4 so that the speed of convergence is completely independent to the irregularity status of the hypothesis.

\begin{cor} \label{cor:Gaussian-case}
    Assume that $X_1, \ldots, X_n$ are i.i.d.~samples of a Gaussian distribution and assume that each individual kernel $h_1, \ldots, h_p$ is a nonconstant polynomial of degree at most $2s$. Suppose that $\mu=\mathbb{E}[h(X_1^m)]=0$ and that there exists $\underline{\sigma}_h^2 > 0$ such that $\underline{\sigma}_h^2 \leq \min_{1 \leq j \leq p} \sigma_{h,j}^2$. Then there exist $\beta \in (0,1]$ and $D_n \geq 1$ such that \ref{C1}, \ref{C2} and \ref{C6} are satisfied and the kernel $h$ is mixed degenerate. If, additionally, $n \leq N \leq C n$ for some constant $C > 0$, then it holds that
    \begin{align*} 
        \sup_{R \in \mathbb{R}^p_{\mathrm{re}}} &|P(\sqrt{n} U_{n,N}^{\prime} \in R) - P(Y \in R)| \leq C_{s,m} \frac{(\overline{\sigma}_h^2 \vee 1)\log(pn)^{1/2+2s}}{(\underline{\sigma}_h^2 \wedge 1) \, n^{1/9}} ,
    \end{align*}
    where $Y \sim N_p(0,m^2 \Gamma_g + \alpha_n \Gamma_h)$, $\overline{\sigma}_h^2=\max_{1 \leq j \leq p} \sigma_{h,j}^2$ and $C_{s,m} > 0$ is a constant only depending on $s$ and $m$. 
\end{cor}

In the bound in Corollary \ref{cor:Gaussian-case}, there is no $\underline{\sigma}_{g^{(1)}}^2$ and no $k$ showing up anymore, so the convergent speed shows no dependence on how close to irregular the points in the null hypothesis really are. Even though the bound might not be optimal, it completely  mathematizes our intuition why incomplete U-statistics are a good choice for a test statistic to guard against irregular points.

\begin{exmp}
    Each coordinate $h_j$ of the proposed kernel in Example~\ref{exmp:tetrad-kernel} is a polynomial of degree $4$ in Gaussian variables, which implies that Corollary \ref{cor:Gaussian-case} is applicable with $s=2$. See Section~\ref{sec:polynomial-hypotheses} for a general method to construct polynomial kernels.
\end{exmp}

\subsection{Bootstrap Approximation} \label{subsec:bootstrap-approx}
Since the covariance matrix $m^2 \Gamma_g + \alpha_n \Gamma_h $ of the approximating Gaussian distribution in Theorem~\ref{thm:incompleteUstat} is typically unknown in statistical applications, we apply a Gaussian multiplier bootstrap. The procedure is exactly the same as in \citet{chen2019randomized} and \citet{song2019approximating} but their error bounds on the approximation require non-degeneracy, that is, there is a constant $c > 0$ such that $\min_{1 \leq j \leq p}\sigma_{g,j}^2 \geq c$.
Under mixed degeneracy this is no longer the case. However, when choosing the computational budget appropriately, we prove that the bootstrap approximation still holds.

The bootstrap is based on the fact that the random vector $Y \sim N_p(0, m^2 \Gamma_g + \alpha_n \Gamma_h)$ is a weighted sum of the two independent random vectors $Y_g \sim N_p(0, \Gamma_g)$ and $Y_h \sim N_p(0, \Gamma_h)$, i.e., $Y = mY_g + \sqrt{\alpha_n} Y_h$. Let $\mathcal{D}_n=\{X_1, \ldots, X_n\} \cup \{Z_{\iota}: \iota \in I_{n,m}\}$ be the data involved in the definition of the incomplete $U$-statistic $U_{n,N}^{\prime}$. We will construct data dependent random vectors $U_{n_1,g}^\#$ and $U_{n,h}^\#$ such that, given the data $\mathcal{D}_n$, both vectors are independent and approximate $Y_g$ and $Y_h$.

To approximate the distribution of $Y_h$ take a collection $\{\xi_{\iota}^{\prime}: \iota \in I_{n,m}\}$ of independent $N(0,1)$ random variables that are also independent of $\mathcal{D}_n$. Define the multiplier bootstrap
\begin{equation} \label{eq:def-Unh}
    U_{n,h}^\# = \frac{1}{\sqrt{\widehat{N}}} \sum_{\iota \in I_{n,m}} \xi_{\iota}^{\prime} \sqrt{Z_{\iota}}(h(X_{\iota}) - U_{n,N}^{\prime})
\end{equation}
and observe that, conditioned on the data $\mathcal{D}_n$, the distribution of $U_{n,h}^\#$ is Gaussian with mean zero and covariance matrix
$
    \widehat{N}^{-1} \sum_{\iota \in I_{n,m}} Z_{\iota} (h(X_{\iota}) - U_{n,N}^{\prime}) (h(X_{\iota}) - U_{n,N}^{\prime})^{\top}. 
$
Intuitively, this covariance matrix should be a good estimator of the true covariance matrix $\Gamma_h$ and therefore the distribution of $U_{n,h}^\#$ should be ``close'' to the distribution of  $Y_h$.

Approximating $Y_g \sim N_p(0, \Gamma_g)$ is more involved since the H\'ajek projection~\eqref{eq:hajek-proj} is in general unknown. Thus, we first construct estimates $G_{i_1}$ of $g(X_{i_1})$ for each $i_1$ in a chosen subset $S_1 \subseteq \{1, \ldots, n\}$ with cardinality $n_1 = |S_1|$. Then, we consider the multiplier bootstrap distribution
\begin{equation} \label{eq:U_n_g}
    U_{n_1,g}^\# = \frac{1}{\sqrt{n_1}} \sum_{i_1 \in S_1} \xi_{i_1} (G_{i_1} - \overline{G}),
\end{equation}
where $\{\xi_{i_1}: i_1 \in S_1\}$ is a collection of independent $N(0,1)$ random variables that is independent of $\mathcal{D}_n$ and $\{\xi_{\iota}^{\prime}: \iota \in I_{n,m}\}$. Here, $\overline{G} = n_{i_1}^{-1}\sum_{i_1 \in S_1} G_{i_1}$ denotes the average of the constructed estimates. The exact form of the estimates $G_{i_1}$ is specified later. Similar as above, the distribution of $U_{n_1,g}^\#$ given the data $\mathcal{D}_n$ should be ``close'' to the distribution of $Y_g$. 
Combining $U_{n,h}^\#$ and $U_{n_1,g}^\#$, we obtain the multiplier bootstrap
$ 
    U_{n,n_1}^\# = m U_{n_1,g}^\# + \sqrt{\alpha_n} U_{n,h}^\#.
$
It approximates the distribution of $Y$ even under mixed degeneracy, as we verify in our next Lemma.  The approximation of $U_{n,n_1}^\#$ will depend on the quality of the estimator $G_{i_1}$ that we measure by the quantity
\begin{equation} \label{eq:delta_g_1}
    \widehat{\Delta}_{g,1} = \max_{1 \leq j \leq p} \frac{1}{n_1} \sum_{i_1 \in S_1} (G_{i_1,j} - g_j(X_{i_1}))^2.
\end{equation}
Moreover, the approximation depends on conditions involving the quantity
\[
    A_n \defeq A_n(\ubar{\sigma}_h, m, \beta, N) =  \frac{m^{\max\{2/\beta,4\}}}{\ubar{\sigma}_h^4 \wedge 1} \max \{\left(N/n\right)^2, 1\},
\]
which is required to be small. In particular, if $m$ does not grow with $n$ and the computational budget $N$ is chosen appropriately, then Assumption~\ref{C3} ensures that $A_n$ does not inflate.

\begin{lem} \label{lem:bootstrap}
Assume the conditions~\ref{C1} -~\ref{C3} hold. If
\begin{equation} \label{eq:bootstrapA1}
    \frac{A_n D_n^4 \log(pn)^{3+4/\beta}}{n_1\wedge N} \leq C_1 n^{-\zeta_1}
\end{equation}
and
\begin{equation} \label{eq:bootstrapA2}
    P\left( A_n D^2_n \widehat{\Delta}_{g,1} \log(p)^4 > C_1 n^{-\zeta_2}\right)\leq \frac{C_1}{n}
\end{equation}
for some constants $C_1 > 0$ and $\zeta_1, \zeta_2 \in (0,1)$, then there exists a constant $C>0$ depending only on $\beta, \zeta_1$ and $C_1$ such that with probability at least $1-C/n$,
\[
    \sup_{R \in \mathbb{R}^p_{\mathrm{re}}}\left|P|_{\mathcal{D}_n}(U_{n,n_1}^\# \in R) - P(Y \in R)\right| \leq C n^{-(\zeta_1 \wedge \zeta_2)/6}.
\]
\end{lem}

Note that Lemma~\ref{lem:bootstrap} formally does not require mixed degeneracy, that is, it is completely independent of the status of degeneracy of the kernel.
Now, we specify $G_{i_1}$ to be a special case of the divide-and-conquer estimator $\hat{g}_{i_1}$ from \citet{chen2019randomized} and \citet{song2019approximating} that is defined as follows.
For each index $i_1 \in S_1$, we partition the remaining indices, $\{1, \ldots, n\}\setminus \{i_1\}$, into disjoint subsets $\{S_{2,k}^{(i_1)}, k= 1, \ldots, K\}$, each of size $m-1$, where $K = \lfloor (n-1)/(m-1) \rfloor$. Then, we define for each $i_1 \in S_1$ the estimator
\begin{equation} \label{eq:estimator_of_g}
    \hat{g}_{i_1} = \frac{1}{K} \sum_{k=1}^K h(X_{i_1}, X_{S_{2,k}^{(i_1)}}).
\end{equation}
Thus, from now on, we again  refer by $U_{n_1,g}^\#$ to the statistics defined in~\eqref{eq:U_n_g} but with the specialized estimator $G_{i_1}=\hat{g}_{i_1}$ as defined in~\eqref{eq:estimator_of_g}. The next theorem builds on Lemma 2.7 and verifies that the bootstrap approximation is valid for this specialized estimator.

\begin{thm} \label{thm:bootstrap}
Assume the Conditions~\ref{C1} -~\ref{C3} hold and
\begin{equation} \label{eq:main-bootstrap-cond}
    \frac{A_n D_n^4 \log(pn)^{3+4/\beta}}{n_1\wedge N} \leq C_1 n^{-\zeta}
\end{equation}
for some constants $C_1 > 0$, $\zeta \in (0,1)$. Then, for any $\nu \in (\max\{7/6,1/\zeta\},\infty)$, there exists a constant $C>0$ depending only on $\beta$, $\nu$, $\zeta$ and $C_1$ such that with probability at least $1-C/n$,
\[
    \sup_{R \in \mathbb{R}^p_{\mathrm{re}}}\left|P|_{\mathcal{D}_n}(U_{n,n_1}^\# \in R) - P(Y \in R)\right| \leq C n^{-(\zeta -1/\nu)/6}.
\]
\end{thm}

Theorem~\ref{thm:bootstrap} says that we can approximate the asymptotic Gaussian distribution on the hyperrectangles via the multiplier bootstrap $U_{n,n_1}^\# = m U_{n_1,g}^\# + \sqrt{\alpha_n} U_{n,h}^\#$. As long as the computational budget $N$ is chosen appropriately, for example $N=\mathcal{O}(n)$, this holds independently of the status of degeneracy since we do not require Conditions~\ref{C4} -~\ref{C6}.  Crucially, we are able to simulate the distribution of $U_{n,n_1}^\#$ given the data by generating independent $N(0,1)$ random variables. The complexity of the bootstrap procedure, and therefore the complexity of our testing methodology, is discussed in Remark~\ref{rem:complexity}.
In a corollary we combine the Gaussian approximation with the bootstrap approximation.

\begin{cor} \label{cor:all-together}
Assume the Conditions~\ref{C1} -~\ref{C6} hold. Further assume that for some constants $C_1 > 0$, $\zeta \in (0,1)$, Condition~\eqref{eq:main-bootstrap-cond} holds and $\omega_{n,1} + \omega_{n,2} + \omega_{n,3} \leq C_1 n^{-\zeta/7}$. Then there exists a constant $C>0$ depending only on $\beta, \zeta$ and $C_1$ such that with probability at least $1-C/n$,
\[
    \sup_{R \in \mathbb{R}^p_{\mathrm{re}}}\left|P(\sqrt{n}(U_{n,N}^{\prime} - \mu) \in R) - P|_{\mathcal{D}_n}(U_{n,n_1}^\# \in R)\right| \leq C n^{-\zeta/7}.
\]
\end{cor}
\begin{proof}
    This follows from Theorem~\ref{thm:incompleteUstat} and Theorem~\ref{thm:bootstrap} with $\nu = 7/\zeta$. 
\end{proof}

\subsection{Studentization} \label{subsec:studentization}
Often, the approximate variances of the coordinates of $U_{n,N}^{\prime}$ are heterogeneous, and it is therefore desirable to studentize the incomplete $U$-statistic. For $j=1, \ldots, p$, we denote by ${\sigma}^2_j = m^2 {\sigma}_{g,j}^2 + \alpha_n {\sigma}_{h,j}^2$ the variance of the $j$-th coordinate of the approximating Gaussian random vector $Y$, that is, ${\sigma}^2_j$ is equal to the diagonal element $m^2\Gamma_{g,jj}+\alpha_n\Gamma_{h,jj}$ of the approximating covariance matrix. In line with this, we define the empirical variances $\hat{\sigma}^2_j = m^2 \hat{\sigma}_{g,j}^2 + \alpha_n \hat{\sigma}_{h,j}^2$ to be the diagonal elements of the conditional covariance matrix of the bootstrap distribution $U_{n,n_1}^\#$ given the data.
Therefore, $\hat{\sigma}_{g,j}^2$ and $\hat{\sigma}_{h,j}^2$ are given by
\begin{equation*} 
   \hat{\sigma}_{g,j}^2 = \frac{1}{n_1} \sum_{i_1 \in S_1} (\hat{g}_{i_1,j} - \overline{g}_j)^2 \quad \textrm{and} \quad \hat{\sigma}_{h,j}^2 = \frac{1}{\widehat{N}} \sum_{\iota \in I_{n,m}} Z_{\iota} (h_j(X_{\iota}) - U_{n,N,j}^{\prime})^2.
\end{equation*}
Moreover, we define a $p \times p$ diagonal matrix $\widehat{\Lambda}$ with diagonal elements
$\widehat{\Lambda}_{jj} = m^2 \hat{\sigma}_{g,j}^2 + \alpha_n \hat{\sigma}_{h,j}^2$ for all $j=1,\ldots,p$. 

\begin{cor} \label{cor:studentization}
Assume the conditions in Corollary~\ref{cor:all-together}. Then there exists a constant $C>0$ depending only on $\beta, \zeta$ and $C_1$ such that with probability at least $1-C/n$,
\[
    \sup_{R \in \mathbb{R}^p_{\mathrm{re}}}\left|P(\sqrt{n} \, \widehat{\Lambda}^{-1/2}(U_{n,N}^{\prime} - \mu) \in R) - P|_{\mathcal{D}_n}(\widehat{\Lambda}^{-1/2}U_{n,n_1}^\# \in R)\right| \leq C n^{-\zeta/7}.
\]
\end{cor}
Corollary~\ref{cor:studentization} allows us to construct a test for hypotheses of the form~\eqref{eq:nullhypothesis} that asymptotically controls type I error and has power against alternatives outside a small neighborhood of the null hypothesis.

\section{Testing Methodology} \label{sec:methodology}
In this section we propose our test based on incomplete $U$-statistics. Recall that $X_1, \ldots, X_n$ are i.i.d.~samples from a distribution $P_\theta$ with parameter $\theta \in \Theta \subseteq \mathbb{R}^d$. 

\subsection{Test Statistic}

We want to test null hypotheses $\Theta_0 \subseteq \Theta$  defined  by constraints $f(\theta)=(f_1(\theta), \ldots, f_p(\theta))$ as specified in~\eqref{eq:nullhypothesis}.
For now, we assume that an $\mathbb{R}^p$-valued, measurable and symmetric  function $h(x_1, \ldots, x_m)$ exists such that $\mathbb{E}[h(X_1^m)]=f(\theta)$. For the case of polynomial hypotheses we show a general construction of kernels in Section~\ref{sec:polynomial-hypotheses}. 
We define the test statistic $\mathcal{T}$ to be the maximum of a studentized incomplete $U$-statistic, that is,
\[
    \mathcal{T} = \max_{1 \leq j\leq p} \sqrt{n} \, U_{n,N,j}^{\prime} / \hat{\sigma}_j,
\]
where $\hat{\sigma}_j^2$ is the empirical approximate variance of the $j$-th coordinate of the incomplete $U$-statistic; recall the definition in Section~\ref{subsec:studentization}.
Large values of $\mathcal{T}$ indicate that the null hypothesis is violated and thus it is natural to reject $H_0$ when $\mathcal{T}$ exceeds a certain critical value. 
\subsection{Critical Value} \label{subsec:critical-value}
The idea for construction of a critical value relies on the observation  
\begin{equation} \label{eq:critical-value}
    \mathcal{T} \leq \max_{1 \leq j \leq p} \sqrt{n} (U_{n,N,j}^{\prime} - f_j(\theta)) / \hat{\sigma}_j,
\end{equation}
whenever $\theta$ is a point in the null hypothesis $\Theta_0$. Hence, to make the test size less or equal than $\alpha$, it is enough to choose the critical value as the $(1-\alpha)$-quantile of the distribution of $\max_{1 \leq j \leq p} \sqrt{n} (U_{n,N,j}^{\prime} - f_j(\theta)) / \hat{\sigma}_j$, which is a centered version of our test statistic $\mathcal{T}$. Since the distribution of $\max_{1 \leq j \leq p} \sqrt{n} (U_{n,N,j}^{\prime} - f_j(\theta)) / \hat{\sigma}_j$ is unknown, we approximate it using the Gaussian multiplier bootstrap introduced in Section~\ref{sec:incomplete-U-stat}. 
The distribution function $P(\max_{1 \leq j\leq p} \sqrt{n} \ (U_{n,N,j}^{\prime} - f_j(\theta)) / \hat{\sigma}_j \leq t)$ for $t \in \mathbb{R}$ corresponds to the probabilities $P(\sqrt{n} \, \widehat{\Lambda}^{-1/2}(U_{n,N}^{\prime} - f(\theta)) \in R_t)$, where
$
    R_t = \{x \in \mathbb{R}^p: - \infty \leq x_j \leq t \textrm{ for all } j=1, \ldots, p\}
$
are $t$-dependent hyperrectangles. Hence, by Corollary~\ref{cor:studentization}, the maximum $\max_{1 \leq j \leq p} \sqrt{n} (U_{n,N,j}^{\prime} - f_j(\theta)) / \hat{\sigma}_j$ can be approximated by the maximum of the studentized Gaussian multiplier statistic $W \defeq \max_{1 \leq j \leq p} U_{n,n_1,j}^\# / \hat{\sigma}_j$. In particular, the quantiles of $W$ approximate the quantiles of  $\max_{1 \leq j\leq p} \sqrt{n} \ (U_{n,N,j}^{\prime} - f_j(\theta)) / \hat{\sigma}_j$ as we verify in the next corollary. For $\alpha \in (0,1)$ we denote by $c_W(1-\alpha)$ the conditional $(1-\alpha)$-quantile of $W$ given the data $\mathcal{D}_n$. 
\begin{cor} \label{cor:quantiles}
Assume the conditions in Corollary~\ref{cor:all-together}. Then there exists a constant $C>0$ depending only on $\beta, \zeta$ and $C_1$ such that
\[
    \sup_{\alpha \in (0,1)}\left|P\left(\max_{1 \leq j\leq p} \sqrt{n} \ (U_{n,N,j}^{\prime} - f_j(\theta)) / \hat{\sigma}_j > c_W(1-\alpha) \right) - \alpha\right| \leq C n^{-\zeta/7}.
\]
\end{cor}

Corollary~\ref{cor:quantiles} says that, under mild regularity conditions involving mixed degeneracy, using $c_W(1-\alpha)$ as a critical value gives an asymptotically valid test, that is, it asymptotically controls type I error for the significance level $\alpha$. 
When the null hypothesis is only defined by equality constraints $f_j(\theta)=0$, then we have equality in ~\eqref{eq:critical-value} and  Corollary~\ref{cor:quantiles} implies that our test has asymptotic type I error exactly equal to the chosen level $\alpha$. When we have $f_j(\theta)<0$ for certain indices $j=1,\ldots, p$, then the  type I error is smaller or equal to the chosen level. 
As discussed in Remark~\ref{rem:irregular-setups}, mixed degeneracy also accommodates irregular settings. The computational budget parameter has to be chosen appropriately such that the conditions in Corollary~\ref{cor:all-together} are satisfied, typically $N=\mathcal{O}(n)$ is a good choice as we elaborate in Remark~\ref{rem:comp-budget}.  
In practice, we use the empirical version of $c_{W}(1-\alpha)$ as a critical value. It is obtained in the following procedure:
\vspace{0.25cm}
\begin{itemize}\setlength{\itemsep}{0.25cm} 
    \item[(i)] Generate many, say $A=1000$, sets of standard normal random variables $\{\xi_{\iota}^{\prime}: \iota \in I_{n,m}\} \cup \{\xi_{i_1}: i_1 \in S_1\}$.
    \item[(ii)] Evaluate $W$ for each of these $A$ sets.
    \item[(iii)] Take $\hat{c}_{W}(1-\alpha)$ to be the $(1-\alpha)$-quantile of the resulting $A$ numbers.
    \item[(iv)] Reject $H_0$ if $\mathcal{T} > \hat{c}_{W}(1-\alpha)$.
\end{itemize}
\vspace{0.25cm}

\begin{rem}[Practical considerations and complexity] \label{rem:complexity}
    There is no need to generate $|I_{n,m}| \approx n^m$ Bernoulli random variables to construct the incomplete $U$-statistic $U_{n,N}^{\prime}$, nor to generate the same amount of standard normal random variables $\xi_{\iota}^{\prime}$. As explained in \citet[Section 2.1]{chen2019randomized}, one can equivalently generate $\hat{N} \sim \textrm{Bin}(|I_{n,m}|, \rho_n)$ once and then choose indices $\iota_1, \ldots, \iota_{\widehat{N}}$ without replacement from $I_{n,m}$. Then we can compute the incomplete $U$-statistic as
    \[
        U_{n,N}^{\prime} = \frac{1}{\hat{N}} \sum_{i=1}^{\hat{N}} h(X_{\iota_i}).
    \]
    Similarly, one only needs to generate $\widehat{N}$ versions of $\xi_{\iota}^{\prime}$ to construct $U_{n,h}^\#$. As another practical guidance, we suggest to compute $n_1=|S_1|=n$ divide-and-conquer estimators. Choosing $n_1$ smaller only decreases the accuracy of the bootstrap, which is only useful if the computational cost is otherwise too high.
    With $n_1=n$, the overall computational cost of the proposed procedure is $\mathcal{O}(n^2p+ A(N+n)p)$  as the estimation step for $g$ can be done outside the bootstrap.
\end{rem}

\begin{rem}[Computational budget parameter $N$] \label{rem:comp-budget}
    
    If one is certain that the underlying parameter $\theta$ is regular and \emph{all} random variables $g_j(X_1)$, $j=1, \ldots, p$ are asymptotically non-degenerate, then 
    it may be appealing to choose the computational budget $0 < N < \binom{n}{m}$ arbitrarily since the bootstrap approximation still holds as shown in \citet[Theorem 3.1]{chen2019randomized} and further refined in \citet[Theorem 2.4]{song2019approximating}. However, the rate of convergence depends on $\ubar{\sigma}_{g, \theta}^2 = \min_{1 \leq j \leq p} \sigma_{g,\theta, j}^2$ as outlined in the introduction. Thus, if $\theta$ is close to an irregular point,  where $\ubar{\sigma}_{g, \theta}^2$ is small, the rate of convergence can be very slow. In contrast, a lower computational budget parameter $N$ may imply convergence  with a reasonable rate, even when the parameter $\theta$ is close to irregular. If mixed degeneracy is satisfied, the rate of convergence only depends on $\ubar{\sigma}_{g^{(1)}}^2$, which is a fixed constant large enough to achieve fast convergence; recall Remark~\ref{rem:irregular-setups} and Corollary \ref{cor:Gaussian-case}. \looseness=-1
    
    On the other hand, consider an underlying parameter $\theta$ such that \emph{all} variances $\sigma_{g,\theta, j}^2$ of the H\'ajek projection are zero, which in particular means that $\theta$ is an irregular point. Then, the incomplete $U$-statistic if fully degenerate and a bootstrap approximation also holds true for larger $N$ \citep[Theorem 3.3]{chen2019randomized}.  
    However, the limiting Gaussian distribution becomes $N_p(0, \Gamma_h)$ under a suitable scaling of the incomplete $U$-statistic. 
    Choosing the computational budget $N$ lower allows for computing the bootstrap based on the Gaussian approximation $N_p(0, m^2 \Gamma_g + \alpha_n \Gamma_h)$, regardless of any potential degeneracy. In particular, we can compute valid critical values without knowing a priori whether or not a point is irregular. Moreover, the degenerate approximation to $N_p(0, \Gamma_h)$ only holds if  \emph{all} variances $\sigma_{g,\theta, j}^2$ are zero and not for more general irregular points where only some of them are zero.
    
    To summarize, our approximation results of Section~\ref{sec:incomplete-U-stat} are very useful if it is unknown whether or not the true point is irregular or close to irregular. Mixed degeneracy allows for such setups when choosing the computational budget $N$ appropriately.
    In particular, we recommend to choose the computational budget parameter $N$ of the same order as the sample size $n$ to guard against irregularities. However, this still allows some flexibility in practice. In Section~\ref{sec:tree-models} we compare test size and power for different choices of $N$ in numerical experiments.
\end{rem}

\begin{rem}[High dimensionality]
    We emphasize that bootstrap approximation holds in settings where the number of polynomials $p$ may be much larger than the sample size $n$, i.e., $p$ may be as large as $\exp(n^c)$ for a constant $c > 0$. Thus, we can test a very large number of polynomial restrictions simultaneously. Moreover, we do not require any restriction on the correlation structure among the polynomials.
\end{rem}

\subsection{Power} \label{sec:power}
The following result pertains to the power of the proposed test.

\begin{prop} \label{prop:power}
Assume the conditions in Corollary~\ref{cor:all-together}, and let $\alpha \in (0,1/2)$. Then there exists a constant $C>0$ depending only on $\beta, \zeta$ and $C_1$ such that for every $\epsilon > 0$, whenever
\begin{equation} \label{eq:power-assumption}
    \max_{1 \leq j \leq p} (f_j(\theta)/\sigma_j) \geq (1+\varepsilon)\left(1 + C n^{-3\zeta/7} \log(p)^{-2} \right) \frac{\sqrt{2 \log(p)} + \sqrt{2 \log(1/\alpha)}}{\sqrt{n}},
\end{equation}
we have
\[
    P(\mathcal{T} > c_W(1-\alpha)) \geq 1 - \exp\left(- \frac{1}{C} \, \varepsilon^2 \log(p/\alpha)\right) - C n^{-\zeta/7}.
\]
\end{prop}
Proposition~\ref{prop:power} shows that our test is consistent against all alternatives where at least one constraint, normalized by its approximate variance, is violated by a small amount. In other words, the test is consistent against all alternatives  excluding the ones in a small neighborhood of the null set $\Theta_0$. The size of the neighborhood shrinks with rate $\sqrt{\log(p)/n}$ as long as $p \rightarrow \infty$ as $n \rightarrow \infty$. A similar result on power is derived in \citet{chernozhukov2019inference} for the special case on independent sums, and we refer to their discussion on the rate the neighborhood converges to zero.

\section{Polynomial Hypotheses} \label{sec:polynomial-hypotheses}

In this section, we assume that the constraints defining the null hypothesis $\Theta_0$ in~\eqref{eq:nullhypothesis} are polynomial. That is, each constraint $f_j \in \mathbb{R}[\theta_1, \ldots, \theta_d]$ is a polynomial in the indeterminates $\theta_1, \ldots, \theta_d$. We propose a general procedure to find a kernel $h$ such that $h_j(X_1^m)$ is an unbiased estimator of $f_j(\theta)$. For now, let $f \in \mathbb{R}[\theta_1, \ldots, \theta_d]$ be a \emph{single} polynomial of total degree $s$ and write
\[
    f(\theta) = a_0 + \sum_{r=1}^s \sum\limits_{\substack{(i_1, \ldots, i_r) \\i_{l} \in \{1, \ldots, d\}}} a_{(i_1, \ldots, i_r)} \theta_{i_1} \cdots \theta_{i_r}
\]
with $a_{(i_1, \ldots, i_r)} \in \mathbb{R}$ for all multi-indices $(i_1, \ldots, i_r)$. Note that this notation of multivariate polynomials is somewhat inefficient in comparison to the usual multi-index notation since there may appear various indices repeatedly in $(i_1, \ldots, i_r)$. However, the representation is useful to define an estimator of the polynomial. We construct $h(X_1^m)$ by the following three steps:
\begin{itemize}\setlength{\itemsep}{0.25cm} 
    \item[1)] For a fixed integer $\eta \geq 1$, find $\mathbb{R}$-valued functions $\widehat{\theta}_i$ such that $\widehat{\theta}_i(X_1^{\eta})$ is an unbiased estimator of $\theta_i$ for all $i=1, \ldots, d$. 
    \item[2)] Let $m =  \eta s$ be the order of the kernel, and define the $\mathbb{R}$-valued function $\breve{h}$ via
    \[
    x_1^m \,\, \mapsto \,\, a_0 + \sum_{r=1}^s \sum\limits_{\substack{(i_1, \ldots, i_r) \\i_{l} \in \{1, \ldots, d\}}} a_{(i_1, \ldots, i_r)} \widehat{\theta}_{i_1}(x_1^{\eta}) \widehat{\theta}_{i_2}(x_{\eta+1}^{2\eta}) \cdots \widehat{\theta}_{i_r}(x_{(r-1)\eta+1}^{r\eta}),
    \]
    where $x_k^l=(x_k, \ldots, x_l)$ for $k < l$. Note that $ \breve{h}(X_1^m)$ is an unbiased estimator of $f(\theta)$ by the linearity of the expectation and the independence of the samples $X_1, \ldots, X_m$.
    \item[3)] To get a symmetric kernel $h$, average over all permutations $\pi \in S_m$ of the set  $\{1, \ldots, m\}$, that is,
    \[
        h(x_1, \ldots, x_m) = \frac{1}{m!} \sum_{\pi \in S_m} \breve{h}(x_{\pi(1)}, \ldots, x_{\pi(m)}).
    \]
\end{itemize}
\vspace{0.25cm}

The construction works for any polynomial as long as we can construct an unbiased estimator of the parameter $\theta$ that only involves a small number $\eta$ of samples.  In this case polynomial hypotheses are estimable and our proposed testing methodology is applicable.

\begin{exmp} \label{ex:forming-tetrad-kernel}
The kernel for estimating the tetrad $f(\Sigma) = \sigma_{uv} \sigma_{wz} - \sigma_{uz} \sigma_{vw}$ in Example~\ref{exmp:tetrad-kernel} is also constructed by steps 1) to 3). The total degree of $f$ is $s=2$ and to estimate an entry $\sigma_{uv}$ in the covariance matrix $\Sigma$ we only need $\eta=1$ sample, i.e., an unbiased estimator of $\sigma_{uv}$ is given by $\widehat{\sigma}_{uv}(X_1) = X_{1u} X_{1v}$. Thus, the order of the kernel is $m=2$ and we obtain
\[
    \breve{h}(X_1, X_2) = X_{1u}X_{1v}X_{2w}X_{2z}-X_{1u}X_{1z}X_{2v}X_{2w}.
\]
The symmetric kernel is given by 
\begin{align*}
    h(X_1,X_2)= \frac{1}{2}\{&(X_{1u}X_{1v}X_{2w}X_{2z}-X_{1u}X_{1z}X_{2v}X_{2w}) \\
    &+ (X_{2u}X_{2v}X_{1w}X_{1z}-X_{2u}X_{2z}X_{1v}X_{1w})\}.
\end{align*}
\end{exmp}

The bootstrap approximation established in Section~\ref{sec:incomplete-U-stat} requires that the individual estimators $h_j(X_1^m)$ are sub-Weibull of order $\beta > 0$ for all $j=1,\ldots,p$; recall Condition~\ref{C2}. If one is able to check that  all estimators $\widehat{\theta}_i(X_1^{\eta})$, $i=1, \ldots, d$ are sub-Weibull of order $\gamma$, then we obtain by Lemma~\ref{lem:bound-polynom}, which we state in the Appendix, that the estimator $h_j(X_1, \ldots, X_m)$ is sub-Weibull of order $\beta = \gamma/s$, where $s$ is the total degree of the polynomial $f_j(\theta)$. For estimating tetrads, this is illustrated in the next example.

\begin{exmp} \label{ex:tetrad-sub-weibull}
For a Gaussian random vector $X_i \sim N_l(0, \Sigma)$, it is easy to check that each component $X_{iu}$, $u=1, \ldots, l$ is sub-Gaussian with $\|X_{iu}\|_{\psi_{2}} \leq \sqrt{8 \sigma_{uu} / 3}$. Therefore $\widehat{\sigma}_{uv}(X_1) = X_{1u} X_{1v}$ is sub-Exponential by Lemma~\ref{lem:product-weibull} and hence the kernel $h(X_1,X_2)$ for estimating a tetrad in Example~\ref{ex:forming-tetrad-kernel} is sub-Weibull of order $\beta = 1/2$.
\end{exmp}

Recall our testing problem~\ref{eq:nullhypothesis} where the null hypothesis $\Theta_0 \subseteq \Theta$ is defined by polynomial constraints. Suppose we have a kernel $h$ such that every point $\theta \in \Theta$ is regular with respect to $h$ and $\ubar{\sigma}_{g,\theta} = \min_{1 \leq j \leq p} \sigma_{g,\theta,j}^2$ is not too small, that is, every point is far away from irregular. Then it is favorable to consider critical values based on the Gaussian approximation for non-degenerate incomplete $U$-statistics as discussed in Remark~\ref{rem:comp-budget}. On the other hand, if all points are irregular with \emph{every} individual H\'ajek projection being degenerate, then one should consider critical values based on the limiting distribution of degenerate incomplete $U$-statistics. For the proposed kernel, we will show in our next result, that  parameter spaces typically contain a measure zero subset of irregular points.

\begin{prop} \label{prop:criterion-regular-points}

Let $f \in \mathbb{R}[\theta_1, \ldots, \theta_d]$ be a polynomial of total degree $s \geq 2$ and define the set of indices
$D(f) = \{j \in [d]: \theta_j$ appears in $f\}$. Further, define 
$
    \hat{\theta}_{i,l}(x)=\mathbb{E}[\hat{\theta}_{i}(X_1, \ldots, X_{l-1}, x, X_{l+1}, \ldots, X_{\eta})]
$ 
for all $i=1, \ldots, d$ and $l=1, \ldots, \eta$ and denote the random vector
\[
    \hat{\theta}_f(X_1) = \left(\sum_{l=1}^{\eta}\hat{\theta}_{i,l}(X_1)\right)_{i \in D(f)}.
\] 
If the covariance matrix $\Cov[\hat{\theta}_f(X_1)]$ is positive definite, then the irregular points $\theta \in \Theta \subseteq \mathbb{R}^d$ with respect to the kernel $h$
form a measure zero set with respect to the Lebesgue measure on $\mathbb{R}^d$. \looseness=-1
\end{prop}

Proposition~\ref{prop:criterion-regular-points} is a sufficient condition for checking that almost all points in the parameter space are regular. By inspecting the proof, we see that the set of irregular points is the zero set of a certain non-zero polynomial and thus has Lebesgue measure zero. However, it will typically be nonempty. As we emphasized in the introduction, the Gaussian approximation of complete $U$-statistics, or incomplete $U$-statistics with high computational budget, may already be very slow if the true parameter is close to an irregular point. Since the null hypothesis $\Theta_0$ can be defined by a very large amount of constraints $p$, the sets of irregular points is the union of irregular points given by $p$ sets with Lebesgue measure zero. In other words, there might be many hypersurfaces with irregular points, such that it is likely that the true parameter is close to irregular point.  Thus, assuming mixed degeneracy is crucial for considering kernels as constructed above. 

\begin{exmp}
Consider again a tetrad $f(\Sigma) = \sigma_{uv} \sigma_{wz} - \sigma_{uz} \sigma_{vw}$ and the corresponding kernel $h$ given in Example~\ref{ex:tetrad-sub-weibull}. Since $\hat{\sigma}_{uv}(X_1) = X_{1u}X_{1v}$, we have that $\hat{\theta}_f(X_1) = (X_{1u}X_{1v}, X_{1w}X_{1z}, X_{1u}X_{1z}, X_{1v}X_{1w})^{\top}$. For Gaussian $X_1 \sim N_l(0,\Sigma)$, it is easy to see that $\Cov[\hat{\theta}_f(X_1)]$ is positive definite since no coordinate of  $\hat{\theta}_f(X_1)$ is a linear function of the remaining coordinates. It follows by Proposition~\ref{prop:criterion-regular-points} that almost all covariance matrices $\Sigma$ are regular points in the parameter space given by the cone of positive definite matrices. Covariance matrices with some  entries equal to zero correspond to irregular points as we have seen in Example~\ref{ex:tetrad-degeneracy}. Thus, all covariance matrices with those entries having small absolute values are close to irregular points.
\end{exmp}

\section{Testing Gaussian Latent Tree Models}  \label{sec:tree-models}
In this section we apply our test for assessing the goodness-of-fit of Gaussian latent tree models. These models are of particular relevance in phylogenetics \citep{semple2003phylogenetics, zwiernik2016semialgebraic} and the problem of model selection is, for example, considered in \citet{shiers2016thecorrelation} and \citet{leung2018algebraic}; for a survey see \citet{sung2009algorithms} and \citet{junker2011analysis}. Our methodology is applicable since a full semialgebraic description of the model is known, that is, all polynomial equalities and inequalities that fully describe the distributions corresponding to a given tree are known. The constraints include tetrads as well as higher-order polynomial constraints. However, as we also point out in the introduction, it is challenging to test the large number of constraints simultaneously. For example, in \citet{shiers2016thecorrelation} only small trees and a subset of constraints is tested. Typical approaches such as the Wald test can only handle $p \leq n$ constraints and the maximum likelihood function is difficult to optimize.  For the latter, we implemented an expectation-maximization algorithm, but it is unclear whether it obtains the global maximum; see Section \ref{sec:simulations}. Moreover, \citet{drton2009likelihood}, \citet{dufour2013wald}, and \citet{drton2016wald} show that irregular points lead to different limiting distributions in Wald and likelihood ratio tests when testing tetrads. Hence, our testing strategy is an approach that targets both challenges in testing Gaussian latent tree models: a large number of constraints and irregular points in the null hypothesis.

\subsection{Model and Constraints} 
We begin by briefly introducing Gaussian latent tree models. Let $T=(V,E)$ be an undirected tree where $V$ is the set of nodes and $E$ is the set of edges. For $L  \subseteq V$ denoting the set of leaves, we assume that every inner node in $V \setminus L$ has minimal degree $3$. By \citet[Chapter 8]{zwiernik2016semialgebraic} we have the following parametric representation of Gaussian latent tree models. Let $PD(l)$ be the cone of symmetric positive definite $l \times l$ matrices. \looseness=-1

\begin{prop} \label{prop:GLTM}
The Gaussian latent tree model of a tree $T=(V,E)$ with leaves $L = \{1, \ldots, l\} \subseteq V$ is the family of Gaussian distributions $N_l(0, \Sigma)$ 
such that $\Sigma$ is in the set
$$
\mathcal{M}(T) = \Big\{ \Sigma = (\sigma_{uv}) \in PD(l) : \sigma_{uv} = \sqrt{\omega_u} \sqrt{\omega_v} \!\! \prod_{e \in ph_T(u,v)} \!\!\! \rho_e
\ \textrm{ for  } v \neq u, \ \sigma_{vv} = \omega_v \Big\},
$$
where $ph_T(u,v)$ denotes the set of \emph{edges} on the unique path from $u$ to $v$, the vector $(\omega_v)_{v \in L} \in \mathbb{R}^l_{>0}$ contains the variances and $(\rho_e)_{e \in E}$ is an $|E|$-dimensional vector with $|\rho_e| \in (0,1)$.
\end{prop}

For details on how this parametrization arises from the paradigm of graphical modeling, see \citet{zwiernik2016semialgebraic} or \citet{drton2017marginal}. We may identify a Gaussian latent tree model with its set of covariance matrices $\mathcal{M}(T)$ and we will simply refer to $\mathcal{M}(T)$ as the model.

From a geometric point of view, the parameter space, i.e., the set of covariance matrices $\mathcal{M}(T)$, is fully understood. It is a semialgebraic set described by polynomial equalities and inequalities in the entries of the covariance matrix. To state the polynomial constraints we need the following notation borrowed from \citet{leung2018algebraic}. Recall that in a tree any two nodes are connected by precisely one path. For four distinct leaves $\{u,v,w,z\} \subseteq L$, there are three possible partitions into two subsets of equal size, namely $\{u,v\}|\{w,z\}$, $\{u,w\}|\{v,z\}$ and $\{u,z\}|\{v,w\}$. These three partitions correspond to the three intersection of path pairs 
\begin{equation} \label{eq:path-pairs}
ph_T(u,v) \cap ph_T(w,z), \,\, ph_T(u,w) \cap ph_T(v,z) \, \textrm{ and } \, ph_T(u,z) \cap ph_T(v,w).
\end{equation}
By the structure of a tree, either all of the intersections give an empty set or exactly one is empty while the others are not; cf. \cite{leung2018algebraic}. We let $\mathcal{Q} \subseteq \{\{u,v,w,z\} \subseteq L\}$ be the collections of subsets of size four such that $\{u,v,w,z\} \in \mathcal{Q}$ if exactly one of the intersections is empty. Given $\{u,v,w,z\} \in \mathcal{Q}$ we write $\{u,v\}|\{w,z\} \in \mathcal{Q}$ to indicate that $\{u,v,w,z\}$ belongs to $\mathcal{Q}$ and the paths $ph_T(u,v)$ and $ph_T(w,z)$ have empty intersection. 

\begin{prop}[\citeauthor{leung2018algebraic}, \citeyear{leung2018algebraic}] \label{prop:constraints}
Let $\Sigma=(\sigma_{uv}) \in PD(l)$ be a covariance matrix with no zero entries. Then $\Sigma \in \mathcal{M}(T)$ if and only if $\Sigma$ satisfies the following constraints:
\begin{itemize}
    \item[(1)] \textit{Inequality constraints:}
    \begin{itemize}
        \item[(a)] For any $\{u,v,w\} \subseteq L$, 
        $-\sigma_{uv}\sigma_{uw}\sigma_{vw} \leq 0$.
        \item[(b)] For any $\{u,v,w\} \subseteq L$, 
        $\sigma^{2}_{uv}\sigma^{2}_{vw} - \sigma^{2}_{vv}\sigma^{2}_{uw} \leq 0, \sigma^{2}_{uw}\sigma^{2}_{vw} - \sigma^{2}_{ww}\sigma^{2}_{uv} \leq 0 \\ \textrm{and } \sigma^{2}_{uv}\sigma^{2}_{uw} - \sigma^{2}_{uu}\sigma^{2}_{vw} \leq 0.$
        \item[(c)] For any $\{u,v\}|\{w,z\} \in \mathcal{Q}$, $\sigma^{2}_{uw}\sigma^{2}_{vz} - \sigma^{2}_{uv}\sigma^{2}_{wz} \leq 0.$
    \end{itemize}
    \item[(2)] \textit{Equality constraints (tetrads):}
    \begin{itemize}
        \item[(a)] For any $\{u,v\}|\{w,z\} \in \mathcal{Q}$, $\sigma_{uw}\sigma_{vz} - \sigma_{uz}\sigma_{vw} = 0.$
        \item[(b)]  For any $\{u,v,w,z\} \not\in \mathcal{Q}$, $\sigma_{uz}\sigma_{vw} - \sigma_{uw}\sigma_{vz} = \sigma_{uv}\sigma_{wz} - \sigma_{uw}\sigma_{vz} = 0.$
    \end{itemize}
\end{itemize}
\end{prop}

We want to test a Gaussian latent tree model against the saturated alternative, that is,
$$H_0 : \Sigma \in \mathcal{M}(T) \quad \textrm{ vs. } \quad H_1: \Sigma \in PD(l) \setminus \mathcal{M}(T)$$
based on i.i.d.~samples $X_1, \ldots, X_n$ taken from a Gaussian distribution $N_{l}(0, \Sigma)$.
With our methodology from Sections~\ref{sec:methodology} and~\ref{sec:polynomial-hypotheses}, we can simultaneously test all constraints defining $\mathcal{M}(T)$ given in Proposition~\ref{prop:constraints}. The number of constraints is $p=2\binom{l}{4}+4\binom{l}{3}$ which is of order $\mathcal{O}(l^4)$. Even for moderate values of $l$, say $l=15$, this is a large number of constraints since $p=4550$. 

\subsection{Simulations} \label{sec:simulations}
We have published an implementation of our tests for Gaussian latent tree models in the \texttt{R} package \texttt{TestGGM}, available at \url{https://github.com/NilsSturma/TestGGM}. In the implementation we use $A=1000$ sets of Gaussian multipliers to compute the critical value $\hat{c}_{W}(1-\alpha)$. The kernel for estimating the polynomials is constructed as demonstrated in Example~\ref{ex:forming-tetrad-kernel}, where we treated the special case of tetrads.   The package \texttt{TestGGM} also provides a routine for determining the specific constraints for a given latent-tree input, i.e., for determining the set $\mathcal{Q}$. For finding a path $ph_T(u,v)$ we use the function \texttt{shortest\_paths} in the  \texttt{igraph} package \citep{igraph}.  \looseness=-1

\begin{figure}[t]
\captionsetup[subfigure]{labelformat=empty}
\centering
\begin{subfigure}{.49\linewidth}
\centering
\begin{tikzpicture}[align=center, node distance=2cm]
\node[circle,draw=black, fill=lightgray,inner sep=3pt] (1) at (0,0) {};
\node[circle,fill,inner sep=3pt] (2) at (0,-0.9) {};
\node[circle,fill,inner sep=3pt]  (3) at (-0.7,-0.7) {};
\node[circle,fill,inner sep=3pt]  (4) at (-0.9,0) {};
\node[circle,fill,inner sep=3pt]  (5) at (-0.7,0.7) {};
\node[circle,fill,inner sep=3pt]  (6) at (0,0.9) {};
\node[circle,fill,inner sep=3pt]  (7) at (0.7,0.7) {};
\node[circle,fill,inner sep=3pt]  (8) at (0.9,0) {};
\node[circle,draw=black,dashed,inner sep=3pt]  (9) at (0.7,-0.7) {};
\draw (1) -- (2);
\draw (1) -- (3);
\draw (1) -- (4);
\draw (1) -- (5);
\draw (1) -- (6);
\draw (1) -- (7);
\draw (1) -- (8);
\draw[dashed] (1) -- (9);
\end{tikzpicture}
\caption{(i)}
\end{subfigure}
\hfill
\begin{subfigure}{.49\linewidth}
\centering
\vspace{0.7cm}
\begin{tikzpicture}[align=center, node distance=2cm]
\node[circle,draw=black, fill=lightgray,inner sep=3pt] (1) at (0,0) {};
\node[circle,draw=black, fill=lightgray,inner sep=3pt] (2) at (-0.7,0) {};
\node[circle,draw=black, fill=lightgray,inner sep=3pt] (3) at (0.7,0) {};
\node[circle,fill,inner sep=3pt] (4) at (-1.4,0) {};
\node[circle,fill,inner sep=3pt] (5) at (-0.7,-0.7) {};
\node[circle,fill,inner sep=3pt] (6) at (0,-0.7) {};
\node[circle,fill,inner sep=3pt] (7) at (0.7,-0.7) {};
\node[circle,fill,inner sep=3pt] (8) at (1.4,0) {};
\draw (1) -- (2);
\draw[dashed] (1) -- (3);
\draw (2) -- (4);
\draw (2) -- (5);
\draw (1) -- (6);
\draw (3) -- (7);
\draw (3) -- (8);
\end{tikzpicture}
\caption{(ii)}
\label{subfigure-2}
\end{subfigure}
\caption{Graphical representation of (i) the star tree and (ii) the binary caterpillar tree. Solid black dots correspond to leaves (observed variables).}
\label{fig:tree-structures}
\end{figure}

We compare our test with the likelihood ratio (LR) test. By \citet{shiers2016thecorrelation} the dimension of a Gaussian latent tree model is $|E| + |L|$. Thus, under regularity conditions, the likelihood ratio test statistic approximates a chi-square distribution with $\binom{l+1}{2} - (|E| + |L|)$ degrees of freedom for $l \geq 4$ \citep{van1998asymptotic}. For the one-factor model, the likelihood ratio test is implemented by the function \texttt{factanal} in the \texttt{base} library of \texttt{R} \citep{R}. Note that it employs a Bartlett correction for better asymptotic approximation. For general trees the implementation shows some of the challenges of the likelihood ratio test. Optimizing the likelihood function on the subspace $\mathcal{M}(T)$ is difficult using the common non-convex methods due to the large number of constraints defining the subspace and the required positive definiteness of the covariance matrices. Therefore we implement an expectation maximization algorithm which was first proposed by \citet{dempster1977maximum} and is described in \citet{friedman2002structural} and \citet{mourad2013survey} for the special case of latent tree models. A drawback of this method is that global optimization is not guaranteed. 

In our numerical experiments we consider two different tree structures. One is the star tree, where the model is equal to the one-factor analysis model. Second, we consider the binary caterpillar tree which is extreme in a sense that it has the longest paths possible between leaves. In Figure~\ref{fig:tree-structures} both structures are illustrated. We fix dimension $l=15$,  sample size $n=500$ and generate data from both tree structures in three experimental setups:

\begin{figure}[t]
\centering
\includegraphics[width=0.8\linewidth]{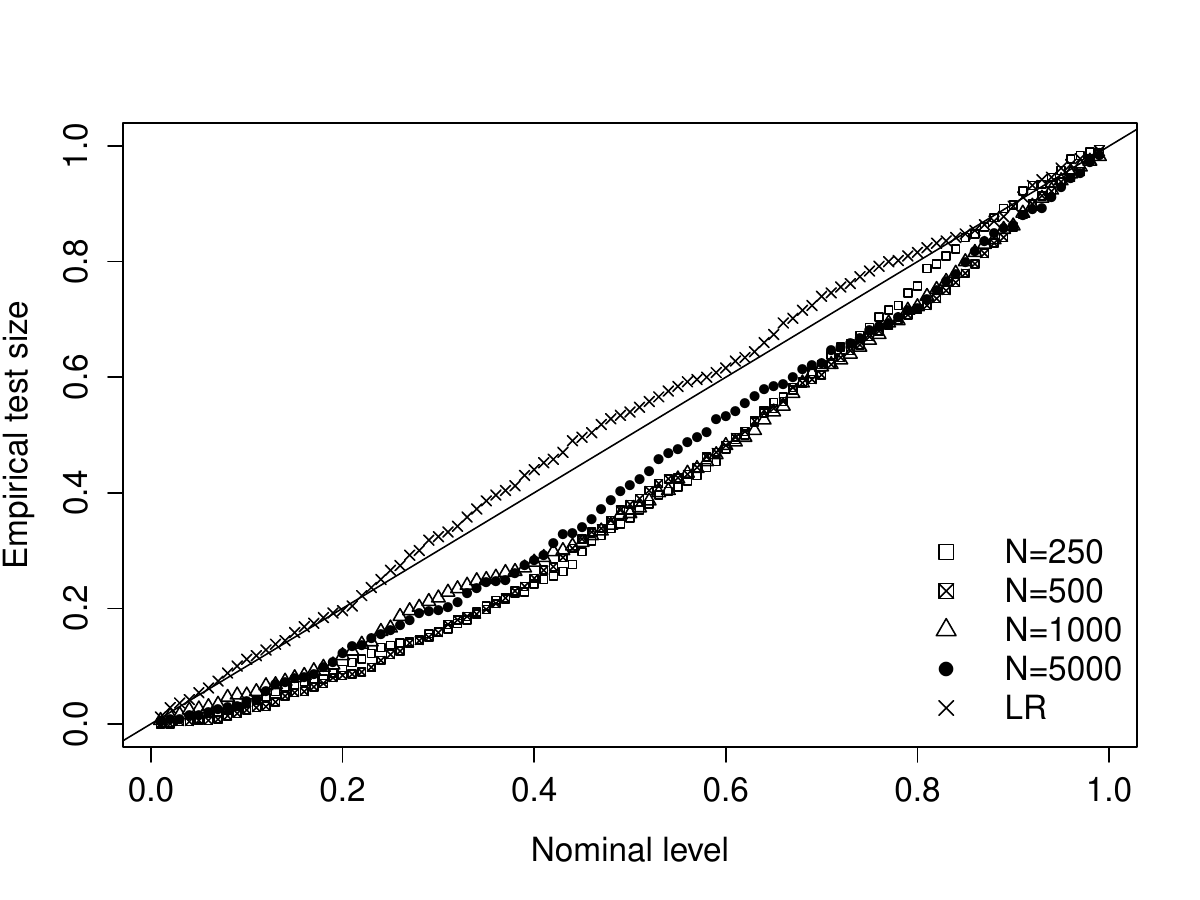}
\caption{Empirical sizes vs. nominal levels for testing tetrads based on $500$ experiments. The computational budget parameter $N$ is varied as indicated and empirical sizes of the LR test are also shown. Data is generated from setup (a) with $(l,n)=(15,500)$.}
\label{fig:sizes-regular}
\end{figure}

\begin{itemize}
    \item[(a)] Star tree: All edge parameters $\rho_e$ are equal to $\sqrt{0.5}$ while we take all variances $\omega_{v}=2$.
    \item[(b)] Star tree: Let $h$ be the unique, unobserved inner node. The parameters for the edges $\{h,1\}$ and $\{h,2\}$ are taken to be $0.998$ while all other edge parameters $\rho_e$  are independently generated from $N(0,0.1)$.
    The marginal variances are taken to be $\omega_v=100$ for $v=1,2$ and $\omega_v=1$ for all other leaves $v \in L$.
    \item[(c)] Caterpillar tree: The edge parameters $\rho_e$ are taken to be $0.998$ except for all edges incident to $4$ selected inner nodes in $V\setminus L$, where the edge parameters are independently generated from $N(0,0.1)$.
    All variances $\omega_{v}$ are taken to be $2$.
\end{itemize}

Setup (a) is a regular problem where the LR test is the gold standard. In contrast, setups (b) and (c) are designed to be irregular problems with small covariances; compare with Example~\ref{ex:tetrad-degeneracy} in the introduction. In setup (b) parameters are chosen such that the covariance matrix has exactly one off-diagonal entry which is far away from zero while all the remaining off-diagonal entries are close to zero. In this case the parameters are also close to an algebraic singularity of the parameter space of the star tree \citep[Proposition 32]{drton2007algebraic}. The LR test should fail since  the likelihood ratio test statistic does not follow a chi-square distribution at singularities \citep{drton2009likelihood}. Similarly, parameters in setup (c) are close to an algebraic singularity since some of the edge parameters are almost zero \citep{drton2017marginal} and thus we also expect the LR test to fail. Moreover, setup (c) is considered to emphasize that the proposed methodology allows for testing arbitrary Gaussian latent tree models.

\begin{figure}[p]
\captionsetup[subfigure]{labelformat=empty}
\centering
\begin{subfigure}{0.8\linewidth}
\centering
\includegraphics[width=\linewidth]{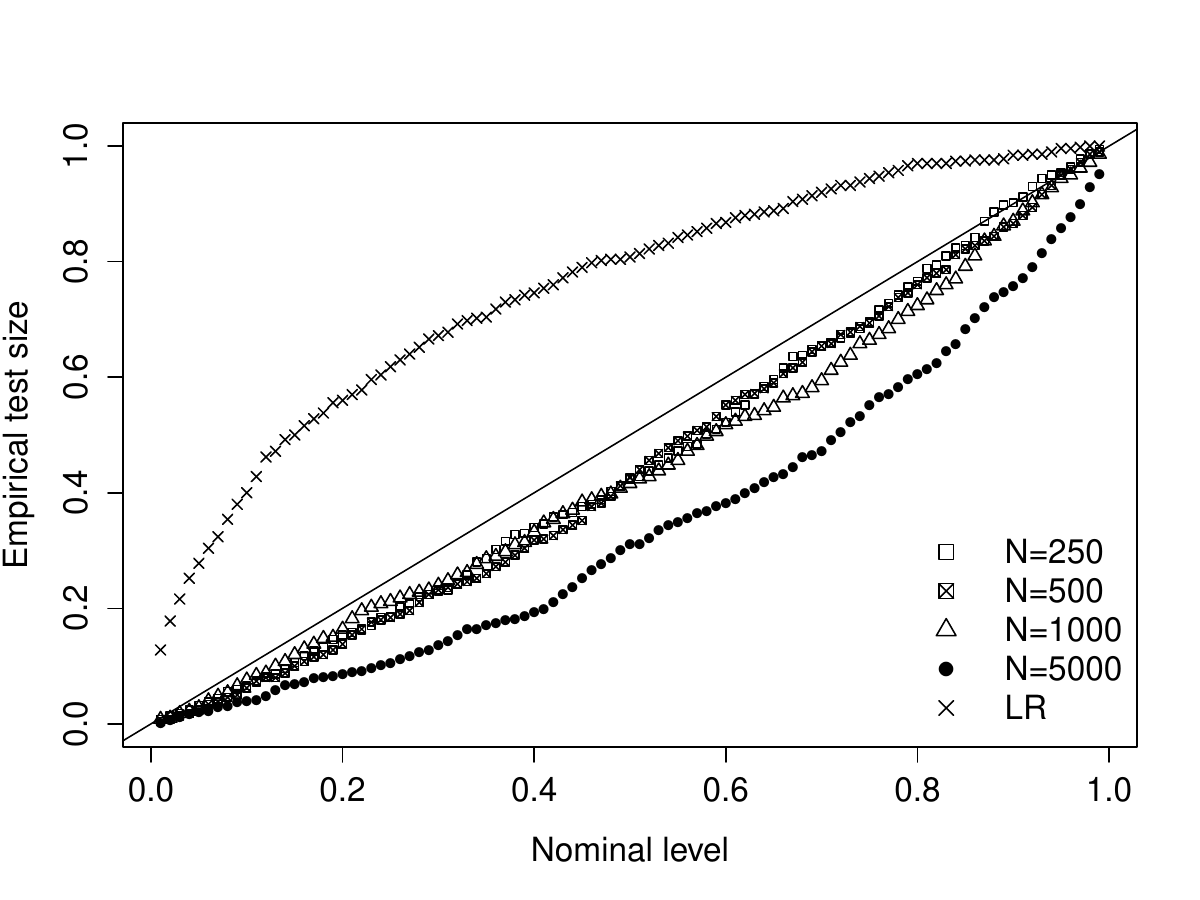}
\caption{Setup (b)}
\end{subfigure}
\hfill
\begin{subfigure}{0.8\linewidth}
\centering
\includegraphics[width=\linewidth]{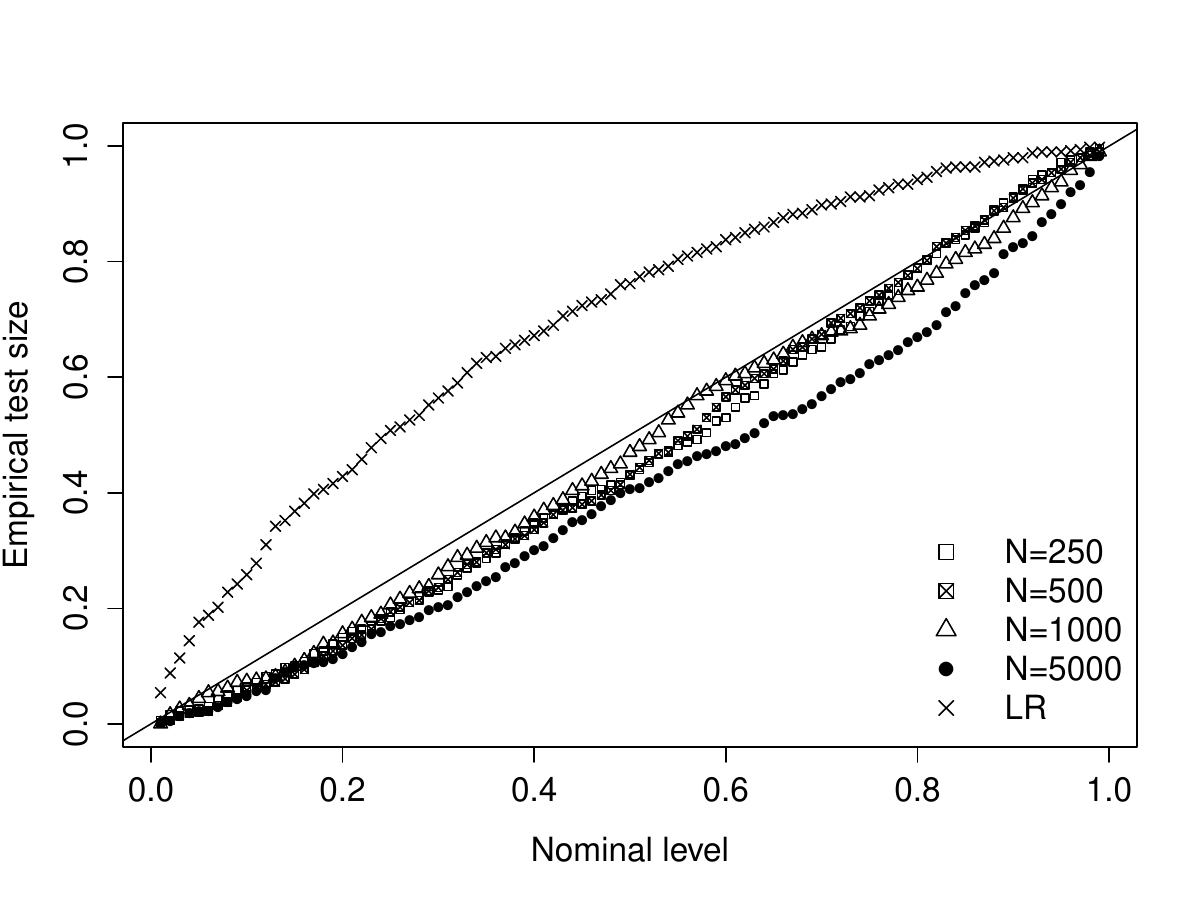}
\caption{Setup (c)}
\end{subfigure}
\caption{Empirical sizes vs. nominal levels for testing tetrads based on $500$ experiments. The computational budget parameter $N$ is varied as indicated and empirical sizes of the LR test are also shown. Data is generated from setups (b) and (c) with $(l,n)=(15,500)$. }
\label{fig:sizes-singular}
\end{figure}

In Figures~\ref{fig:sizes-regular} and~\ref{fig:sizes-singular} we compare empirical test sizes of our test with the LR test for different, fixed significance levels $\alpha \in (0,1)$ that we call the ``nominal level''. In addition, we consider different computational budget parameters $N$ to see its influence on the behaviour of our test. To better compare test sizes we only consider testing the equality constraints (tetrads) in Proposition~\ref{prop:constraints}. This simplifies our analysis since in this case we can expect that the test size is equal to the level $\alpha$. A statistical test that controls type I error is expected to have empirical test sizes below the 45 degree line in the plots. The lower the line, the more ``conservative'' the test is. On the other hand, a test does not control type I error if the empirical sizes are above the 45 degree line.  In this case the true null hypothesis is rejected too often.

\begin{figure}[t]
\centering
\includegraphics[width=0.8\linewidth]{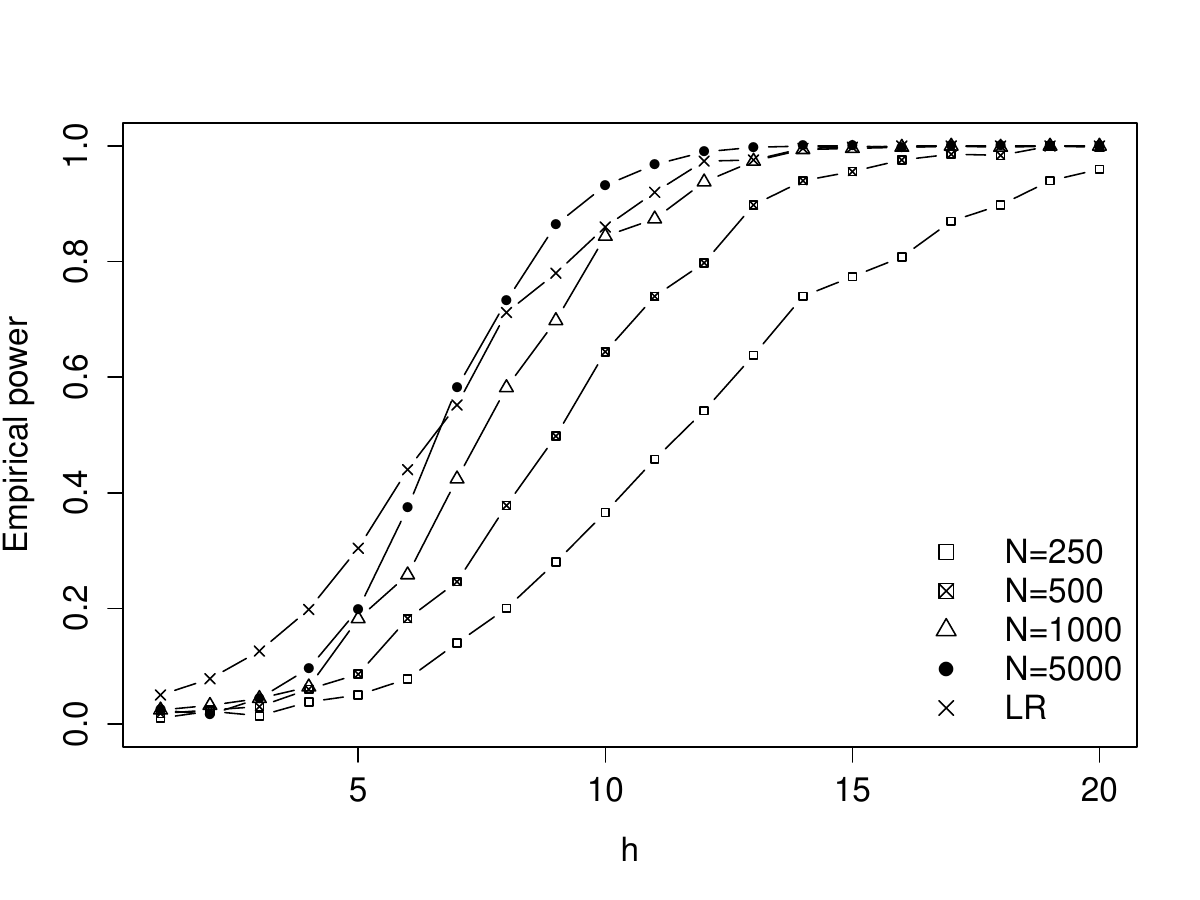}
\caption{Empirical power for different local alternatives based on $500$ experiments. The computational budget parameter $N$ is varied as indicated and empirical power of the LR test is also shown. Local alternatives are generated as described in the text for setup (a) with $(l,n)=(15,500)$ and level $\alpha=0.05$.}
\label{fig:power-setup1}
\end{figure}

The irregular setups (b) and (c) show the advantage of the proposed testing method because the empirical test sizes of the likelihood ratio test do not hold the nominal levels, that is, the empirical sizes are above the 45 degree line. In comparison, our testing method holds level for all choices of the computational budget parameter $N$. However, there is an interesting difference in the behaviour of the empirical test sizes between the regular setup (a) and setups (b) and (c) where parameters are close to an irregular point. In the regular setup, a higher computational budget parameter $N$ tends to yield empirical sizes closer to the nominal level $\alpha$ while in the irregular setups this is the other way round. The behaviour coincides with the theoretical guarantees discussed in Section~\ref{sec:incomplete-U-stat} and Remark~\ref{rem:comp-budget}, suggesting that the bootstrap approximation holds under mixed degeneracy only when choosing the computational budget $N$ appropriately, in this case $N=\mathcal{O}(n)$.

Besides size we also study the empirical power of our test.  For $\Sigma \in \mathcal{M}(T)$ we consider local alternatives 
\begin{equation} \label{eq:power}
\widetilde{\Sigma} = \Sigma + \gamma \gamma^{\top} \frac{h}{\sqrt{n}},
\end{equation}
where $\gamma = (0, \ldots, 0, 1, 1) \in \mathbb{R}^l$ is fixed and $h>0$ varies. According to Proposition~\ref{prop:power}, our test is consistent against alternatives outside of a neighborhood of size $\sqrt{\log(p)/n}$ up to an unknown constant. Thus, it is natural to consider alternatives as in~\eqref{eq:power} that also depend on the sample size through $\sqrt{n}$. Moreover, all entries of the alternative covariance matrix $\widetilde{\Sigma}$ are equal to the entries of $\Sigma$ except for four entries, which implies that only few constraints in Proposition~\ref{prop:constraints} are violated. However, since our test statistic is the maximum of an incomplete $U$-statistic, we expect that our test is consistent even for such sparse signals.

For the regular setup (a), the values of $h$ are plotted against the empirical power in Figure~\ref{fig:power-setup1}. 
This time we considered all constraints in Proposition~\ref{prop:constraints}, equalities and inequalities. 
The power increases with $h$ for all choices of the computational budget $N$. A higher budget empirically yields better power which is reasonable due to more precise estimation of the individual polynomials. For large computational budgets the power on the local alternatives is comparable to the likelihood ratio test. However, as discussed before, one should be sceptical about the bootstrap approximation for large computational budgets $N$ when the underlying true parameter is close to singular. The above experiments suggest that the choice $N=2n$ is a good trade-off between guarding against irregularities and statistical efficiency. A similar plot for the singular setups can be found in Appendix~\ref{sec:additional-simulations}. Moreover, we provide another simulation study in Appendix~\ref{sec:2-factor}, where we apply our methodology to test minors in the two-factor analysis model.

\subsection{Application to Gene Expression Data} 
In this section we analyze data from \citet{brawand2011theevolution}, where they obtained gene expression levels from sequencing of polyadenylated RNA from the cerebellum across different species such as gorillas, opossums or chickens. We focus on data of $20$ individuals from different species, where we have $5,636$ gene expression levels for each individual. Based on the assumption that species evolved from a common ancestor, a phylogenetic tree describes the evolutionary history among a collection of species \citep{semple2003phylogenetics}.  In the work of \citet{brawand2011theevolution} the authors reconstruct the phylogenetic tree in Figure~\ref{fig:brawand-tree} by neighbor-joining based on pairwise distance matrices. 
If we log-transform the original gene expression data, it is approximately normal distributed and it becomes a natural question to ask whether the data follows a Gaussian latent tree model where the underlying tree structure is the same. 
\looseness=-1

\begin{figure}[t]
\centering
\vspace{-1.5cm}
\includegraphics[angle=90,origin=c,width=0.6\linewidth]{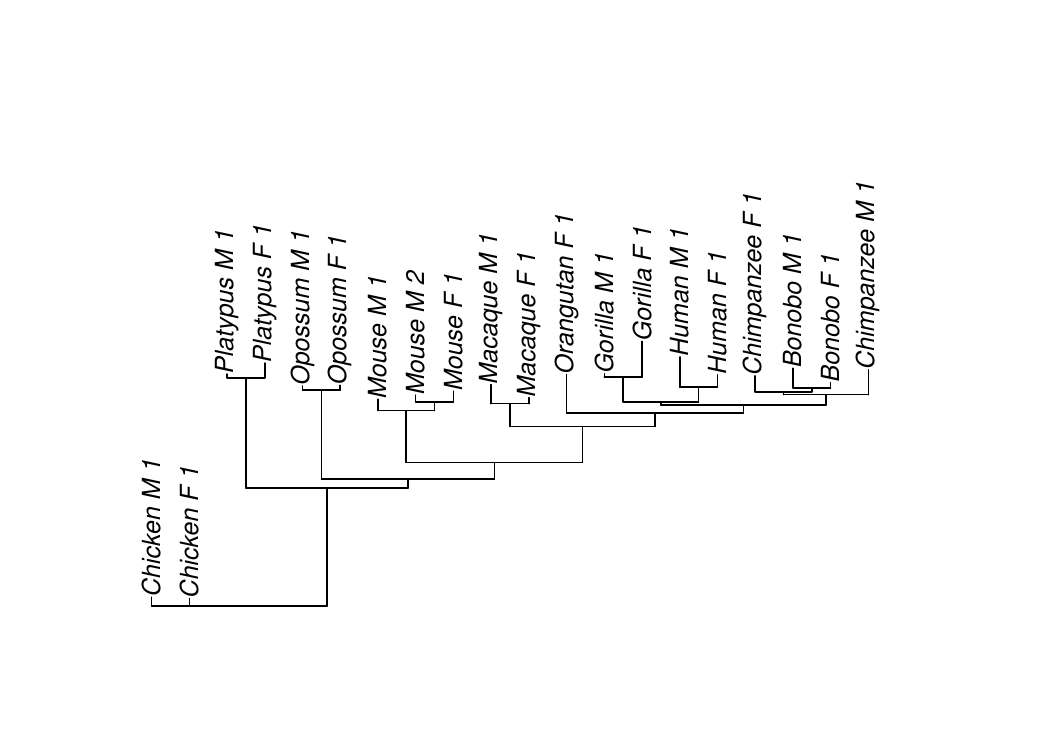}
\vspace{-2.2cm}
\caption{Mammalian gene expression phylogenies for cerebellum.}
\label{fig:brawand-tree}
\end{figure}

In our analysis to check tree compatibility we remove rows containing a zero such that we are left with $n=4615$ observations. Since the tree has $20$ leaves, we have to test $p=14250$ constraints that define the null hypothesis; recall proposition \ref{prop:constraints}. We compute $p$-values using the LR test and our test strategy with computational budget $N=2n$.  
Both tests indicate overwhelming evidence against the Gaussian latent tree model since we obtain a $p$-value of zero in both cases. 
To check whether the assumption of a Gaussian latent tree model might hold locally, we also run the same analysis on all substructures of the tree induced by $6$ neighboring leaves. This implies that $p=110$. Again, we only observe $p$-values equal to zero, indicating that the Gaussian latent tree assumption is also too strong for the considered subsets of variables.
We summarize that our approach is able to reach the same conclusion as the LR test, but our strategy is optimization-free and safe with respect to possible irregularities.

\section{Discussion} \label{discussion}
We proposed a general methodology to simultaneously test many constraints on a statistical parameter. The test statistic is given by a studentized maximum of an incomplete $U$-statistic and critical values are approximated by Gaussian multiplier bootstrap. Our method is applicable to all hypotheses that can be defined by estimable equality and inequality constraints. 
If the constraints are polynomial, we presented a general procedure to construct kernels. 
In the suggested methodology, there is no need to maximize a possibly multimodal function as one would do for a likelihood ratio test. Moreover, our method allows for many constraints, that is, the number of constraints can be much larger than the sample size. If the computational budget $N$ for constructing the incomplete $U$-statistic is appropriately chosen, typically of the same order as the sample size $n$, then the test asymptotically controls type I error even if the true parameter is irregular or close to irregular.

The latter fact is due to nonasymptotic Berry-Esseen type error bounds on the high-dimensional Gaussian and bootstrap approximation of incomplete U-statistics. Compared to previous work, we have shown that the approximation is also valid under mixed degeneracy if we choose a suitable computational budget $N$. This yields control of type I error as well as consistency of our testing methodology even in irregular settings. 

In practice, the requirement $N=\mathcal{O}(n)$ allows some flexibility in choosing the computational budget $N$. As illustrated in our simulations, a higher computational budget parameter yields more efficient estimates of the constraints and therefore a more powerful test. However, the computations become more involved for higher order kernels and one should be sceptical about the theoretical guarantees of the bootstrap approximation if $N$ is large and the underlying true parameter is close to an irregular point. Thus, there is a trade-off between efficiency and guarding against irregularities. In our simulations we found that choosing the computational budget parameter as $N=2n$ is a reasonable choice. 

If the inequality constraints defining the null hypothesis are strict, then there is a decrease in power. One may improve power by using a two-step approach for testing inequalities that was initially introduced in \citet{romano2014practical} and then further developed by \citet{bai2021two} for the high-dimensional setup. In the first step, a confidence region for the constraints of interest is constructed and in the second step, this set is used to provide information about which constraints are negative. In future work, it would be of interest to improve our strategy using the mentioned two-step approach.

\section*{Acknowledgments}
This project has received funding from the European Research Council (ERC) under the European Union’s Horizon 2020 research and innovation programme (grant agreement No 883818), and from the German Federal Ministry of Education and Research and the Bavarian State Ministry for Science and the Arts. Nils Sturma acknowledges support by the Munich Data Science Institute (MDSI) at Technical University of Munich (TUM) via the Linde/MDSI PhD Fellowship program.

\bibliographystyle{apalike}  
\bibliography{literature}

\newpage
\setcounter{page}{1}
\begin{center}
{\Large \textbf{Supplement to ``Testing Many Constraints in Possibly Irregular Models Using Incomplete U-Statistics''}}
\end{center}
\vspace{0.5cm}

This supplement contains additional material such as all technical proofs (Appendix~\ref{sec:proofs}), additional lemmas (Appendix~\ref{sec:useful-lemmas}), properties of sub-Weibull random variables (Appendix~\ref{sec:sub-Weibull}), additional simulation results for Gaussian latent tree models (Appendix~\ref{sec:additional-simulations}) and another application of our methodology, where we test minors in two-factor analysis models (Appendix~\ref{sec:2-factor}).

\begin{appendix}

\section{Proofs} \label{sec:proofs}
Since we only consider centered incomplete $U$-statistics, we assume without loss of generality $\mu = 0$ in all our proofs of the results in Section~\ref{sec:incomplete-U-stat}. 

\subsection{Proof of Theorem~\ref{thm:incompleteUstat}}
In this section we prove the main theorem. 
Let $R=[a,b] \in \mathbb{R}^p_{\mathrm{re}}$ be a hyperrectangle. For any partition $p=p_1 + p_2$ with $p_1, p_2 \geq 0$ , we write $R=(R^{(1)},R^{(2)})$, where $R^{(1)}$ and $R^{(2)}$ are hyperrectangles in $\mathbb{R}^{p_1}$ and $\mathbb{R}^{p_2}$. Moreover, for a given hyperrectangle and $t > 0$, we denote the following $t$-dependent quantities
\[
    R_{t} = [a-t,b+t], \quad R_{-t} = [a+t, b-t] \quad \text{and} \quad R_{\textrm{env},t} = R_{t} \setminus  R_{-t} = [a-t, a+t) \cup (b-t, b+t].
\] 
We begin with stating an abstract bound on \emph{complete} $U$-statistics; recall the definition
$
    U_n = \frac{1}{|I_{n,m}|} \sum_{\iota \in I_{n,m}}  h(X_{\iota}).
$ Let $\Gamma_g^{(1)} = \textrm{Cov}[g^{(1)}(X_1)]$ be the covariance matrix of the first $p_1$ indices of $g(X_1)$ and denote the symmetric $p \times p$ block matrix
\[
    \overline{\Gamma}_g = \begin{pmatrix}
    \Gamma_g^{(1)} & 0\\
    0 & 0
    \end{pmatrix}.
\]

\begin{lem} \label{lem:completeUstat}
Assume~\ref{C2} and~\ref{C4} -~\ref{C6} hold and let $R \in \mathbb{R}^p_{\mathrm{re}}$. Then there is a constant  $C_{\beta} > 0$ only depending on $\beta$ such that
\begin{align*}
    |P(\sqrt{n}(U_n - \mu) \in R) - P(m  \overline{Y}_g \in R)|
    \leq C_{\beta} \left(\frac{m^2 D_n^2 \log(p_1n)^{1+6/\beta}}{(\ubar{\sigma}_{g^{(1)}}^2 \wedge 1) \, n}\right)^{1/6} + \mathds{1}_{\{0 \in R^{(2)}_{\textrm{env},t}\}}
\end{align*}
with $t=C_{\beta}n^{-\min\{k, 1/3\}} m^2 D_n \log(p_2 n)^{2/\beta}$ and $\overline{Y}_g \sim N_p(0, \overline{\Gamma}_g)$.
\end{lem}

The bound in Lemma~\ref{lem:completeUstat} is not uniform for all hyperrectangles and is therefore of limited use on its own. However, it will be the main tool to prove Theorem~\ref{thm:incompleteUstat}.

\begin{proof}[Proof of Lemma~\ref{lem:completeUstat}]
Let $C_{\beta}$ be a constant that depends only on $\beta$ and that varies its value from place to place. For any other constant $c>0$ we assume without loss of generality 
\begin{equation} \label{eq:extra-assumption}
\frac{m^2\log(p_2)}{n} \leq c
\end{equation}
since otherwise the statement from Lemma~\ref{lem:completeUstat} becomes trivial by choosing $C_{\beta}$ large enough. First, we denote preliminary observations which we will use throughout the proof. Let $X$ be a random vector in $\mathbb{R}^{p_1}$ and $c$ be a non-random vector in $\mathbb{R}^{p_2}$.
\begin{enumerate}[label=(\roman*)]
    \item \label{obs1} $P((X,c) \in R) = P(\{X \in R^{(1)}\} \cap \{c \in R^{(2)}\}) = P(X \in R^{(1)}) \mathds{1}_{\{c \in R^{(2)}\}}.$ 
    \item \label{obs2} Let $Y$ be another random vector in $\mathbb{R}^{p_1}$. Then
        \begin{align*} 
        |P((Y,c) \in R) - P((X,c) \in R)| 
        &= |P(Y \in R^{(1)}) - P(X \in R^{(1)})| \mathds{1}_{\{c \in R^{(2)}\}} \\
        & \leq |P(Y \in R^{(1)}) - P(X \in R^{(1)})|.
        \end{align*}
    \item  \label{obs3} Suppose we are given two hyperrectangles $R=(R^{(1)},R^{(2)})$ and  $\tilde{R}=(R^{(1)},\tilde{R}^{(2)})$ having identical first part $R^{(1)}$. Then 
        \begin{align*}
            |P((X,c)\in\tilde{R}) - P((X,c)\in R)| 
            &= |P(X \in R^{(1)}) \mathds{1}_{\{c \in \tilde{R}^{(2)}\}} - P(X \in R^{(1)}) \mathds{1}_{\{c \in R^{(2)}\}}| \\
            &= P(X \in R^{(1)}) |\mathds{1}_{\{c \in \tilde{R}^{(2)}\}} - \mathds{1}_{\{c \in R^{(2)}\}}| \\
            &\leq |\mathds{1}_{\{c \in \tilde{R}^{(2)}\}} - \mathds{1}_{\{c \in R^{(2)}\}}|.
        \end{align*}
\end{enumerate}
Now, let $A_n = \sqrt{n} \, U_n$. Recall the notation of splitting the vector $A_n=(A_n^{(1)}, A_n^{(2)})$ according to the partition $p=p_1 + p_2$. Let $R=[a,b] \in \mathbb{R}^p_{\mathrm{re}}$ be a hyperrectangle and denote
\[
    \omega_n = \left(\frac{m^2 D_n^2 \log(p_1 n)^{1+6/\beta}}{(\ubar{\sigma}_{g^{(1)}}^2 \wedge 1) \, n}\right)^{1/6}.
\]
For any $t > 0$ we have
\begin{align*}
    P(A_n \in R)
    &\leq P(\{A_n \in R\} \cap \{\|A^{(2)}_n - 0\|_{\infty} \leq t\}) + P(\|A^{(2)}_n - 0\|_{\infty} > t) \\
    &\leq P(\{A_n^{(1)} \in R^{(1)}\} \cap \{0 \in R^{(2)}_t\}) + P(\|A^{(2)}_n - 0\|_{\infty} > t) \\
    &\leq P(\{mY_g^{(1)} \in R^{(1)}\} \cap \{0 \in R^{(2)}_t\}) + C_{\beta} \omega_n + P(\|A^{(2)}_n\|_{\infty} > t) \\
    &\leq P(m(Y_g^{(1)},0) \in R) + \mathds{1}_{\{0 \in R^{(2)}_{t} \setminus R^{(2)}\}} + C_{\beta} \omega_n + P(\|A^{(2)}_n\|_{\infty} > t),
\end{align*}
where the second to last inequality follows from observation~\ref{obs2} and Lemma~\ref{lem:CLT-complete-Ustat} that is applicable due to Conditions~\ref{C2},~\ref{C4} and~\ref{C6}. The last inequality follows from observation~\ref{obs3}. Observe that $(Y_g^{(1)},0) \sim \overline{Y}_g$ and $\mathds{1}_{\{0 \in R^{(2)}_{t} \setminus R^{(2)}\}} \leq \mathds{1}_{\{0 \in R^{(2)}_{\mathrm{env}, t}\}}$. Thus, it is left to bound  $P(\|A^{(2)}_n\|_{\infty} > t)$. We have
\begin{align*}
    P(\|A^{(2)}_n\|_{\infty} > t)
    &\leq P\left(\left\|A^{(2)}_n - \frac{m}{\sqrt{n}} \sum_{i=1}^n g^{(2)}(X_i)\right\|_{\infty} + \left\| \frac{m}{\sqrt{n}} \sum_{i=1}^n g^{(2)}(X_i)\right\|_{\infty} > t\right) \\
    &\leq P\left(\left\|A^{(2)}_n - \frac{m}{\sqrt{n}} \sum_{i=1}^n g^{(2)}(X_i)\right\|_{\infty} > \frac{t}{2}\right) + P\left(\left\|\frac{m}{\sqrt{n}} \sum_{i=1}^n g^{(2)}(X_i)\right\|_{\infty} > \frac{t}{2}\right).
\end{align*}
Due to~\eqref{eq:extra-assumption} we may apply  Theorem 5.1 in \citet{song2019approximating} to bound the first summand. Together with the Markov inequality we obtain
\begin{align}
     P\left(\left\|A^{(2)}_n - \frac{m}{\sqrt{n}} \sum_{i=1}^n g^{(2)}(X_i)\right\|_{\infty} > \frac{t}{2}\right)
     &= P\left(\left\|\frac{A^{(2)}_n}{\sqrt{n}} - \frac{m}{n} \sum_{i=1}^n g^{(2)}(X_i)\right\|_{\infty} > \frac{t}{2\sqrt{n}}\right) \nonumber \\ 
     &\leq C_{\beta} \frac{m^2 D_n \log(p_2)^{1+1/\beta}}{t \sqrt{n}}. \label{eq:first-summand}
\end{align}
To bound the second summand we apply Lemma A.2 in \citet{song2019approximating} which yields
\[
    P\left(\left\|\sum_{i=1}^n g^{(2)}(X_i)\right\|_{\infty} >  C_{\beta}(\sqrt{a_n} \log(p_2 n)^{1/2} + u_n \log(p_2 n)^{2/\beta} )\right) \leq \frac{4}{n},
\]
where $a_n = n \max_{p_1 + 1 \leq j\leq p} \|g_j(X_1)\|_2^2$ and $u_n = \max_{p_1 + 1 \leq j\leq p} \|g_j(X_)\|_{\psi_{\beta}}$. By Assumption~\ref{C5} we have $u_n \leq n^{-k} D_n$ and due to Lemma~\ref{lem:weibull-moments} it holds $a_n \leq C_{\beta} n \|g_j(X_i)\|^2_{\psi_{\beta}} \leq C_{\beta} n^{1-2k} D_n^2$. Hence,
\[
    P\left(\left\|\frac{m}{\sqrt{n}} \sum_{i=1}^n g^{(2)}(X_i)\right\|_{\infty} >  C_{\beta}\left(\frac{m D_n \log(p_2 n)^{1/2}}{n^k} + \frac{m D_n \log(p_2 n)^{2/\beta}}{n^{k+1/2}}\right)\right) \leq \frac{4}{n}.
\]
Now, choosing $t=C_{\beta}n^{-\min\{k, 1/3\}} m^2 D_n \log(p_2 n)^{2/\beta}$ we get
\[
    P\left(\left\|\frac{m}{\sqrt{n}} \sum_{i=1}^n g^{(2)}(X_i)\right\|_{\infty} > \frac{t}{2}\right) \leq \frac{4}{n} \leq \omega_n
\]
and due to~\eqref{eq:first-summand} we have
\[
    P\left(\left\|A^{(2)}_n - \frac{m}{\sqrt{n}} \sum_{i=1}^n g^{(2)}(X_i)\right\|_{\infty} > \frac{t}{2}\right) \leq C_{\beta} \frac{n^{\min\{k, 1/3\}}}{n^{1/2}} \leq C_{\beta}  n^{-1/6} \leq C_{\beta} \omega_n.
\]
Therefore, we conclude that $P(\|A^{(2)}_n\|_{\infty} > t) \leq C_{\beta} \omega_n$ and we have shown that
\[
    P(A_n \in R) \leq P(m \overline{Y}_g \in R) + C_{\beta} \omega_n + \mathds{1}_{\{0 \in R^{(2)}_{\mathrm{env}, t}\}}.
\]
Likewise, for the other direction it holds
\begin{align*}
    &P(A_n \in R) + P(\|A_n^{(2)} - 0\|_{\infty} > t)  = P(\{A_n^{(1)} \in R^{(1)}\} \cap \{A_n^{(2)} \in R^{(2)}\}) + P(\|A_n^{(2)} - 0\|_{\infty} > t) \\
    & \geq P(\{A_n^{(1)} \in R^{(1)}\} \cap [\{A_n^{(2)} \in R^{(2)}\} \cup \{|A_n^{(2)} - 0\|_{\infty} > t\}]) 
     \geq P(\{A_n^{(1)} \in R^{(1)}\} \cap \{0 \in R^{(2)}_{-t}\}).
\end{align*}
Thus, we have by Lemma~\ref{lem:CLT-complete-Ustat} and observations~\ref{obs2} and~\ref{obs3},
\begin{align*}
     P(A_n \in R) &\geq P(\{A_n^{(1)} \in R^{(1)}\} \cap \{0 \in R^{(2)}_{-t}\}) - P(\|A_n^{(2)} - 0\|_{\infty} > t) \\
     &\geq P(\{m Y_g^{(1)} \in R^{(1)}\} \cap \{0 \in R^{(2)}_{-t}\}) - C_{\beta} \omega_n - P(\|A_n^{(2)}\|_{\infty} > t) \\
     &\geq  P(m(Y_g^{(1)},0) \in R) - \mathds{1}_{\{0 \in R^{(2)} \setminus R^{(2)}_{-t}\}} - C_{\beta} \omega_n - P(\|A_n^{(2)}\|_{\infty} > t).
\end{align*}
Since $\mathds{1}_{\{0 \in R^{(2)} \setminus R^{(2)}_{-t}\}} \leq \mathds{1}_{\{0 \in R^{(2)}_{\mathrm{env}, t}\}}$ and $P(\|A_n^{(2)}\|_{\infty} > t) \leq C_{\beta} \omega_n$, the proof is complete.
\end{proof}

If the bounds $a^{(2)},b^{(2)}$ of the hyperrectangle $R^{(2)}=[a^{(2)},b^{(2)}]$ in Lemma~\ref{lem:completeUstat} depend on a (Gaussian) random vector, then one may use anti-concentration inequalities to bound  $\mathbb{E}[ \mathds{1}_{\{0 \in R^{(2)}_{\textrm{env},t}\}}] = P(0 \in R^{(2)}_{\textrm{env},t})$. We use this strategy to prove Theorem~\ref{thm:incompleteUstat}. 
The proof is split into two further propositions, where we first prove the approximation $\sqrt{n}(U_{n,N}^{\prime}) - \mu \approx N_p(0, m^2, \overline{\Gamma}_g + \alpha_n \Gamma_h)$ and then n $N_p(0, m^2, \overline{\Gamma}_g + \alpha_n \Gamma_h) \approx N_p(0, m^2, \Gamma_g + \alpha_n \Gamma_h)$ on the hyperrectangles.

\begin{prop} \label{prop:1}
Assume~\ref{C1} -~\ref{C6} hold. Then there is a constant $C_{\beta} > 0$ only depending on $\beta$ such that
\begin{align} 
    \sup_{R \in \mathbb{R}^p_{\mathrm{re}}} &|P(\sqrt{n}(U_{n,N}^{\prime} - \mu) \in R) - P(\overline{Y} \in R)| \leq C_{\beta} \{\omega_{n,1} + \omega_{n,2} + \omega_{n,3}\}
\end{align}
with $\overline{Y} \sim N_p(0,m^2 \overline{\Gamma}_g + \alpha_n \Gamma_h)$. 
\end{prop}

\begin{proof}[Proof of Proposition~\ref{prop:1}]
Let $C_{\beta} > 0$ be a constant that depends only on $\beta$ and that varies its value from place to place. As in the proofs of Theorem 3.1 in \citet{chen2019randomized} and Theorem 2.4 in \citet{song2019approximating} we write
$
    U_{n,N}^{\prime} = (N/\widehat{N})(U_n + \sqrt{1-\rho_n} B_n)
$
with $B_n = \frac{1}{N} \sum_{\iota \in I_{n,r}} \frac{Z_{\iota}-\rho_n}{\sqrt{1-\rho_n}}h(X_{\iota})$. For any hyperrectangle $R=[a,b] \in \mathbb{R}^p_{\mathrm{re}}$ we have
\begin{align*}
    &P(\sqrt{n}(U_n + \sqrt{1-\rho_n} B_n) \in R) \\
    &= \mathbb{E}\left[P|_{X_1^n}\left(\sqrt{N}B_n \in \left(\frac{1}{\sqrt{\alpha_n (1-\rho_n)}}R - \sqrt{\frac{N}{1-\rho_n}}U_n\right)\right)\right]\\
    &\leq \mathbb{E}\left[P|_{X_1^n}\left(Y_h \in \left(\frac{1}{\sqrt{\alpha_n (1-\rho_n)}}R - \sqrt{\frac{N}{1-\rho_n}}U_n\right)\right)\right] + C_{\beta} \omega_{n,1}\\
    &= P(\sqrt{n}U_n \in [R - \sqrt{\alpha_n(1-\rho_n)} Y_h]) +  C_{\beta}\omega_{n,1}  \\
    &= \mathbb{E}[P|_{Y_h}(\sqrt{n}U_n \in [R - \sqrt{\alpha_n(1-\rho_n)} Y_h])] + C_{\beta}\omega_{n,1}  \\
    & \leq \mathbb{E}[P|_{Y_h}(m \overline{Y}_g \in [R - \sqrt{\alpha_n(1-\rho_n)} Y_h]) + \mathds{1}_{\{0 \in \tilde{R}^{(2)}_{\textrm{env},t}\}}] + C_{\beta}\omega_{n,1}  \\
    & = P(m \overline{Y}_g \in [R - \sqrt{\alpha_n(1-\rho_n)} Y_h]) + \mathbb{E}[\mathds{1}_{\{0 \in \tilde{R}^{(2)}_{\textrm{env},t}\}}] + C_{\beta} \omega_{n,1},
\end{align*}
where $\tilde{R} = R - \sqrt{\alpha_n(1-\rho_n)} Y_h$ and $t=C_{\beta}n^{-\min\{k, 1/3\}} m^2 D_n \log(p_2 n)^{2/\beta}$. The first inequality follows from Lemma~\ref{lem:CLT} since $B_n$ is an independent sum given the samples $X_1,\ldots, X_n$ and Conditions~\ref{C1} -~\ref{C3} are satisfied; for a detailed proof see \citet[Lemma A.14]{song2019approximating}. The second inequality follows from Lemma~\ref{lem:completeUstat}. To bound the expectation appearing on the right-hand side we use an anti-concentration inequality due to Nazarov. It is stated in Lemma A.1 in \cite{chernozhukov2017central} and a detailed proof is given in the note \citet{chernozhukov2017detailed}. Let $\tilde{a} = a -  \sqrt{\alpha_n(1-\rho_n)} Y_h$ and  $\tilde{b} = b -  \sqrt{\alpha_n(1-\rho_n)} Y_h$ be random vectors in $\mathbb{R}^p$ such that the rectangle $\tilde{R}$ is given by $\tilde{R}=[\tilde{a}, \tilde{b}]$. Then we have
\begin{align*}
    \mathbb{E}&[\mathds{1}_{\{0 \in\tilde{R}^{(2)}_{\textrm{env},t}\}}] = P(0 \in \tilde{R}^{(2)}_{\textrm{env},t}) = P(\{\tilde{a}^{(2)}-t \leq 0 < \tilde{a}^{(2)}+t\} \cup \{\tilde{b}^{(2)}-t < 0 \leq \tilde{b}^{(2)} + t\}) \\
    &\leq P(\tilde{a}^{(2)}-t \leq 0 < \tilde{a}^{(2)}+t) + P(\tilde{b}^{(2)}-t < 0 \leq \tilde{b}^{(2)} + t) \\
    &=P(a^{(2)}-t \leq \sqrt{\alpha_n(1-\rho_n)} Y_h^{(2)} < a^{(2)}+t)  + P(b^{(2)}-t < \sqrt{\alpha_n(1-\rho_n)} Y_h^{(2)} \leq b^{(2)}+t) \\
    &= P\Bigl((\frac{a^{(2)}-t}{\sqrt{\alpha_n(1-\rho_n)}} \leq  Y_h^{(2)} < \frac{a^{(2)}+t}{\sqrt{\alpha_n(1-\rho_n)}}\Bigr)  + P\Bigl(\frac{b^{(2)}-t}{\sqrt{\alpha_n(1-\rho_n)}} <  Y_h^{(2)} \leq \frac{b^{(2)}+t}{\sqrt{\alpha_n(1-\rho_n)}}\Bigr) \\
    &\leq C_{\beta} \ \frac{t}{\ubar{\sigma}_h \sqrt{\alpha_n(1-\rho_n)}}\sqrt{\log(p_2)} \leq C_{\beta} \ \frac{t}{\ubar{\sigma}_h \sqrt{\alpha_n}}\sqrt{\log(p_2)} \\
    &\leq C_{\beta} \ \frac{N^{1/2} m^2 D_n \log(p_2 n)^{1/2 + 2/\beta}}{\ubar{\sigma}_h n^{\min\{1/2+k, \, 5/6\}}} = C_{\beta} \omega_{n,2}.
\end{align*}
Importantly, Nazarov's inequality is applicable because $\mathbb{E}[Y_{h,j}^{(2)}] \geq \ubar{\sigma}_h$ for all $j=1, \ldots, p$ and $\ubar{\sigma}_h > 0$ due to Condition~\ref{C3}. The second to last inequality follows from the fact that $1-\rho_n \geq 1/2$. We have
\begin{align*}
    &P(\sqrt{n}(U_n + \sqrt{1-\rho_n} B_n) \in R) \leq P(m \overline{Y}_g \in [R - \sqrt{\alpha_n(1-\rho_n)} Y_h]) + C_{\beta} \{\omega_{n,1} + \omega_{n,2}\}\\
    &= P(m \overline{Y}_g + \sqrt{\alpha_n(1-\rho_n)} Y_h \in R) + C_{\beta} \{\omega_{n,1} + \omega_{n,2}\} \\
    &= P \left( \alpha_n^{-1/2} \Lambda_h^{-1/2} m \overline{Y}_g + \sqrt{1-\rho_n} \Lambda_h^{-1/2} Y_h \in \alpha_n^{-1/2} \Lambda_h^{-1/2} R \right) + C_{\beta} \{\omega_{n,1} + \omega_{n,2}\}.
\end{align*}
Observe that $\mathbb{E}[(\sqrt{1-\rho_n}\sigma_{h,j}^{-1} Y_{h,j})^2] = 1-\rho_n \geq 1/2$ for any $1 \leq j \leq p$ and $\|\Lambda_h^{-1/2} \Gamma_h \Lambda_h^{-1/2} \|_{\infty} = 1$ by definition of $\Lambda_h$. Since $\rho_n = N/\binom{n}{m} \leq C_{\beta} N/n^m$ we have by the Gaussian comparison inequality \citep[Lemma C.5]{chen2019randomized}
\begin{align*}
    &P(\sqrt{n}(U_n + \sqrt{1-\rho_n} B_n) \in R) \\
    &\leq P \left( \alpha_n^{-1/2} \Lambda_h^{-1/2} (m \overline{Y}_g + \sqrt{\alpha_n} Y_h) \in \alpha_n^{-1/2} \Lambda_h^{-1/2} R \right) + C_{\beta} \{\omega_{n,1} + \omega_{n,2}\} + C_{\beta} \left(\frac{N \log(p)^2}{ n^m}\right)^{1/3}\\
    &\leq P(m \overline{Y}_g + \sqrt{\alpha_n} Y_h \in R) + C_{\beta} \{\omega_{n,1} + \omega_{n,2} + \omega_{n,3}\} 
\end{align*}
Likewise, we can show
\[
    P(\sqrt{n}(U_n + \sqrt{1-\rho_n} B_n) \in R) \geq P(m \overline{Y}_g + \sqrt{\alpha_n} Y_h \in R) - C_{\beta} \{\omega_{n,1} + \omega_{n,2}  + \omega_{n,3}\}.
\]
It is left to show that the latter two inequalities hold with $\sqrt{n}(U_n + \sqrt{1-\rho_n} B_n)$ replaced by $\sqrt{n}U_{n,N}^{\prime}$. This part of the proof is similar to step $5$ in the proof of Theorem 3.1 in \citet{chen2019randomized} and steps $2$ and $3$ in the proof of Theorem 2.4 in \citet{song2019approximating} and therefore omitted. 
\end{proof}

\begin{prop} \label{prop:2}
Assume~\ref{C2},~\ref{C3} and~\ref{C5} hold. Then, for a constant $C_{\beta} > 0$ only depending on $\beta$, it holds
\vspace{-0.5cm}
\begin{align*} 
    \sup_{R \in \mathbb{R}^p_{\mathrm{re}}} &|P(\overline{Y} \in R) - P(Y \in R)| \leq C_{\beta} \omega_{n,3}.
\end{align*}
\end{prop}

For the proof of Proposition~\ref{prop:2}, we define a $p \times p$ diagonal matrix $\Lambda_h$ such that $\Lambda_{h,jj} = \sigma_{h,j}^2$ for $1 \leq j \leq p.$

\begin{proof}[Proof of Proposition~\ref{prop:2}]
The proof is an application of the Gaussian comparison inequality \citep[Lemma C.5]{chen2019randomized} with appropriate normalizing. Let $C_{\beta} > 0$ be a constant that depends only on $\beta$ and that varies its value from place to place. Observe that $\mathbb{E}[(\alpha_n^{-1/2} \sigma_{h,j}^{-1} \overline{Y}_j )^2] \geq \mathbb{E}[(\sigma_{h,j}^{-1} Y_{h,j})^2]=1$ for any $1 \leq j \leq p$ due to independence of $\overline{Y}_g$ and $Y_h$ and Condition~\ref{C3}. Moreover, \looseness=-1
\begin{align*}
    &\|\alpha_n^{-1} \Lambda_h^{-1/2} \textrm{Cov}[\overline{Y}]\Lambda_h^{-1/2} - \alpha_n^{-1} \Lambda_h^{-1/2} \textrm{Cov}[Y]\Lambda_h^{-1/2} \|_{\infty} = \| \alpha_n^{-1} m^2 \Lambda_h^{-1/2} (\overline{\Gamma}_g - \Gamma_g) \Lambda_h^{-1/2}\|_{\infty} \\
    &\leq \frac{N m^2}{\ubar{\sigma}_h^2 n} \|\overline{\Gamma}_g - \Gamma_g\|_{\infty} \leq \frac{N m^2}{\ubar{\sigma}_h^2 n} \max\left\{\max_{p_1 + 1 \leq j \leq p} \sigma_{g,j}^2, \left(\max_{1 \leq j \leq p_1} \sigma_{g,j}\right) \cdot \left(\max_{p_1 + 1 \leq j \leq p} \sigma_{g,j}\right) \right\} \\
    &\leq C_{\beta} \frac{N m^2}{\ubar{\sigma}_h^2 n} \max\{n^{-2k} D_n^2, n^{-k}D_n^2\} \leq C_{\beta} \frac{N m^2 D_n^2}{\ubar{\sigma}_h^2 n^{1+k}},
\end{align*}
where we used that $\max_{1 \leq j \leq p_1} \sigma_{g,j} \leq \max_{1 \leq j \leq p_1} \sigma_{h,j} \leq  C_{\beta} D_n$ and $\max_{p_1 + 1 \leq j \leq p} \sigma_{g,j} \leq C_{\beta} n^{-k} D_n$ due to Conditions~\ref{C2} and~\ref{C5} and Lemma~\ref{lem:weibull-moments}. Note that Assumption~\ref{C2} also implies $\|g_j(X_1)-\mu_j\|_{\psi_{\beta}} \leq D_n$ for all $j = 1, \ldots, p$. Thus, we have by the Gaussian comparison inequality for any $R \in \mathbb{R}^p_{\mathrm{re}}$
\begin{align*}
    P(\overline{Y} \in R) &= P(\alpha_n^{-1/2} \Lambda_h^{-1/2} \overline{Y} \in \alpha_n^{-1/2} \Lambda_h^{-1/2} R) \\
    &\leq P(\alpha_n^{-1/2} \Lambda_h^{-1/2} Y \in \alpha_n^{-1/2} \Lambda_h^{-1/2} R) + C_{\beta} \left(\frac{N m^2 D_n^2 \log(p)^2}{\ubar{\sigma}_h^2 n^{1+k}}\right)^{1/3}\\
    &\leq P(Y \in R) + C_{\beta} \omega_{n,3}.
\end{align*}
Similarly, we can show $P(\overline{Y} \in R) \geq P(Y \in R) - C_{\beta} \omega_{n,3}$, which concludes the proof.
\end{proof}
    
\begin{proof}[Proof of Theorem~\ref{thm:incompleteUstat}]
The theorem follows from Proposition~\ref{prop:1} and Proposition~\ref{prop:2}.
\end{proof}

\subsection{Other Proofs of Section~\ref{sec:incomplete-U-stat}}

\begin{proof}[Proof of Corollary~\ref{cor:Gaussian-case}]
    Let $C_{s,m} > 0$ be a constant that depends only on $s$ and $m$ and that may change its value in different occurences. Since the polynomials $h_j$ are not constant, we have that $s \geq 1$.  Define $\beta=1/s \in (0,1]$ and $D_n=C_{s,m} \overline{\sigma}_h$. Since $\overline{\sigma}_h \geq \underline{\sigma}_h$, it holds that $D_n \geq 1$ by potentially  increasing $C_{s,m}$. We will first prove that Assumptions \ref{C1}, \ref{C2} and \ref{C6} as well as mixed degeneracy are satisfied, where we fix an arbitrary $k \in (0,1)$. This implies that the Berry-Esseen type bound of Theorem~\ref{thm:incompleteUstat} holds. We then optimize this bound over all $k \in (0,1)$ to obtain the desired result.

    To show that \ref{C1} holds, it is enough to note that $\|h_j(X_1^m)\|_{2+l} \leq C_{s,m} \sigma_{h,j}$ holds for $l=1,2$ due to the hypercontractivity property of polynomials in Gaussian random variables (Theorem~3.2.10 in \citeauthor{pena1999decoupling}, \citeyear{pena1999decoupling} and Lemma~2.2 in \citeauthor{leung2024singularityagnostic}, \citeyear{leung2024singularityagnostic}). To show \ref{C6}, observe that each $g_j$ is also a polynomial in Gaussian variables of degree at most $2s$. Hence, it also holds that $\|g_j(X_1)\|_{2+l} \leq C_{s,m} \sigma_{g,j}$ for $l=1,2$. This implies that \ref{C6} is satisfied since $\sigma_{g,j}^2 \leq \sigma_{h,j}^2$ by Jensen's inequality. Assumption \ref{C2} follows from Lemma \ref{lem:hypercontractivity}. 
    
    It remains to show mixed degeneracy, that is, Assumptions \ref{C4} and \ref{C5}. Fix an arbitrary $k \in (0,1)$ and define $\underline{\sigma}_{g^{(1)}}^2 = C_{s,m} n^{-2k} \, \overline{\sigma}_h^2$. Since each $g_j$ is a polynomial in Gaussian variables of degree at most $2s$, we have by Lemma \ref{lem:hypercontractivity} that $\|g_j(X_1)\|_{\psi_{1/s}} \leq C_{s,m} \sigma_{g,j}$. Hence, if $\sigma_{g,j}^2 < \underline{\sigma}_{g^{(1)}}^2$ for some index $j \in \{1, \ldots, p\}$, then  
    \[
    \|g_j(X_1)\|_{\psi_{1/s}} \leq  C_{s,m}  \underline{\sigma}_{g^{(1)}} \leq C_{s,m} n^{-k} \, \overline{\sigma}_h.
    \]
    That is, if \ref{C4} is not satisfied, then \ref{C5} must be satisfied for this index $j$, and we conclude that mixed degeneracy holds.

    Now, Theorem~\ref{thm:incompleteUstat} implies that 
    \[
        \sup_{R \in \mathbb{R}^p_{\mathrm{re}}} |P(\sqrt{n} U_{n,N}^{\prime} \in R) - P(Y \in R)| \leq C_{\beta} \{\omega_{n,1} + \omega_{n,2} + \omega_{n,3}\}.
    \]
    Since we have $\beta=1/s$, $D_n=C_{s,m} \overline{\sigma}_h$ and $\underline{\sigma}_{g^{(1)}}^2 = C_{s,m} n^{-2k} \, \overline{\sigma}_h^2$ with $k \in (0,1)$, we can further bound $\omega_{n,1}$, $\omega_{n,2}$ and $\omega_{n,3}$ as follows:
    \begin{align*}
        \omega_{n,1} &\leq C_{s,m} \left(\frac{(\overline{\sigma}_h^2 \vee 1) \log(pn)^{1+6s}}{(\underline{\sigma}_h^2 \wedge 1) \, n^{1-k}}\right)^{1/6}
        \leq C_{s,m} \frac{(\overline{\sigma}_h^2 \vee 1) \log(pn)^{1/6+s}}{(\underline{\sigma}_h^2 \wedge 1) \, n^{(1-k)/6}}, \\
        \omega_{n,2} &\leq C_{s,m} \frac{(\overline{\sigma}_h^2 \vee 1) \log(pn)^{1/2+2s}}{(\underline{\sigma}_h^2 \wedge 1) \, n^{\min\{k, 1/3\}}}, \text{ and } \\
        \omega_{n,3} &\leq C_{s,m} \frac{(\overline{\sigma}_h^2 \vee 1) \log(p)^{2/3}}{(\underline{\sigma}_h^2 \wedge 1) \, n^{k/3}},
    \end{align*}
    where we have also used that $n \leq N \leq C_{s,m} n$. 
    Since $k \in (0,1)$, this implies the overall bound
    \[
        \omega_{n,1} + \omega_{n,2} + \omega_{n,3} \leq C_{s,m} \frac{(\overline{\sigma}_h^2 \vee 1) \log(p)^{1/2+2s}}{(\underline{\sigma}_h^2 \wedge 1) \, n^{\min\{(1-k)/6,k/3\}}}.
    \]
    It remains to choose $k \in (0,1)$ such that $\min\{(1-k)/6,k/3\}$ is maximal. This yields $k=1/3$ and $\min\{(1-k)/6,k/3\}=1/9$.
\end{proof}

Recall the definition of the diagonal matrix $\Lambda_h$, that is, $\Lambda_{h,jj} = \sigma_{h,j}^2$ for $1 \leq j \leq p.$ We define the standardized kernel $\tilde{h}(x_1, \ldots, x_m) = \Lambda_h^{-1/2}h(x_1, \ldots, x_m)$. In accordance, let $\Gamma_{\tilde{h}} = \textrm{Cov}[\tilde{h}(X_1^m)]$ be the covariance matrix of the standardized kernel and let $\tilde{U}_{n,N}^{\prime} = \Lambda_h^{-1/2} {U}^{\prime}_{n,N} $ be the standardized incomplete $U$-statistic.

\begin{proof}[Proof of Lemma~\ref{lem:bootstrap}]
Throughout the proof let $C > 0$ be a constant only depending on $\beta, \zeta_1$ and $C_1$ and that varies its value from place to place. It will be useful to note that
\[
    \max \left\{ \frac{m^{2/\beta}}{\ubar{\sigma}_h^4}, \frac{N^2m^4}{\ubar{\sigma}_h^4 n^2}, \frac{N m^2}{\ubar{\sigma}_h^2 n}\right\} \leq A_n.
\]
We start by observing $\mathbb{E}[(\alpha_n^{-1/2} \sigma_{h,j}^{-1} Y_j)^2] \geq \mathbb{E}[(\sigma_{h,j}^{-1} Y_{h,j})^2] \geq 1$ due to the independence of $Y_g$ and $Y_h$ and Condition~\ref{C3}. Thus, by the Gaussian comparison inequality \citep[Lemma C.5]{chen2019randomized} we have
\begin{align}
    &\sup_{R \in \mathbb{R}^p_{\mathrm{re}}}\left|P|_{\mathcal{D}_n}(U_{n,n_1}^\# \in R) - P(Y \in R)\right| \nonumber\\
    &\leq \sup_{R \in \mathbb{R}^p_{\mathrm{re}}}\left|P|_{\mathcal{D}_n}(U_{n,n_1}^\# \in R) - P|_{\mathcal{D}_n}( m U_{n_1,g}^\# + \sqrt{\alpha_n} Y_h \in R)\right| \nonumber\\
    &\quad + \sup_{R \in \mathbb{R}^p_{\mathrm{re}}}\left|P|_{\mathcal{D}_n}( m U_{n_1,g}^\# + \sqrt{\alpha_n} Y_h \in R) - P(Y \in R)\right|\nonumber \\
    &= \sup_{R \in \mathbb{R}^p_{\mathrm{re}}}\left|P|_{\mathcal{D}_n}(\alpha_n^{-1/2} \Lambda_h^{-1/2} U_{n,n_1}^\# \in R) - P|_{\mathcal{D}_n}\left(\Lambda_h^{-1/2} (\alpha_n^{-1/2}  m U_{n_1,g}^\# + Y_h) \in R\right)\right| \nonumber\\
    &\quad + \sup_{R \in \mathbb{R}^p_{\mathrm{re}}}\left|P|_{\mathcal{D}_n}\left( \Lambda_h^{-1/2} (\alpha_n^{-1/2}  m U_{n_1,g}^\# + Y_h) \in R\right) - P(\alpha_n^{-1/2} \Lambda_h^{-1/2} Y \in R)\right| \nonumber\\
    &\leq C \left(\widehat{\Delta}_{\tilde{h}} \log(p)^2\right)^{1/3} + C \left(\frac{Nm^2}{\ubar{\sigma}_h^2 n}\widehat{\Delta}_g \log(p)^2\right)^{1/3}, \nonumber
\end{align}
where
\[
    \widehat{\Delta}_{\tilde{h}} = \left\|\frac{1}{\widehat{N}} \sum_{\iota \in I_{n,m}} Z_{\iota}  (\tilde{h}(X_{\iota}) - \tilde{U}_{n,N}^{\prime})(\tilde{h}(X_{\iota}) - \tilde{U}_{n,N}^{\prime})^{\top} - \Gamma_{\tilde{h}}\right\|_{\infty}
\]
and
\begin{align*}
    \widehat{\Delta}_g &= \left\|\frac{1}{n_1} \sum_{i_1 \in S_1} (G_{i_1} - \overline{G}) (G_{i_1} - \overline{G})^{\top}- \Gamma_g\right\|_{\infty}.
\end{align*}
Now, it is enough to show for each of the two summands involving $\widehat{\Delta}_{\tilde{h}}$ and $\widehat{\Delta}_g$ that the probability of being greater than $C n^{-(\zeta_1 \wedge \zeta_2)/6}$ is at most $C/n$. We start with the first summand.

\underline{Step 1: Bounding $\widehat{\Delta}_{\tilde{h}}$.} 
This is similar to the proof of Theorem 4.1 in \citet{chen2019randomized} and to the proof of Theorem 3.1 in \citet{song2019approximating}. However, we give a full proof for completeness verifying that everything remains true under our assumptions even though some parts will be identical. In particular, we track the  constants $\beta$, $m$ and $\ubar{\sigma}_h^2$. Observe that for any integer $l$, there exists some constant $\tilde{C}$ that depends only on $l$ and $\zeta_1$ such that
\begin{equation} \label{eq:observation}
    \frac{\log(n)^{l}}{n^{\zeta_1}} \leq \tilde{C}
\end{equation}
for all $n \geq 1$. Define 
\begin{align*}
    &\widehat{\Delta}_{\tilde{h},1} = \left\|\frac{1}{N} \sum_{\iota \in I_{n,m}} (Z_{\iota} - \rho_n)\tilde{h}(X_{\iota})\tilde{h}(X_{\iota})^{\top} \right\|_{\infty}, \quad \widehat{\Delta}_{\tilde{h},2} = \left\|\widehat{\Gamma}_{\tilde{h}} - \Gamma_{\tilde{h}}\right\|_{\infty},\\
    &\widehat{\Delta}_{\tilde{h},3} = |N/\widehat{N}-1| \|\Gamma_{\tilde{h}}\|_{\infty} \quad \textrm{and} \quad
    \widehat{\Delta}_{\tilde{h},4} = \left\| \frac{1}{N} \sum_{\iota \in I_{n,m}} Z_{\iota}\tilde{h}(X_{\iota})\right\|_{\infty}^2,
\end{align*}
where $\widehat{\Gamma}_{\tilde{h}} = |I_{n,m}|^{-1} \sum_{\iota \in I_{n,m}} \tilde{h}(X_{\iota})\tilde{h}(X_{\iota})^{\top}$ and observe that 
\[
    \widehat{\Delta}_{\tilde{h}} \leq |N/\widehat{N}|\left(\widehat{\Delta}_{\tilde{h},1} + \widehat{\Delta}_{\tilde{h},2}\right) + \widehat{\Delta}_{\tilde{h},3} + |N/\widehat{N}|^2\widehat{\Delta}_{\tilde{h},4}.
\]
First, we consider the quantity $|N/\widehat{N}|$. We can assume without loss of generality  $C_1 n^{-\zeta_1} \leq 1/16$ since otherwise we may always take $C$ to be large enough and the statement of Lemma~\ref{lem:bootstrap} becomes trivial. Hence, by Condition~\eqref{eq:bootstrapA1}, we have $\sqrt{\log(n)/N} \leq (C_1 n^{-\zeta_1})^{1/2} \leq 1/4$ and thus it follows by Lemma A.12 in \citet{song2019approximating} that
$ 
    P(|N/\widehat{N}-1| > C)  \leq 2/n.
$
Since $|N/\widehat{N}-1| \geq |N/\widehat{N}|-1$ we also have $P(|N/\widehat{N}| > C) \leq 2/n$. Therefore it suffices to show
\[
    P\left(\widehat{\Delta}_{\tilde{h},i} \log(p)^2 > C n^{-\zeta_1/2}\right) \leq \frac{C}{n}
\]
for all $i=1,\ldots,4$, which we naturally divide into four sub-steps.

\underline{Step 1.1: Bounding $\widehat{\Delta}_{\tilde{h},1}$.} 
Conditioned on $\{X_1, \ldots, X_n\}$ we have by Lemma A.3 in \citet{song2019approximating} the inequality
\begin{equation} \label{eq:bound-delta-h1}
    P|_{X_1^n}\left(N \widehat{\Delta}_{\tilde{h},1} > C(\sqrt{N V_n \log(pn)} + M_1 \log(p n))\right) \leq \frac{C}{n},
\end{equation}
where 
$
    V_n = \max_{1 \leq j,l\leq p} |I_{n,m}|^{-1} \sum_{\iota \in I_{n,m}} \tilde{h}_j(X_{\iota})^2 \tilde{h}_l(X_{\iota})^2 
$
and
$
    M_1 = \max_{\iota \in I_{n,m}} \max_{1 \leq j \leq p}$ $\tilde{h}_j(X_{\iota})^2.
$
Next we shall find bounds such that the probability for $V_n$ and $M_1$ being larger than the bounds is at most $C/n$. First, by Lemma~\ref{lem:maximum-weibull}, Lemma~\ref{lem:product-weibull} and due to Assumption~\ref{C2},
\begin{align*}
    \|M_1\|_{\psi_{\beta/2}} &\leq \log(p |I_{n,m}|)^{2/\beta}  \max_{\iota \in I_{n,m}} \max_{1 \leq j \leq p}  \|\tilde{h}_j(X_1^m)^2\|_{\psi_{\beta/2}} \\
    &\leq C \ubar{\sigma}_{h}^{-2}\log(p n^m)^{2/\beta}  \max_{\iota \in I_{n,m}} \max_{1 \leq j \leq p}  \|h_j(X_1^m)\|_{\psi_{\beta}}^2 \\
    &\leq C \ubar{\sigma}_{h}^{-2} m^{2/\beta} \log(pn)^{2/\beta} D_n^2.
\end{align*}
Thus, by the definition of $\|\cdot\|_{\psi_{\beta/2}}$ together with the Markov inequality,
\[
    P(M_1 > C \ubar{\sigma}_{h}^{-2} m^{2/\beta} D_n^2 \log(pn)^{2/\beta} \log(n)^{2/\beta}) \leq 2/n.
\]
Second, we apply Lemma A.6 in \citet{song2019approximating} to bound $V_n$. Due to Assumption~\ref{C1} we have
\[
    \max_{1 \leq j,l\leq p} \mathbb{E}[\tilde{h}_j(X_1^m)^2 \tilde{h}_l(X_1^m)^2] \leq  \ubar{\sigma}_{h}^{-2} D_n^2
\]
and due to Assumption~\ref{C2} and Lemma~\ref{lem:product-weibull} it holds
\[
    \max_{1 \leq j,l\leq p} \|\tilde{h}_j(X_1^m)^2 \tilde{h}_l(X_1^m)^2\|_{\psi_{\beta/4}} \leq \ubar{\sigma}_{h}^{-4} \max_{1 \leq j\leq p} \|h_j(X_1^m)\|_{\psi_{\beta}}^4 \leq \ubar{\sigma}_{h}^{-4} D_n^4.
\]
By  Condition~\eqref{eq:bootstrapA1} and Observation~\eqref{eq:observation} we have 
\begin{align*}
    &\frac{m D_n^2 \log(pn)^{1+4/\beta} \log(n)^{4/\beta - 1}}{\ubar{\sigma}_{h}^{2} \, n} 
    + \frac{m^{2/\beta} D_n^2 \log(pn)^{1+8/\beta} \log(n)^{8/\beta-1}}{ \ubar{\sigma}_{h}^{2} \, n^2} \\
    & \leq C_1 n^{-\zeta_1} \log(n)^{4/\beta - 1} + C_1 n^{-2\zeta_1} \log(n)^{8/\beta - 1} \leq C
\end{align*}
and we obtain the bound $P(V_n > C \ubar{\sigma}_{h}^{-2} D_n^2) \leq \frac{8}{n}$  from \citet[Lemma A.6]{song2019approximating}. Here, we used that $|I_{n,m}|^{-1} \leq C n^{-2} m^{-1/\beta}$; cf. \citet[Lemma A.9]{song2019approximating}. Putting the results together and recalling~\eqref{eq:bound-delta-h1} we have by Fubini
\[
    P\left(\widehat{\Delta}_{\tilde{h},1} > C(N^{-1/2} \ubar{\sigma}_{h}^{-1} D_n \log(pn)^{1/2} + N^{-1} \ubar{\sigma}_{h}^{-2} m^{2/\beta}D_n^2 \log(pn)^{1+4/\beta})\right) \leq \frac{C}{n}
\]
Now, we use again Condition~\eqref{eq:bootstrapA1} to observe
\begin{align*}
    &\log(p)^2 \, C \left(\frac{D_n \log(pn)^{1/2}}{\ubar{\sigma}_{h} \, N^{1/2}} + \frac{m^{2/\beta}D_n^2 \log(pn)^{1+4/\beta}}{\ubar{\sigma}_{h}^2 \, N} \right) \\
    &\leq C \left( \left(\frac{D_n^2 \log(pn)^{5}}{\ubar{\sigma}_h^2 \, N}\right)^{1/2} + \frac{m^{2/\beta}D_n^2 \log(pn)^{3+4/\beta}}{\ubar{\sigma}_h^2 \, N} \right) \\
    &\leq C \left(n^{-\zeta_1/2} + n^{-\zeta_1}\right) \leq Cn^{-\zeta_1/2},
\end{align*}
which implies 
\[
    P\left(\widehat{\Delta}_{\tilde{h},1} \log(p)^2 > C n^{-\zeta_1/2}\right) \leq \frac{C}{n}.
\]

\underline{Step 1.2: Bounding $\widehat{\Delta}_{\tilde{h},2}$.}
By Lemma A.5 in \citet{song2019approximating} we have 
\[
    P\left(\widehat{\Delta}_{\tilde{h},2} > C((n^{-1} m \kappa^2 \log(pn))^{1/2} + n^{-1} m z_n \log(pn)^{4/\beta})\right) \leq \frac{4}{n},
\]
where
\begin{align*}
    &\kappa^2 = \max_{1\leq j,l \leq p} \mathbb{E}[\{\tilde{h}_j(X_1^m)\tilde{h}_l(X_1^m) - \Gamma_{\tilde{h},jl}\}^2], \\
    &z_n = \max_{1\leq j,l \leq p} \|\tilde{h}_j(X_1^m)\tilde{h}_l(X_1^m) - \Gamma_{\tilde{h},jl}\|_{\psi_{\beta/2}}.
\end{align*}
Due to Condition~\ref{C1} we have the bound $\kappa^2 \leq  \max_{1\leq j,l \leq p} \mathbb{E}[\tilde{h}_j(X_1^m)^2\tilde{h}_l(X_1^m)^2] \leq \ubar{\sigma}_h^{-2} D_n^2$ and due to Condition~\ref{C2} and Lemmas~\ref{lem:centering} and~\ref{lem:product-weibull} we have $z_n \leq C \ubar{\sigma}_h^{-2} D_n^2$. Thus,
\[
    P\left(\widehat{\Delta}_{\tilde{h},2} > C\left(\left(\frac{m D_n^2 \log(pn)}{\ubar{\sigma}_h^{2} \, n} \right)^{1/2} + \frac{m D_n^2 \log(pn)^{4/\beta}}{\ubar{\sigma}_h^{2} \, n}\right)\right) \leq \frac{4}{n}.
\]
Using Condition~\eqref{eq:bootstrapA1} we conclude by the same arguments as in Step 1.1,
\[
    P(\widehat{\Delta}_{\tilde{h},2} \log(p)^2 > C n^{-\zeta_1/2}) \leq C/n.
\]

\underline{Step 1.3: Bounding $\widehat{\Delta}_{\tilde{h},3}$.}
By Lemma A.12 in \citet{song2019approximating} we have
\[
    P\left(|N/\widehat{N}-1| > C \sqrt{\log(n)/N}\right) \leq \frac{2}{n}.
\]
Since $\|\Gamma_{\tilde{h}}\|_{\infty} = 1$ this yields
\[
    P\left(\widehat{\Delta}_{\tilde{h},3} \log(p)^2 > C \left(\frac{\log(n) \log(p)^4}{N}\right)^{1/2} \right) \leq \frac{2}{n}.
\]
Observing $\frac{\log(n) \log(p)^4}{N} \leq C_1 n^{-\zeta_1}$ by Condition~\eqref{eq:bootstrapA1} we conclude
\[
    P\left(\widehat{\Delta}_{\tilde{h},3} \log(p)^2 > C n^{-\zeta_1/2} \right) \leq \frac{C}{n}.
\]

\underline{Step 1.4: Bounding $\widehat{\Delta}_{\tilde{h},4}$.}
Observe that $\widehat{\Delta}_{\tilde{h},4} \leq 2(\widehat{\Delta}_{\tilde{h},5}^2 + \widehat{\Delta}_{\tilde{h},6}^2),$ where
\[
    \widehat{\Delta}_{\tilde{h},5} = \left\|\frac{1}{N} \sum_{\iota \in I_{n,m}} (Z_{\iota} - \rho_n) \tilde{h}(X_{\iota})\right\|_{\infty} \quad \textrm{and} \quad \widehat{\Delta}_{\tilde{h},6} = \left\|\frac{1}{|I_{n,m}|} \sum_{\iota \in I_{n,m}}  \tilde{h}(X_{\iota})\right\|_{\infty}
\]
such that it is enough to bound the two terms $\widehat{\Delta}_{\tilde{h},5}^2$ and $\widehat{\Delta}_{\tilde{h},6}^2$ separately. But bounding $\widehat{\Delta}_{\tilde{h},5}$ is similar to bounding $\widehat{\Delta}_{\tilde{h},1}$ by applying \citet[Lemma A.3]{song2019approximating}. Moreover, bounding $\widehat{\Delta}_{\tilde{h},6}$ is similar to bounding $\widehat{\Delta}_{\tilde{h},2}$,  where we applied \citet[Lemma A.5]{song2019approximating}. Therefore, we omit the details.

\underline{Step 2: Bounding $\widehat{\Delta}_g$.}
By the same argument as in the proof of Theorem 4.2 in \citet{chen2019randomized} and using $\max_{1 \leq j \leq p} \sigma_{g,j} \leq C D_n$, which holds true due to Condition~\ref{C2} and Lemma~\ref{lem:weibull-moments},  we have
\[
    \widehat{\Delta}_g \leq C (D_n \widehat{\Delta}_{g,1}^{1/2} + \widehat{\Delta}_{g,1} + \widehat{\Delta}_{g,2} + \widehat{\Delta}_{g,3}^2),
\]
where $\widehat{\Delta}_{g,1}$ is defined in~\eqref{eq:delta_g_1} and
\begin{align*}
    \widehat{\Delta}_{g,2} = \left\| \frac{1}{n_1} \sum_{i_1 \in S_1} \{g(X_{i_1}) g(X_{i_1})^{\top} - \Gamma_{g}\}\right\|_{\infty}, \quad \widehat{\Delta}_{g,3} = \left\|\frac{1}{n_1}\sum_{i_1 \in S_1}g(X_{i_1})\right\|_{\infty}.
\end{align*}
Thus, similar to Step 1, we divide the proof into three sub-steps.

\underline{Step 2.1: Bounding $\widehat{\Delta}_{g,1}$.}
By Condition~\eqref{eq:bootstrapA2} we have
\[
    P\left(\frac{N m^2}{\ubar{\sigma}_h^2 n} D_n \widehat{\Delta}_{g,1}^{1/2} \log(p)^2 > C_1^{1/2} n^{-\zeta_2 / 2}\right) \leq \frac{C}{n}
\]
and 
\[
    P\left(\frac{N m^2}{\ubar{\sigma}_h^2 n} \widehat{\Delta}_{g,1} \log(p)^2 > C_1 n^{-\zeta_2}\right) \leq \frac{C}{n},
\]
where we used $D_n \geq 1$ and $p \geq 3$ for the second bound.

\underline{Step 2.2: Bounding $\widehat{\Delta}_{g,2}$.}
By Lemma A.2 in \citet{song2019approximating} we have
\[
    P\left(\widehat{\Delta}_{g,2} > C( n_1^{-1} \nu \log(pn)^{1/2} + n_1^{-1} u_n \log(pn)^{4/\beta}) \right) \leq \frac{4}{n},
\]
where
\begin{align*}
    & \nu^2 = \max_{1 \leq j,l \leq p} \sum_{i_1 \in S_1} \mathbb{E}[\{g_j(X_{i_1}) g_l(X_{i_1}) - \Gamma_{g,jl}\}^2], \\
    & u_n = \max_{i_1 \in S_1}\max_{1 \leq j,l \leq p} \|g_j(X_{i_1}) g_l(X_{i_1}) - \Gamma_{g,jl}\|_{\psi_{\beta/2}}.
\end{align*}
Due to Condition~\ref{C2} and Lemmas~\ref{lem:product-weibull},~\ref{lem:weibull-moments} and~\ref{lem:centering}  we have the bounds $\nu^2 \leq C n_1 D_n^4$ and $u_n \leq C D_n^2$. Thus,
\[
    P\left(\widehat{\Delta}_{g,2} > C( n_1^{-1/2} D_n^2 \log(pn)^{1/2} + n_1^{-1} D_n^2 \log(pn)^{4/\beta}) \right) \leq \frac{4}{n}.
\]
Now, we use Condition~\eqref{eq:bootstrapA1} to observe
\begin{align*}
    &\frac{N m^2}{\ubar{\sigma}_h^2 n} \log(p)^2 C \left(\frac{D_n^2 \log(pn)^{1/2}}{n_1^{1/2}} + \frac{D_n^2 \log(pn)^{4/\beta}}{n_1}\right) \\
    &\leq C \left(\left(\frac{N^2 m^4 D_n^4 \log(pn)^5}{\ubar{\sigma}_h^4 \, n^2 \, n_1}\right)^{1/2} + \frac{N m^2 D_n^2 \log(pn)^{2+4/\beta}}{\ubar{\sigma}_h^2 \, n \, n_1}\right) \\
    &\leq C \left( n^{-\zeta_1 /2} + n^{-\zeta_1} \right) \leq C n^{-\zeta_1 /2}.
\end{align*}
Thus we conclude
\[
    P\left( \frac{N m^2}{\ubar{\sigma}_h^2 n}  \widehat{\Delta}_{g,2} \log(p)^2 > C n^{-\zeta_1 /2}\right)\leq \frac{4}{n} \leq \frac{C}{n}.
\]

\underline{Step 2.3: Bounding $\widehat{\Delta}_{g,3}$.}
As in Step 2.2 we apply Lemma A.2 in \citet{song2019approximating} once again and obtain
\[
    P\left(\widehat{\Delta}_{g,3} > C( n_1^{-1} \tilde{\nu} \log(pn)^{1/2} + n_1^{-1} \tilde{u}_n \log(pn)^{2/\beta}) \right) \leq \frac{4}{n},
\]
where $\tilde{\nu}^2 = \max_{1 \leq j\leq p} \sum_{i_1 \in S_1} \mathbb{E}[g_j(X_{i_1})^2]$ and  $\tilde{u}_n = \max_{i_1 \in S_1}\max_{1 \leq j \leq p} \|g_j(X_{i_1})\|_{\psi_{\beta}}.$ Due to Assumption~\ref{C2} and Lemma~\ref{lem:weibull-moments} and Lemma~\ref{lem:centering}, we have the bounds $\tilde{\nu}^2 \leq C n_1 D_n^2$ and $\tilde{u}_n \leq C D_n$ giving
\[
    P\left(\widehat{\Delta}_{g,3}^2 > C( n_1^{-1/2} D_n  \log(pn)^{1/2} + n_1^{-1} D_n \log(pn)^{2/\beta}) \right) \leq \frac{4}{n}.
\]
By the same arguments as in Step 2.2 we conclude
\[
    P\left( \frac{N m^2}{\ubar{\sigma}_h^2 n}  \widehat{\Delta}_{g,3}^2 \log(p)^2 > C n^{-\zeta_1 /2}\right) \leq \frac{C}{n}.
    \qedhere
\]
\end{proof}

\begin{proof}[Proof of Theorem~\ref{thm:bootstrap}]
Throughout the proof let $C > 0$ be a constant only depending on $\beta, \nu, \zeta$ and $C_1$ and that varies its value from place to place. Let $\zeta_1 = \zeta$ and $\zeta_2 = \zeta - 1/\nu$. Then, due to Lemma~\ref{lem:bootstrap}, it suffices to show that~\eqref{eq:bootstrapA2} holds. By Step 2 of the proof of Theorem 3.3 in \citet{song2019approximating} we have 
\[
    \mathbb{E}[\widehat{\Delta}_{g,1}^\nu] \leq C \left(\frac{m D_n^2 \log(p)^{1+2/\beta}}{n}\right)^\nu.
\]
Then we have by the Markov inequality
\begin{align*}
    &P(A_n D^2_n \widehat{\Delta}_{g,1} \log(p)^4 > C_1 n^{-(\zeta - 1/\nu)} ) \\
    &\quad \leq C A_n^\nu D_n^{2\nu} \log(p)^{4 \nu} n^{(\zeta - 1/\nu) \nu}  \left(\frac{m D_n^2 \log(p)^{1+2/\beta}}{n}\right)^\nu \\
    &\quad = C n^{-1} \left(\frac{n^{\zeta} A_n D_n^4 \log(p)^{5+2/\beta}}{n}\right)^\nu \\
    &\quad \leq C n^{-1} \left(n^{\zeta} C_1 n^{-\zeta} \right)^\nu = \frac{C}{n},
\end{align*}
where the last inequality follows from Condition~\eqref{eq:main-bootstrap-cond}.
\end{proof}

\begin{proof}[Proof of Corollary~\ref{cor:studentization}]
Throughout the proof let $C > 0$ be a constant only depending on $\beta, \zeta$ and $C_1$ and that varies its value from place to place. We begin by bounding the quantity $\max_{1 \leq j \leq p}|\hat{\sigma}_j^2/\sigma_j^2-1|$. Note that
\begin{align*}
    \left|\frac{\hat{\sigma}_j^2}{\sigma_j^2}-1\right| &= \left|\frac{m^2(\hat{\sigma}_{g,j}^2-\sigma_{g,j}^2) + \alpha_n(\hat{\sigma}_{h,j}^2-\sigma_{h,j}^2)}{m^2\sigma_{g,j}^2+\alpha_n \sigma_{h,j}^2}\right| \\
    &\leq \frac{m^2|\hat{\sigma}_{g,j}^2-\sigma_{g,j}^2|}{m^2\sigma_{g,j}^2+\alpha_n \sigma_{h,j}^2} + \frac{\alpha_n|\hat{\sigma}_{h,j}^2-\sigma_{h,j}^2|}{m^2\sigma_{g,j}^2+\alpha_n \sigma_{h,j}^2}\\
    &\leq \frac{m^2}{\alpha_n \sigma_{h,j}^2} |\hat{\sigma}_{g,j}^2-\sigma_{g,j}^2| + \frac{1}{\sigma_{h,j}^2} |\hat{\sigma}_{h,j}^2-\sigma_{h,j}^2|.
\end{align*}
Taking the maximum we get 
\begin{align*}
    \max_{1 \leq j \leq p}\left|\frac{\hat{\sigma}_j^2}{\sigma_j^2}-1\right| \leq \frac{N m^2}{\ubar{\sigma}_{h}^2 n} \max_{1 \leq j \leq p}|\hat{\sigma}_{g,j}^2-\sigma_{g,j}^2| + \max_{1 \leq j \leq p}  \frac{1}{\sigma_{h,j}^2}|\hat{\sigma}_{h,j}^2-\sigma_{h,j}^2| \leq  \frac{N m^2}{\ubar{\sigma}_{h}^2 n} \widehat{\Delta}_g + \widehat{\Delta}_{\tilde{h}},
\end{align*}
where $\widehat{\Delta}_g$ and $\widehat{\Delta}_{\tilde{h}}$ were defined in the proof of Lemma~\ref{lem:bootstrap}. We have shown in Step 1 in the proof of Lemma~\ref{lem:bootstrap}
that 
\[
   P\left(\widehat{\Delta}_{\tilde{h}} \log(p)^2  > Cn^{-\zeta/2}\right) \leq \frac{C}{n}.
\]
Moreover, if we take $\nu = 7/\zeta$ and $\zeta_2 = \zeta - 1/\nu$, then in the proofs of Lemma~\ref{lem:bootstrap} and Theorem~\ref{thm:bootstrap} we have shown that
\[
   P\left(\frac{N m^2}{\ubar{\sigma}_{h}^{-2} n} \widehat{\Delta}_g \log(p)^2 > C n^{-3\zeta/7} \right) \leq \frac{C}{n}.
\]
Hence, we have
\begin{equation} \label{eq:covariance-ratio}
    P\left( \max_{1 \leq j \leq p}|\hat{\sigma}_j^2/\sigma_j^2-1| \log(p)^2 > C n^{-3\zeta/7} \right) \leq \frac{C}{n}.
\end{equation}
The rest of the proof is the same as the proof for Corollary A.1 in \citet{chen2019randomized} and thus omitted.
\end{proof}

\subsection{Proofs of Section~\ref{sec:methodology}}
\begin{proof}[Proof of Corollary~\ref{cor:quantiles}]
Throughout the proof let $C > 0$ be a constant only depending on $\beta, \zeta$ and $C_1$ and that varies its value from place to place. Moreover, we write $\mu_j = f_j(\theta)$ and for $t \in \mathbb{R}$ we consider hyperrectangles of the form
$
    H_t = \{x \in \mathbb{R}^p: - \infty \leq x_j \leq t \textrm{ for all } j=1, \ldots, p\}.
$
The last steps in the proof of Corollary~\ref{cor:studentization} are equal to the last steps in the proof of  \citet[Corollary A.1]{chen2019randomized} and they show that we have the bound
\begin{equation}  \label{eq:bound-for-max-1}
    \sup_{t \in \mathbb{R}} \left|P\left(\max_{1 \leq j\leq p} \sqrt{n} \ (U_{n,N,j}^{\prime} - \mu_j) / \hat{\sigma}_j \leq t \right) - P(\max_{1 \leq j\leq p} Y_j / \sigma_j \leq t) \right|  \leq C n^{-\zeta/7},
\end{equation}
where $Y \sim N_p(m^2 \Gamma_g + \alpha_n \Gamma_h )$. Moreover, with  probability at least $1-C/n$ we have
\begin{equation} \label{eq:bound-for-max-2}
\sup_{t \in \mathbb{R}} \left|P(\max_{1 \leq j\leq p} Y_j / \sigma_j \leq t ) - P|_{\mathcal{D}_n}(W \leq t) \right| \leq C n^{-\zeta/7}.
\end{equation}

Now, let $\alpha \in (0,1)$. To simplify notation we write $T_0 = \max_{1 \leq j\leq p} \sqrt{n} \ (U_{n,N,j}^{\prime} - \mu_j) / \hat{\sigma}_j$ and  $Y_0 = \max_{1 \leq j\leq p} Y_j / \sigma_j$. Moreover, we denote by $c_{Y_0}(1-\alpha)$ the $(1-\alpha)$-quantile of $Y_0$.  To begin, we use~\eqref{eq:bound-for-max-2} to establish a relation between the quantiles of $Y_0$ and $W$. With probability at least $1-C/n$ we have
\begin{align*}
    P|_{\mathcal{D}_n}(W \leq c_{Y_0}(1-\alpha-C n^{-\zeta/7})) & \leq P(Y_0 \leq  c_{Y_0}(1-\alpha-C n^{-\zeta/7})) + C n^{-\zeta/7} \\
    &= 1-\alpha-C n^{-\zeta/7} +  C n^{-\zeta/7} \\
    &= 1 - \alpha.
\end{align*}
Thus, by definition of a quantile, we have the relation
\begin{equation} \label{eq:rel-quantiles-1}
    P(c_W(1-\alpha) \geq c_{Y_0}(1-\alpha-C n^{-\zeta/7})) \geq 1-C/n.
\end{equation}
Similarly, we can show $P|_{\mathcal{D}_n}(W \leq c_{Y_0}(1-\alpha+C n^{-\zeta/7})) \geq 1-\alpha$ and therefore we have
\begin{equation} \label{eq:rel-quantiles-2}
    P(c_W(1-\alpha) \leq c_{Y_0}(1-\alpha+C n^{-\zeta/7})) \geq 1-C/n.
\end{equation}
Now, we apply~\eqref{eq:rel-quantiles-1} and~\eqref{eq:bound-for-max-1}  to conclude
\begin{align*}
    P(T_0 > c_W(1-\alpha)) & \leq P(T_0 > c_{Y_0}(1-\alpha - C n^{-\zeta/7})) + C/n \\
    & \leq P(Y_0 > c_{Y_0}(1-\alpha - C n^{-\zeta/7})) + C n^{-\zeta/7} + C/n \\
    &=  \alpha + C n^{-\zeta/7} + C n^{-\zeta/7} + C/n \\
    &= \alpha + C n^{-\zeta/7},
\end{align*}
where we used that $Y_0$ is a continuous distribution with no point mass. The constant $C$ may have changed its value in the last step. Equivalently, by using the relation~\eqref{eq:rel-quantiles-2}, we get the other direction 
\begin{align*}
    P(T_0 > c_W(1-\alpha)) &\geq P(T_0 > c_{Y_0}(1-\alpha + C n^{-\zeta/7})) - C/n\\
    &\geq P(Y_0 > c_{Y_0}(1-\alpha + C n^{-\zeta/7})) - C n^{-\zeta/7} - C/n\\
    &= \alpha - C n^{-\zeta/7} - C n^{-\zeta/7} - C/n \\
    &= \alpha - C n^{-\zeta/7}.
\end{align*}
Noting that both inequalities hold uniformly over all $\alpha \in (0,1)$ finishes the proof.
\end{proof}

\begin{proof}[Proof of Proposition~\ref{prop:power}]
Throughout the proof let $C > 0$ be a constant only depending on $\beta, \zeta$ and $C_1$, that varies its value from place to place. Write $\mu_j = f_j(\theta)$ and let $j^{\ast} \in \{1 \leq j \leq p\}$ be any index such that $\mu_{j^{\ast}}/\sigma_{j^{\ast}} = \max_{1 \leq j \leq p} (\mu_j/\sigma_j)$. To simplify notation, we also define $\kappa_{p,\alpha} = \sqrt{2 \log(p)} + \sqrt{2 \log(1/\alpha)}$ and $\delta =  C n^{-3\zeta/7} \log(p)^{-2}$ such that Inequality~\eqref{eq:power-assumption} is equivalent to
\[
    \mu_{j^{\ast}}/\sigma_{j^{\ast}} \geq (1+\varepsilon)\left(1 + \delta \right) \frac{\kappa_{p,\alpha}}{\sqrt{n}}.
\]
We assume without loss of generality $\delta \leq 1/2$, since otherwise the statement from Corollary~\ref{prop:power} becomes trivial by choosing $C$ large enough. Now, let $A$ be the event such that $|\widehat{\sigma}_{j^{\ast}}/\sigma_{j^{\ast}}-1| \leq \delta$. Then, we have
\begin{align*}
    &P\left(\mathcal{T} > c_W(1-\alpha)\right) \\
    &= P\left(\sqrt{n} \, U_{n,N,j^{\ast}}/\widehat{\sigma}_{j^{\ast}} > c_W(1-\alpha) \right) \\
    &= P\left(\sqrt{n} \, \mu_{j^{\ast}}/\widehat{\sigma}_{j^{\ast}} + \sqrt{n} \, (U_{n,N,j^{\ast}} - \mu_{j^{\ast}})/\widehat{\sigma}_{j^{\ast}}  > c_W(1-\alpha)\right) \\
    &= P\left((\sigma_{j^{\ast}}/\widehat{\sigma}_{j^{\ast}}) \sqrt{n} \, \mu_{j^{\ast}}/\sigma_{j^{\ast}}   + \sqrt{n} \, (U_{n,N,j^{\ast}} - \mu_{j^{\ast}})/\widehat{\sigma}_{j^{\ast}}  > c_W(1-\alpha)\right) \\
    &\geq  P\left( (1/(1+\delta)) \sqrt{n} \, \mu_{j^{\ast}}/\sigma_{j^{\ast}} + \sqrt{n} \, (U_{n,N,j^{\ast}} - \mu_{j^{\ast}})/\widehat{\sigma}_{j^{\ast}}  > c_W(1-\alpha)\right) - P\left(A^{c}\right) \\
    &\geq  P\left((1+\varepsilon) \kappa_{p,\alpha} + \sqrt{n} \, (U_{n,N,j^{\ast}} - \mu_{j^{\ast}})/\widehat{\sigma}_{j^{\ast}}  > c_W(1-\alpha)\right) - P\left(A^{c}\right),
\end{align*}
where we used that $c_W(1-\alpha) \geq 0$, since $\alpha \in (0,1/2)$. Furthermore, by
 \citet[Lemma D.4]{chernozhukov2019inference} we have the bound  $c_W(1-\alpha) \leq \kappa_{p,\alpha}$. Hence, it follows that
\begin{align*}
     P\left(\mathcal{T} > c_W(1-\alpha)\right) 
     &\geq P\left(\sqrt{n} \, (U_{n,N,j^{\ast}} - \mu_{j^{\ast}})/\widehat{\sigma}_{j^{\ast}}  > - \varepsilon \, \kappa_{p,\alpha} \right) - P\left(A^{c}\right) \\
     &=  P\left((\sigma_{j^{\ast}}/\widehat{\sigma}_{j^{\ast}}) \sqrt{n} \, (U_{n,N,j^{\ast}} - \mu_{j^{\ast}})/\sigma_{j^{\ast}}  > - \varepsilon \, \kappa_{p,\alpha} \right) - P\left(A^{c}\right) \\
     &\geq P\left(\sqrt{n} \, (U_{n,N,j^{\ast}} - \mu_{j^{\ast}})/\sigma_{j^{\ast}}  > - (1-\delta) \varepsilon \, \kappa_{p,\alpha} \right) - P\left(A^{c}\right) \\
     &= 1 - P\left(\sqrt{n} \, (U_{n,N,j^{\ast}} - \mu_{j^{\ast}})/\sigma_{j^{\ast}}  \leq - (1-\delta) \varepsilon \, \kappa_{p,\alpha} \right) - P\left(A^{c}\right). 
\end{align*}
By Theorem~\ref{thm:incompleteUstat}, we have
\begin{align*}
    &P\left(\sqrt{n} \, (U_{n,N,j^{\ast}} - \mu_{j^{\ast}})/\sigma_{j^{\ast}}  \leq - (1-\delta) \varepsilon \, \kappa_{p,\alpha} \right) \\
    &\leq P\left(Y_{j^{\ast}}/\sigma_{j^{\ast}}  \leq - (1-\delta) \varepsilon \, \kappa_{p,\alpha} \right) + C\{\omega_{n,1}+\omega_{n,2}+\omega_{n,3}\} \\
    &\leq P\left(Y_{j^{\ast}}/\sigma_{j^{\ast}}  \leq - (1-\delta) \varepsilon \, \kappa_{p,\alpha} \right) + C n^{-\zeta/7}.
\end{align*}
Moreover, by Equation~\ref{eq:covariance-ratio} in the proof of Corollary~\ref{cor:studentization}, we have 
\[
    P\left(A^{c}\right) = P\left(|\widehat{\sigma}_{j^{\ast}}/\sigma_{j^{\ast}}-1| > \delta\right) \leq \frac{C}{n} \leq  C n^{-\zeta/7}.
\]
This implies that
\begin{align*}
    P\left(\mathcal{T} > c_W(1-\alpha)\right) 
    &\geq 1 - P\left(Y_{j^{\ast}}/\sigma_{j^{\ast}}  \geq (1-\delta) \varepsilon \, \kappa_{p,\alpha} \right) - C n^{-\zeta/7}.
\end{align*}
Finally, by the Chernoff bound for univariate normal distributions, we have
\begin{align*}
    P\left(Y_{j^{\ast}}/\sigma_{j^{\ast}}  \geq (1-\delta) \varepsilon \, \kappa_{p,\alpha} \right)
    & \leq \exp\left(-\frac{1}{2}(1-\delta)^2 \varepsilon^2 \, \kappa_{p,\alpha}^2\right) \\
    &\leq \exp\left(-\frac{1}{2}(1-\delta)^2 \varepsilon^2 \, (2 \log(p) + 2 \log(1/\alpha))\right) \\
    &= \exp\left(- (1-\delta)^2 \varepsilon^2 \, \log(p/\alpha)\right) \\
    &\leq \exp\left(- \frac{1}{C} \, \varepsilon^2 \, \log(p/\alpha)\right),
\end{align*}
where we used $\delta \leq 1/2$ in the last step. 
\end{proof}

\subsection{Proofs of Section~\ref{sec:polynomial-hypotheses}}

\begin{proof}[Proof of Proposition~\ref{prop:criterion-regular-points}]
By Step 3) in the construction of the kernel $h$, we have that
\begin{align*}
h(X_1^m) &= \frac{1}{m!} \sum_{\pi \in S_m}  \breve{h}(X_{\pi(1)}, \ldots, X_{\pi(m)}) \\
&= \frac{1}{m} \sum_{k=1}^m \underbrace{\frac{1}{(m-1)!} \sum_{\pi^{\ast} \in S_{m-1}^{\ast}} \breve{h}(X_{\pi^{\ast}(2)}, \ldots, X_{\pi^{\ast}(k-1)}, X_1, X_{\pi^{\ast}(k)}, \ldots, X_{\pi^{\ast}(m)})}_{=:\breve{h}^{(k)}(X_1^m)} \\
&= \frac{1}{m} \sum_{k=1}^m \breve{h}^{(k)}(X_1^m),
\end{align*}
where $S^{\ast}_{m-1}$ denotes the group of all permutations of the set $\{2, \ldots, m\}$. 
Now, let $k \leq m$ be a fixed integer and define $K=\lceil k/\eta \rceil$. Then we can write $k=(K-1) \eta + l$ for $1 \leq l \leq \eta$. Applying the definition of $\breve{h}$ in step 2) of the construction, we see that
\begin{align*}
\mathbb{E}[\breve{h}^{(k)}(X_1^m)|X_1] &= \underbrace{a_0 + \sum_{r=1}^{K-1} \sum\limits_{\substack{(i_1, \ldots, i_r) \\i_{j} \in \{1, \ldots, d\}}} a_{(i_1, \ldots, i_r)} \theta_{i_1} \cdots \theta_{i_r}}_{=: R^{(K-1)}(\theta)} \\
& \hspace{0.82cm} + \sum_{r=K}^{s} \sum\limits_{\substack{(i_1, \ldots, i_r) \\i_{j} \in \{1, \ldots, d\}}} a_{(i_1, \ldots, i_r)} \theta_{i_1} \cdots \theta_{i_{K-1}} \hat{\theta}_{i_K,l}(X_1) \theta_{i_{K+1}} \cdots \theta_{i_{r}} \\
&= R^{(K-1)}(\theta) + \sum_{i_K \in D(f)} \hat{\theta}_{i_K,l}(X_1) \tilde{g}_{i_K}^{(K)}(\theta),
\end{align*}
where $R^{(K-1)}$ and $\tilde{g}_{i_K}^{(K)}$ are polynomials in $\mathbb{R}[\theta_1, \ldots, \theta_d]$ and $\tilde{g}_{i_K}^{(K)}$ is defined by
\[
    \tilde{g}_{i_K}^{(K)}(\theta) = \sum_{r=K}^{s}  \sum\limits_{\substack{(i_1, \ldots, i_{K-1},i_{K+1},\ldots, i_r) \\i_{j} \in \{1, \ldots, d\}}} a_{(i_1, \ldots, i_r)} \theta_{i_1} \cdots \theta_{i_{K-1}} \theta_{i_{K+1}} \cdots \theta_{i_{r}}.
\]
Now, recall that $m=\eta s$ and observe that we can write $g(X_1)$ as follows.
\begin{align*}
g(X_1) &= \mathbb{E}[h(X_1^m)|X_1] = \frac{1}{m} \sum_{k=1}^m \mathbb{E}[\breve{h}^{(k)}(X_1^m)|X_1] \\
& = \frac{1}{m} \sum_{K=1}^s \sum_{l=1}^{\eta} \left(R^{(K-1)}(\theta) + \sum_{i \in D(f)} \hat{\theta}_{i,l}(X_1) \tilde{g}_{i}^{(K)}(\theta)\right) \\
&= \frac{\eta}{m} \sum_{K=1}^s R^{(K-1)}(\theta) + \frac{1}{m}  \sum_{K=1}^s \sum_{l=1}^{\eta}  \sum_{i \in D(f)} \hat{\theta}_{i,l}(X_1) \tilde{g}_{i}^{(K)}(\theta) \\
&= \underbrace{\frac{\eta}{m} \sum_{K=1}^{s} R^{(K-1)}(\theta)}_{=: R(\theta)} +    \sum_{i \in D(f)} \underbrace{\sum_{l=1}^{\eta} \hat{\theta}_{i,l}(X_1)}_{= \hat{\theta}_f(X_1)_i} \underbrace{\left( \frac{1}{m} \sum_{K=1}^s \tilde{g}_{i}^{(K)}(\theta)\right)}_{=:\tilde{g}_i(\theta)} \\
&= R(\theta) +  \sum_{i \in D(f)} \hat{\theta}_f(X_1)_i \,\, \tilde{g}_i(\theta) \\
&= R(\theta) + \tilde{g}(\theta)^{\top} \, \hat{\theta}_f(X_1),
\end{align*}
where $\tilde{g}(\theta)$ is the vector $(\tilde{g}_i(\theta))_{i \in D(f)}$. We observe that the variance of $g(X_1)$ is given by
\[
    \var_{\theta}[g(X_1)] = \tilde{g}(\theta)^{\top} \Cov[\hat{\theta}_f(X_1)] \, \tilde{g}(\theta).
\]
Note that this is a polynomial, i.e., $\var_{\theta}[g(X_1)] \in \mathbb{R}[\theta_1, \ldots,\theta_d]$. We will now argue that $\var_{\theta}[g(X_1)]$ is not the zero polynomial. 
It is easy to see by the above derivations that $\mathbb{E}_{\theta}[h(X_1^m)]=f(\theta)= R(\theta)$ if all components of $\tilde{g}(\theta)$ are identically to the zero polynomial. But this is a contradiction since the degree of the polynomial $R(\theta)$ is at most $s-1$. Hence, there is at least one component in $\tilde{g}(\theta)$ that is not zero. Since $\Cov[\hat{\theta}_f(X_1)]$ is positive definite, we conclude that $\var_{\theta}[g(X_1)]$ is also not identical to the zero polynomial. The proof is finished by recalling that the zero set of a real polynomial that is not the zero polynomial has Lebesgue measure zero; see the lemma of \citet{okamoto1973distinctness}.
\end{proof}

\section{Additional Lemmas} \label{sec:useful-lemmas}
The following Lemma concerns Gaussian approximation for independent sums. It is a generalization of Proposition 2.1 in \citet{chernozhukov2017central} and implicitly proved in the work of \citet{song2019approximating}. However, we state an explicit version here since our proofs rely on this Lemma.

\begin{lem} \label{lem:CLT}
Let $X_1, \ldots, X_n$ be i.i.d.~centered $\mathbb{R}^p$-valued random vectors. Let $\beta \in (0,1]$ be an absolute constant and assume there exists $\ubar{\sigma}^2 > 0$ and $D_n \geq 1$ such that, for all $i=1, \ldots, n$,
\begin{align*}
    &\mathbb{E}[X_{ij}^2] \geq \ubar{\sigma}^2, \qquad \mathbb{E}[X_{ij}^{2+l}] \leq \mathbb{E}[X_{ij}^2] D_n^l \ \textrm{ for } \ j=1,\ldots,p, \ \ l=1,2, \\
    &\|X_{ij}\|_{\psi_{\beta}} \leq D_n \textrm{ for } j=1,\ldots,p.
\end{align*}
Then there is a constant $C_{\beta} > 0$ only depending on $\beta$ such that
\[
\sup_{R \in \mathbb{R}^p_{\mathrm{re}}} \left|P\left(\frac{1}{\sqrt{n}} \sum_{i = 1}^n X_i \in R\right) - P(Y \in R)\right| \leq C_{\beta} \left(\frac{D_n^2 \log(pn)^{1+6/\beta}}{\ubar{\sigma}^2 \, n}\right)^{1/6},
\]
where $Y \sim N_p(0, \Sigma)$ with $\Sigma = \mathbb{E}[X_i X_i^{\top}]$.
\end{lem}
\begin{proof}
Essentially, this Lemma is obtained from the high dimensional central limit theorem \citep[Proposition 2.1]{chernozhukov2017central} with generalizing for $\beta \in (0,1]$ and suitable normalization such that the constant $\ubar{\sigma}^2$ is explicit in the bound. Generalizing for $\beta \in (0,1]$ is established in \citet[Lemma A.8]{song2019approximating} and the proof with suitable normalization is implicit in the proof of Corollary 2.2 in \citet{song2019approximating}.
\end{proof}

The next Lemma verifies Gaussian approximation of complete $U$-statistics. It is proved in \citet{song2019approximating}. Since the Lemma is essential for the proof of our main theorem, we state it here.

\begin{lem}[\citeauthor{song2019approximating}, \citeyear{song2019approximating}, Corollary 2.2] \label{lem:CLT-complete-Ustat}
Assume the sub-Weibull condition~\ref{C2} is satisfied. Moreover, assume~\ref{C4} holds with $p_1 = p$ and~\ref{C6} holds. Then there is a constant $C_{\beta} > 0$ only depending on $\beta$ such that 
\[
    \sup_{R \in \mathbb{R}^p_{\mathrm{re}}} \left|P\left(\sqrt{n}(U_n - \mu) \in R\right) - P(mY_g \in R)\right| \leq C_{\beta} \left(\frac{m^2 D_n^2 \log(pn)^{1+6/\beta}}{(\ubar{\sigma}_g^2 \wedge 1) \, n}\right)^{1/6},
\]
where $Y_g \sim N_p(0, \Gamma_g)$.
\end{lem}

\section{Sub-Weibull Random Variables} \label{sec:sub-Weibull}

In this section we collect important properties of sub-Weibull random variables. For $0 < \beta < 1$ we only have that $\|\cdot\|_{\psi_{\beta}}$ is a quasinorm, i.e. the triangle inequality does not hold. Nevertheless, we have the following result which is a substitute.

\begin{lem} \label{lem:triangle-inequality}
For any $0<\beta<1$ and any random variables $X_1, \ldots, X_n$, we have
$$\left\|\sum_{i=1}^{n}X_i\right\|_{\psi_{\beta}} \leq n^{\frac{1}{\beta}} \sum_{i=1}^{n} \|X_i\|_{\psi_{\beta}}.$$
\end{lem}
\begin{proof} 
For $n=2$ the claim is equal to Lemma A.3 in \citet{goetze2021concentration}. We generalize the statement for arbitrary $n\in \mathbb{N}$ using similar arguments. Let $L_i=\|X_i\|_{\psi_{\beta}}$ and define $t \defeq n^{\frac{1}{\beta}} \sum_{i=1}^{n} L_i$. We have 
\begin{eqnarray*}
&& \mathbb{E}\left[\exp\left(\frac{\left| \sum_{i=1}^n X_i \right|^{\beta}}{t^{\beta}} \right)\right] \\
&\leq& \mathbb{E}\left[\exp\left(\frac{\left( \sum_{i=1}^n |X_i| \right)^{\beta}}{t^{\beta}} \right)\right] \overset{\textrm{(1)}}{\leq} \mathbb{E}\left[\exp\left(\frac{\sum_{i=1}^n |X_i|^{\beta}}{n\left(\sum_{i=1}^{n} L_i\right)^{\beta}}  \right)\right]\\
&=& \mathbb{E}\left[\prod_{i=1}^n \exp\left(\frac{|X_i|^{\beta}}{n\left(\sum_{i=1}^{n} L_i\right)^{\beta}}\right)\right]\\
&\leq&\mathbb{E}\left[\prod_{i=1}^n \exp\left(\frac{|X_i|^{\beta}}{ n L_i^{\beta}}\right)\right] = \mathbb{E}\left[\prod_{i=1}^n \exp\left(\frac{|X_i|^{\beta}}{ L_i^{\beta}}\right)^{\frac{1}{n}}\right]\\
&\overset{\textrm{(2)}}{\leq}& \mathbb{E}\left[\frac{1}{n} \sum_{i=1}^n \exp\left(\frac{|X_i|^{\beta}}{ L_i^{\beta}}\right)\right]
= \frac{1}{n} \sum_{i=1}^n \underbrace{\mathbb{E}\left[ \exp\left(\frac{|X_i|^{\beta}}{ L_i^{\beta}}\right)\right]}_{\leq 2 \textrm{ for all } i} \leq 2.
\end{eqnarray*}
Here, (1) follows from the inequality $\left( \sum_{i=1}^{n}x_i\right)^{\beta} \leq \sum_{i=1}^{n}x_i^{\beta}$ valid for all $x_1, \ldots, x_n \geq 0$ and $\beta \in [0,1]$, and (2) from $\prod_{i=1}^n a_i^{\frac{1}{n}} \leq \frac{1}{n} \sum_{i=1}^n a_i$ for all $a_1, \ldots, a_n \geq 0$. The latter inequality is known as the inequality of geometric and arithmetic means. By the calculation and the definition of the sub-Weibull norm we have 
$$\left\|\sum_{i=1}^{n}X_i\right\|_{\psi_{\beta}} \leq t = n^{\frac{1}{\beta}} \sum_{i=1}^{n} \|X_i\|_{\psi_{\beta}}.$$
\end{proof}

For products of sub-Weibull random variables we recall the following useful result from \citet{kuchibhotla2022moving}.

\begin{lem}\label{lem:product-weibull} 
Let $X_1, \ldots, X_n$ be random variables satisfying $\|X_i\|_{\psi_{\beta_i}} < \infty$ for some $\beta_i > 0$, $i=1, \ldots, n$. Then
$$\left\|\prod_{i=1}^{n} X_i\right\|_{\psi_{\beta}} \leq \prod_{i=1}^{n} \|X_i\|_{\psi_{\beta_i}} \quad \textrm{ where } \frac{1}{\beta}= \sum_{i=1}^{n} \frac{1}{\beta_i}.$$
\end{lem}
\begin{proof}
See \citet[Proposition D.2]{kuchibhotla2022moving}. 
\end{proof}

Combining Lemma~\ref{lem:triangle-inequality} and Lemma~\ref{lem:product-weibull} we get a bound on polynomial functions in random variables. This is especially useful for showing that the estimators $h_j(X_1^m)$, $j=1, \ldots, p$ are sub-Weibull.

\begin{lem} \label{lem:bound-polynom}
Let $\beta>0$ and let $X_1, \ldots, X_n$ be (possibly dependent) random variables satisfying $\|X_i\|_{\psi_{\beta}} \leq C$ for some constant $C>0$. Let $f:\mathbb{R}^n \rightarrow \mathbb{R}$ be a polynomial of total degree $s$ with $t$ terms. Then we have
$$\|f(X_1, \ldots, X_n)\|_{\psi_{\frac{\beta}{s}}}  \leq A \ \begin{cases}
    t \ C^s & \frac{\beta}{s} \geq 1 \\
    t^{\frac{s}{\beta}+1} \ C^s& \frac{\beta}{s} < 1,
\end{cases}$$
where $A$ is the maximum over all absolute values of the coefficients in the polynomial $f$.
\end{lem}
\begin{proof}
For random variables $X_1, \ldots, X_s$  we have by Lemma~\ref{lem:product-weibull}
$$
\left\|\prod_{i=1}^{s} X_i\right\|_{\psi_{\frac{\beta}{s}}} \leq \|X_i\|_{\psi_{\beta}}^s.
$$
Hence, for $f(X_1, \ldots, X_n)$ being a monomial of degree $s$, it follows the inequality
$$\|f(X_1, \ldots, X_n)\|_{\psi_{\frac{\beta}{s}}} \leq A \ C^s.$$
To prove the general case where $f(X_1, \ldots, X_n)$ is a polynomial with $t$ terms we make a case distinction. For $\frac{\beta}{s} \geq 1$ we have that $\|\cdot\|_{\psi_{\frac{\beta}{s}}}$ is a norm and the triangle inequality holds. For $\frac{\beta}{s} < 1$ the claim follows immediately from Lemma ~\ref{lem:triangle-inequality}. 
\end{proof}

The next well-known Lemma states that sub-Weibull random variables satisfy a stronger moment condition than the existence of all finite $q$-th moments.

\begin{lem} \label{lem:weibull-moments}
For any $\beta > 0$ and any random variable $X$ we have
$$d_{\beta} \sup_{q \geq 1} \frac{\|X\|_q}{q^{\frac{1}{\beta}}} \leq \|X\|_{\psi_{\beta}} \leq D_{\beta} \sup_{q \geq 1} \frac{\|X\|_q}{q^{\frac{1}{\beta}}},$$
where $d_{\beta} = \begin{cases}
    \frac{1}{2} & \beta \geq 1 \\
    \frac{(\beta e)^{1/\beta}}{2} & \beta < 1
\end{cases}$ and $D_{\beta} = \begin{cases}
    2e & \beta \geq 1 \\
    2e^{1/\beta} & \beta < 1.
\end{cases}$
\end{lem}
\begin{proof}
See for example \citet[Lemma A.2]{goetze2021concentration}.
\end{proof}

According to Lemma~\ref{lem:weibull-moments}, the parameter $\beta$ measures how fast $\|X\|_q$ increases with $q$. For small $\beta$, the norms $\|X\|_q$ are allowed to increase fast and viceversa. In particular, Lemma~\ref{lem:weibull-moments} implies that for all $q \geq 1$ there is a constant $C_{\beta}>0$ only depending on $\beta$ such that $\|X\|_q \leq C_{\beta} \|X\|_{\psi_{\beta}}$. 
 In the special case of polynomials in Gaussian variables, the sub-Weibull norm is always bounded by the standard deviation. This follows from the hypercontractivity property of polynomials in Gaussian random variables.
\begin{lem} \label{lem:hypercontractivity}
    Let $\mathbf{X}=(X_1, \ldots, X_r) \sim N_r(0, \Sigma)$ be an $r$-variate centered Gaussian random vector, and let $f: \mathbb{R}^r \rightarrow \mathbb{R}$ be a polynomial of total degree $s$. If $0 < \beta \leq 2/s$, then 
    \[
        \|f(\mathbf{X})]\|_{\psi_{\beta}} \leq C_{\beta, s} \, \|f(\mathbf{X})\|_2,
    \]
    where $C_{\beta, s}$ is a constant depending only on $\beta$ and $s$.
\end{lem}
\begin{proof}
    By Lemma \ref{lem:weibull-moments}, we have the inequality 
    \begin{equation} \label{eq:bound-by-q-norm}
        \|f(\mathbf{X})\|_{\psi_{\beta}} \leq C_{\beta} \, \sup_{q \geq 1} \frac{\|f(\mathbf{X})\|_q}{q^{1/\beta}}.
    \end{equation}
    Now, for all $q \geq 2$, by the hypercontractivity property of polynomials in Gaussian random variables (Theorem 3.2.10 in \citeauthor{pena1999decoupling}, \citeyear{pena1999decoupling} and Lemma 2.2 in \citeauthor{leung2024singularityagnostic}, \citeyear{leung2024singularityagnostic}),
    \begin{equation}\label{eq:bound-from-hypercontr}
        \|f(\mathbf{X})\|_q \leq C_s \, q^{s/2} \, \|f(\mathbf{X})\|_2.
    \end{equation}
    Since $\|f(\mathbf{X})\|_1 \leq \|f(\mathbf{X})\|_2$, we obtain by combining \eqref{eq:bound-by-q-norm} and \eqref{eq:bound-from-hypercontr} that
    \[
        \|f(\mathbf{X})\|_{\psi_{\beta}} \leq C_{\beta, s} \, \|f(\mathbf{X})\|_2 \, \sup_{q \geq 2} q^{\frac{s}{2} - \frac{1}{\beta}}.
    \]
    To conclude the proof, we note that $s/2 \leq 1/\beta$ holds due to our assumptions, which implies that the supremum is achieved for $q=2$.
\end{proof}

Next, we cite a result from \citet{chen2019randomized} that gives a sub-Weibull bound on the maximum of multiple  random variables in relation to the bound on the individual variables. 

\begin{lem} \label{lem:maximum-weibull} 
Let $\beta > 0$  and let $X_1, \ldots, X_n$ be (possibly dependent) random variables such that $\|X_i\|_{\psi_{\beta}} < \infty$ for all $i=1, \ldots, n$. Then for $n \geq 2$
$$
\left\|\max_{1\leq i \leq n} |X_i|\right\|_{\psi_{\beta}} 
\leq   C_{\beta} \log(n)^{\frac{1}{\beta}} \max_{1 \leq i \leq n}\|X_i\|_{\psi_{\beta}},
$$
where $C_{\beta}$ is a constant depending only on $\beta$.
\end{lem}
\begin{proof}
See \citet[Lemma C.1]{chen2019randomized}.
\end{proof}

The important conclusion of Lemma~\ref{lem:maximum-weibull} is that the bound on the individual variables $\|X_i\|_{\psi_{\beta}}$ yields a bound on the rate of growth on $\|\max_{1\leq i \leq n}|X_i|\|_{\psi_{\beta}}$. It is determined by the slowly growing function $\log(n)^{1/\beta}$, i.e. at most logarithmic in $n$. 
The last Lemma is a result considering centered random variables.

\begin{lem} \label{lem:centering}
For any $\beta>0$ and any random variable $X$ with $\mathbb{E}[X] < \infty$ we have
$$
\|X-\mathbb{E}[X]\|_{\psi_{\beta}} \leq C_{\beta} \|X\|_{\psi_{\beta}},
$$
where $C_{\beta}$ is a constant depending only on $\beta$.
\end{lem}
\begin{proof}
The proof is similar to the proof of \citet[Lemma 2.6.8]{vershynin2018high} which treats the special case $\beta = 2$. Let $C_{\beta}$ be a constant only depending on $\beta$ but the value can change from place to place. For $\beta \geq 1$ recall that $\|\cdot\|_{\psi_{\beta}}$ is a norm. Thus we can use the triangle inequality and get
$$\|X-\mathbb{E}[X]\|_{\psi_{\beta}} \leq \|X\|_{\psi_{\beta}} + \|\mathbb{E}[X]\|_{\psi_{\beta}}.$$
If $0 < \beta < 1$ we use Lemma~\ref{lem:triangle-inequality} with $n=2$ and get
$$\|X-\mathbb{E}[X]\|_{\psi_{\beta}} \leq 2^{\frac{1}{\beta}} \left(\|X\|_{\psi_{\beta}} + \|\mathbb{E}[X]\|_{\psi_{\beta}}\right).$$
We only have to bound the second term of both inequalities. Note that, for any constant $a$, we trivially have $\|a\|_{\psi_{\beta}} \leq C_{\beta} |a|$ by the Definition of $\|\cdot\|_{\psi_{\beta}}$. Using this for $a=\mathbb{E}[X]$, we get
\begin{eqnarray*}
\|\mathbb{E}[X]\|_{\psi_{\beta}} &\leq& C_{\beta} |\mathbb{E}[X]| \\
&\leq& C_{\beta} \mathbb{E}[|X|] \quad (\textrm{by Jensen's inequality}) \\
&=& C_{\beta} \|X\|_1 \\
&\leq& C_{\beta} \|X\|_{\psi_{\beta}} \quad(\textrm{by Lemma~\ref{lem:weibull-moments} with } q=1).
\end{eqnarray*}
\end{proof}

\clearpage

\section{Additional Simulation Results for Gaussian Latent Tree Models} \label{sec:additional-simulations} 
\begin{figure}[ht]
\captionsetup[subfigure]{labelformat=empty}
\centering
\begin{subfigure}{0.75\linewidth}
\centering
\includegraphics[width=\linewidth]{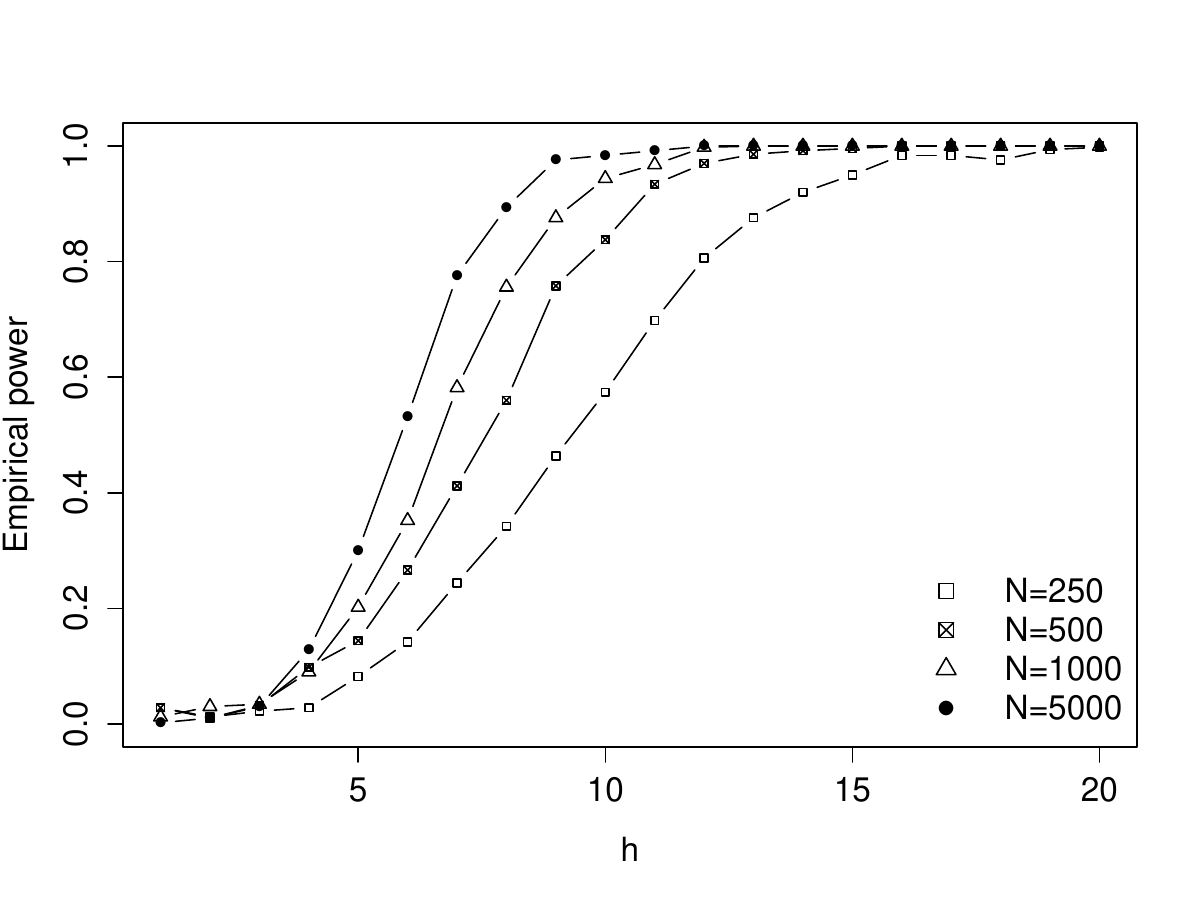}
\caption{Setup (b)}
\end{subfigure}
\hfill
\begin{subfigure}{0.75\linewidth}
\centering
\includegraphics[width=\linewidth]{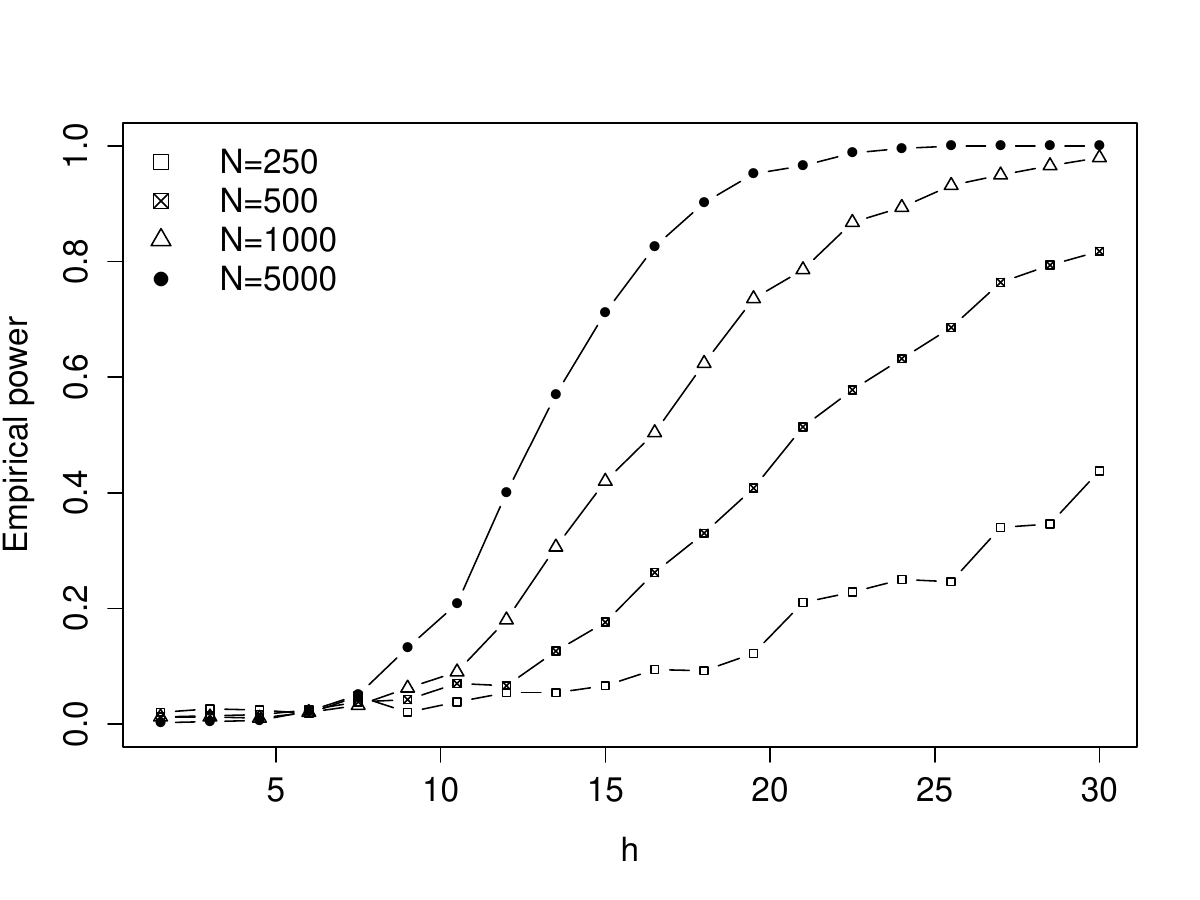}
\caption{Setup (c)}
\end{subfigure}
\caption{Empirical power for different local alternatives based on $500$ experiments. The computational budget parameters $N$ is varied as indicated. Local alternatives are generated as described in Section~\ref{sec:tree-models} for setups (b) and (c) with $(l,n)=(15,500)$ and level $\alpha=0.05$.}
\label{fig:power-setup2-3}
\end{figure}

\section{Testing Minors in Two-Factor Analysis Models} \label{sec:2-factor}
In this section, we consider testing model invariants of two-factor analysis models as another application of our testing methodology. By \citet{drton2007algebraic} we have the following parametric representation of two-factor analysis models.
\begin{prop}
The two-factor analysis model is the family of Gaussian distributions $N_l(\mu, \Sigma)$ on $\mathbb{R}^l$ whose mean vector $\mu$ is an arbitrary vector in $\mathbb{R}^l$ and whose covariance matrix $\Sigma$ is in the set
$$
F_{l,2} = \{\Lambda + \Psi \in \mathbb{R}^{l \times l}: \Lambda > 0 \textrm{ diagonal, } \Psi \geq 0 \textrm{ symmetric, } \rank(\Psi) \leq 2 \}.
$$
\end{prop}
Here the notation $B > 0$ means that $B$ is a positive definite matrix, and  $B \geq 0$ means that $B$ is positive semidefinite. Given a covariance matrix $\Sigma \in F_{l,2}$ in the two-factor analysis model, all off-diagonal $3 \times 3$ minors are vanishing \citep{drton2007algebraic}. We are interested in testing all of these model invariants simultaneously based on i.i.d.~samples $X_1, \ldots, X_n \sim N(0,\Sigma)$, where the covariance matrix $\Sigma$ in unknown. Up to sign, there are $p=10 \binom{l}{6}$ off-diagonal $(3 \times 3)$-minors, that is, the number of minors grows very fast with the dimension $l$. Since minors of $\Sigma=(\sigma_{uv})$ are polynomials in the entries $\sigma_{uv}$, we can apply the kernel proposed in Section ~\ref{sec:polynomial-hypotheses}. We consider two experimental setups:
\begin{itemize}
    \item[(\textit{Reg})] $\Psi = \Gamma \Gamma^{\top}$ with $\Gamma \in \mathbb{R}^{l \times 2}$ and all entries of $\Gamma$ are independently generated based on a standard normal distribution. All diagonal entries $\Lambda_{vv}$ are taken to be $1$.
    \item[(\textit{Irreg})] $\Psi = \beta_1 \beta_1^{\top} + \beta_2 \beta_2^{\top}$ with $\beta_1, \beta_2 \in \mathbb{R}^l$. All entries of $\beta_1$ are equal to $1$.  The first two entries of $\beta_2$, i.e., $\beta_{21}$ and $\beta_{22}$, are taken to be $10$ while all other entries of $\beta_2$ are independently generated based on a normal distribution with mean $0$ and variance $0.2$. All diagonal entries $\Lambda_{vv}$ are taken to be $1/3$. 
\end{itemize}

Setup \textit{Reg} is designed to be a regular setup while setup \textit{Irreg} is irregular. Moreover, the parameters in the irregular setup are close to an algebraic singularity \citep[Example 33]{drton2007algebraic} such that we expect that the likelihood ratio test fails to control test size.

In the implementation we use $A = 1000$ sets of Gaussian multipliers to compute the critical value $c_W(1 - \alpha)$. We apply a more general version of the divide-and-conquer bootstrap discussed in Section~\ref{subsec:bootstrap-approx}, where we use a block-diagonal sampling scheme with a block size of $L=25$; see \citet{chen2019randomized}. However, the theoretical bound for the asymptotic approximation remains the same. As before, we compare our methodology with the likelihood ratio test implemented by the \texttt{factanal} function in \texttt{R} \citep{R}. 

In Figure~\ref{fig:sizes-2-factor} we compare empirical test sizes for different, fixed nominal levels $\alpha \in (0,1)$. Once again, we see the advantage of our test methodology in the irregular setup. While the likelihood ratio test fails to control test size, our test has lower empirical size than nominal level for all computational budgets $N$, i.e., our test controls type I error, albeit conservatively.  Moreover,  we compare the empirical power in Figure~\ref{fig:power-2-factor}, where we construct the local alternatives equivalently to the alternatives considered in Section~\ref{sec:tree-models}. That is, for $\Sigma \in F_{l,2}$, the alternatives are of the form $\widetilde{\Sigma} = \Sigma + \gamma \gamma^{\top} h/\sqrt{n}$, where $\gamma = (0, \ldots, 0,1,1) \in \mathbb{R}^l$ and $h>0$ varies. As before, we observe that the empirical power is better for larger computational budgets $N$.

\begin{figure}[p]
\captionsetup[subfigure]{labelformat=empty}
\centering
\begin{subfigure}{0.8\linewidth}
\centering
\includegraphics[width=\linewidth]{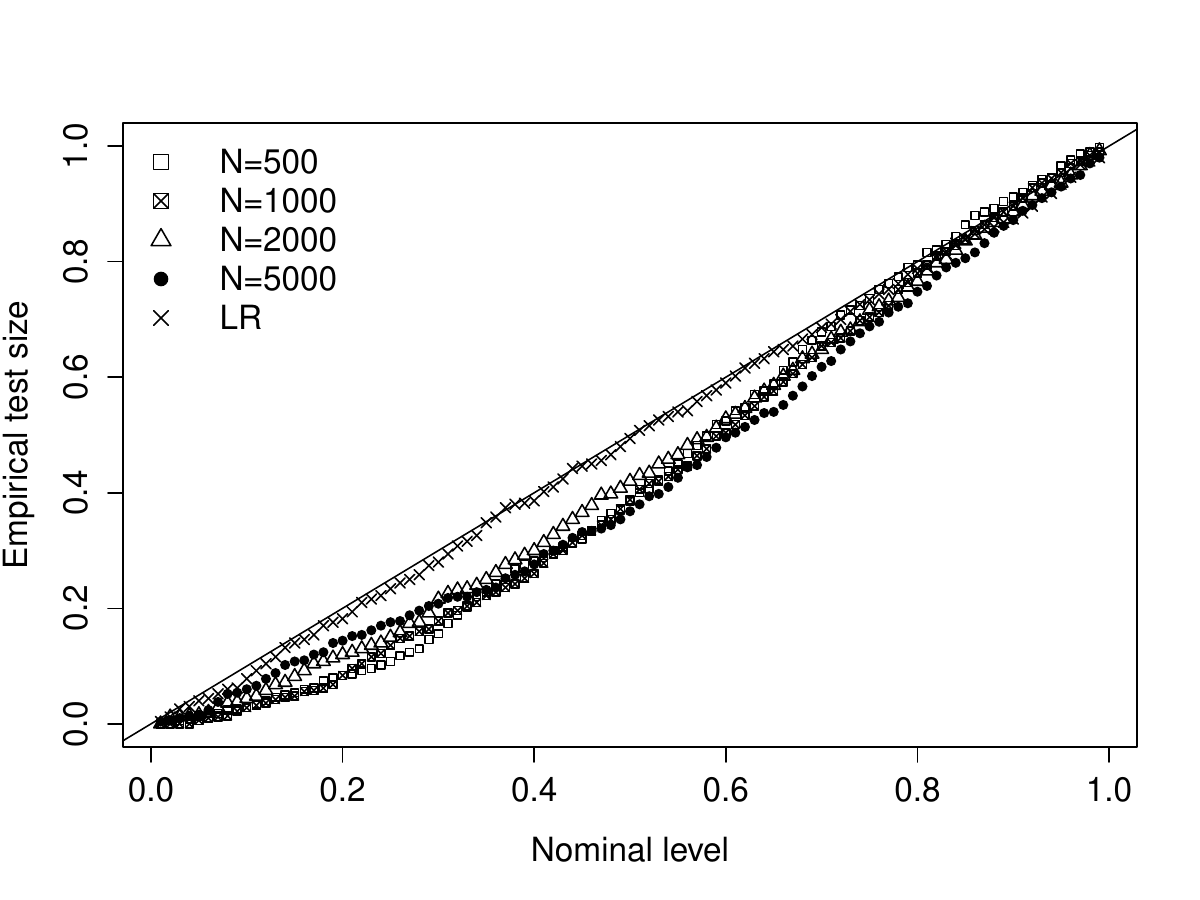}
\caption{Regular setup}
\end{subfigure}
\hfill
\begin{subfigure}{0.8\linewidth}
\centering
\includegraphics[width=\linewidth]{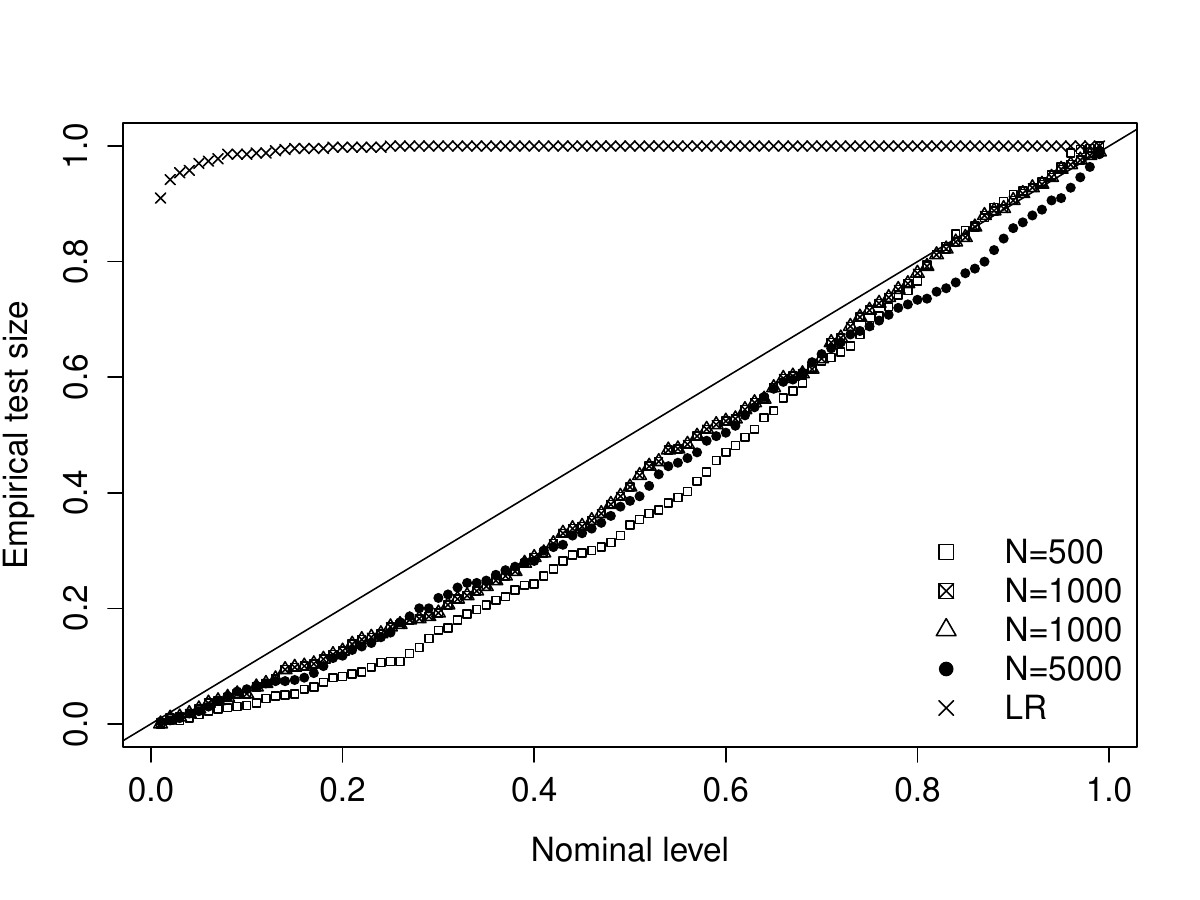}
\caption{Irregular setup}
\end{subfigure}
\caption{Empirical sizes vs. nominal levels for testing $(3 \times 3)$-minors based on $500$ experiments. The computational budget parameter $N$ is varied as indicated and empirical sizes of the LR test are also shown. Data is generated from regular and singular setups with $(l,n)=(10,500)$. }
\label{fig:sizes-2-factor}
\end{figure}

\begin{figure}[p]
\captionsetup[subfigure]{labelformat=empty}
\centering
\begin{subfigure}{0.8\linewidth}
\centering
\includegraphics[width=\linewidth]{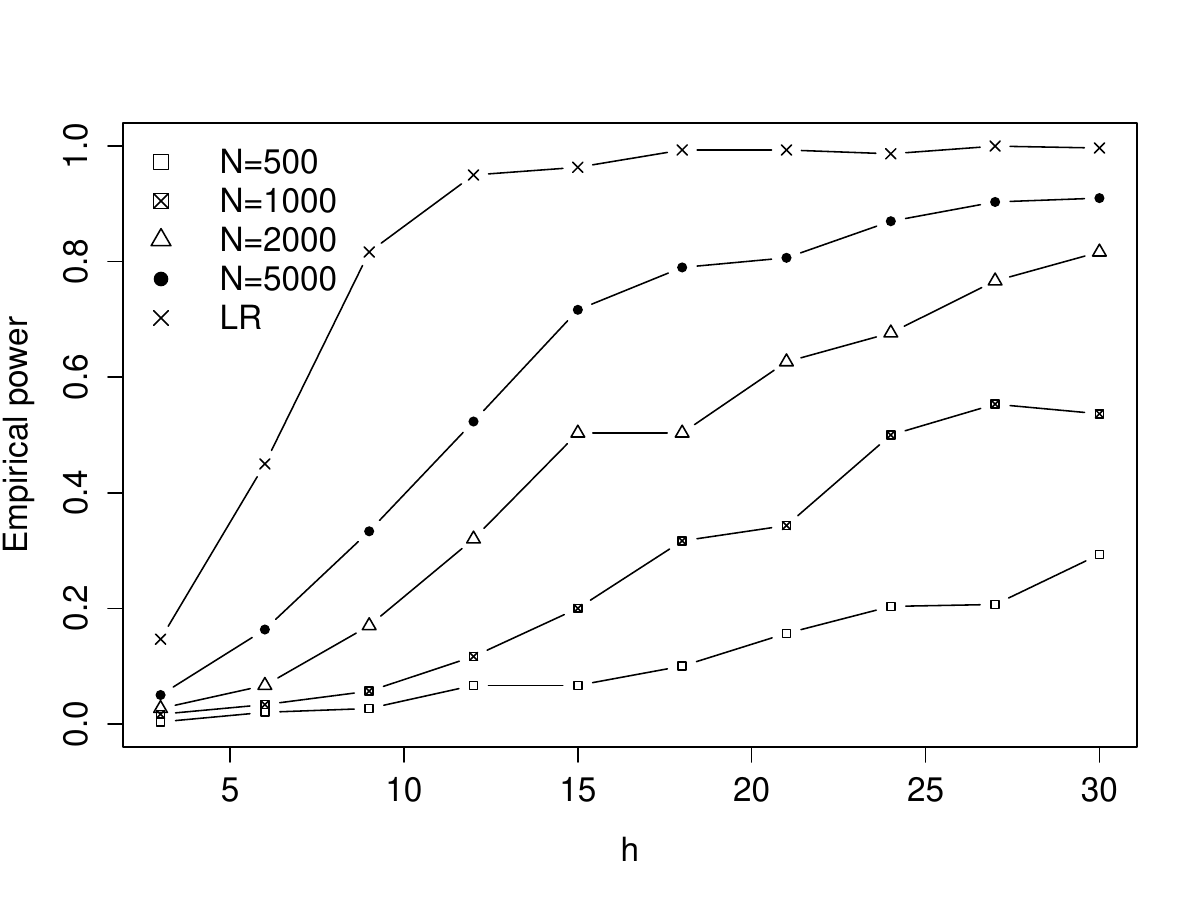}
\caption{Regular setup}
\end{subfigure}
\hfill
\begin{subfigure}{0.8\linewidth}
\centering
\includegraphics[width=\linewidth]{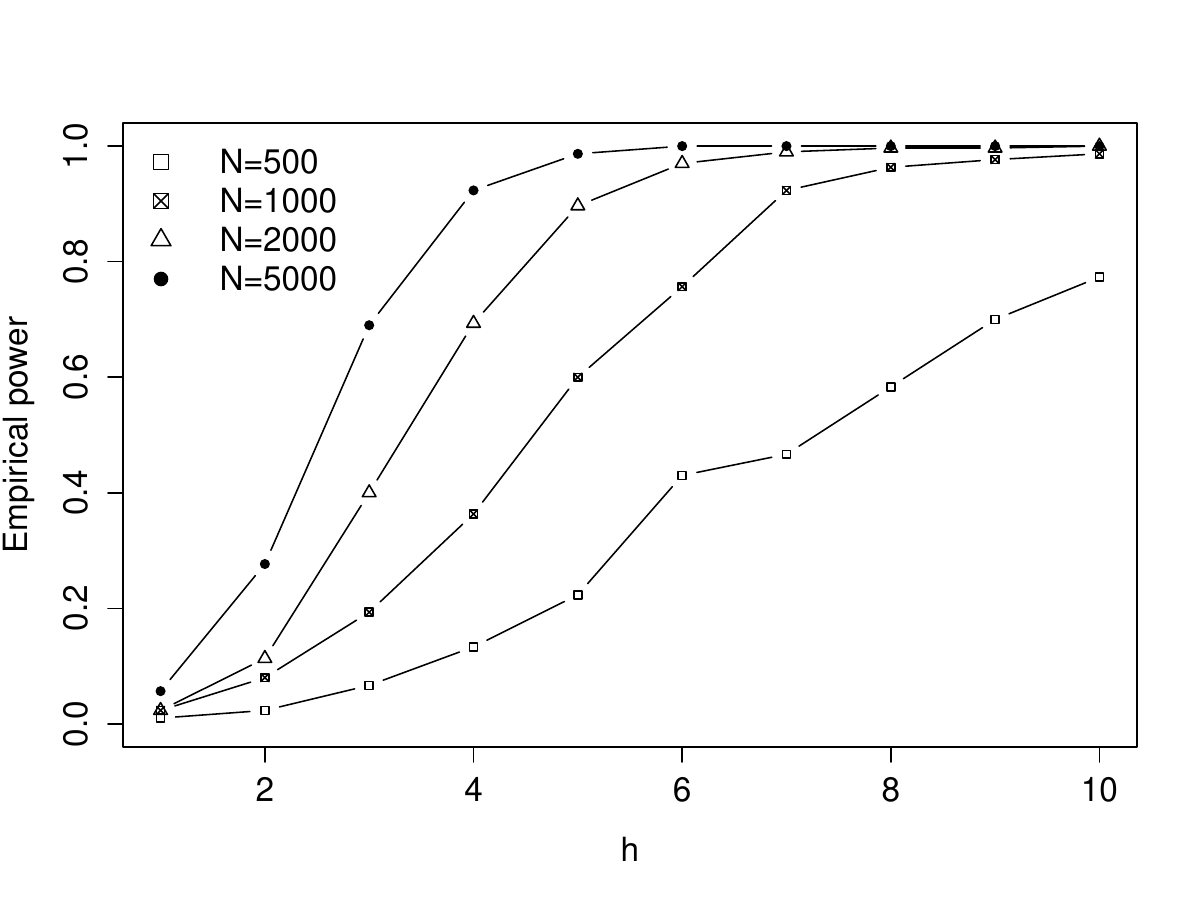}
\caption{Irregular setup}
\end{subfigure}
\caption{Empirical power in testing $(3 \times 3)$-minors for different local alternatives based on $300$ experiments. The computational budget parameters $N$ is varied as indicated. Local alternatives are generated as described in the text for regular and irregular setups with $(l,n)=(10,500)$ and level $\alpha=0.05$.}
\label{fig:power-2-factor}
\end{figure}

\end{appendix}
\end{document}